\documentclass[11pt,onecolumn]{IEEEtran}
\usepackage{amsmath,amssymb}
\usepackage{graphicx}
\usepackage{epsfig}
\usepackage{cite}
\usepackage{setspace}
\usepackage{amsthm}
\usepackage{color}
\usepackage{enumerate}
\usepackage{subfigure}
\usepackage{soul}
\usepackage{framed}

\author{Hao Feng, \emph{Student Member, IEEE} and Andreas F. Molisch, \emph{Fellow, IEEE}\\
    University of Southern California, Los Angeles, CA\\
    Email: \{ haofeng, molisch \} @usc.edu}

\newtheorem{thm}{Theorem}

\newtheorem{lemma}{Lemma}
\newtheorem{corollary}{Corollary}

\begin{document}
\definecolor{shadecolor}{rgb}{1.00,1.00,0.00}
\setstcolor{red}
\sethlcolor{yellow}

\title{Diversity Backpressure Scheduling and Routing with Mutual Information Accumulation in Wireless Ad-hoc Networks}
\maketitle
\pagestyle{plain}
\thispagestyle{plain}

\begin{abstract}
We suggest and analyze algorithms for routing in multi-hop wireless ad-hoc networks that exploit mutual information accumulation as the physical layer transmission scheme, and are capable of routing multiple packet streams (commodities) when only the \emph{average} channel state information is present, and that only locally. The proposed algorithms are modifications of the \emph{Diversity Backpressure} (DIVBAR) algorithm, under which the packet whose commodity has the largest "backpressure metric" is chosen to be transmitted and is forwarded through the link with the largest differential backlog (queue length). In contrast to traditional DIVBAR, each receiving node stores and accumulates the partially received packet in a separate "partial packet queue", thus increasing the probability of successful reception during a later possible retransmission. We present two variants of the algorithm: DIVBAR-RMIA, under which all the receiving nodes clear the received partial information of a packet once one or more receiving nodes firstly decode the packet; and DIVBAR-MIA, under which all the receiving nodes retain the partial information of a packet until the packet has reached its destination. We characterize the network capacity region with RMIA and prove that (under certain mild conditions) it is strictly larger than the network capacity region with the repetition (REP) transmission scheme that is used by the traditional DIVBAR. We also prove that DIVBAR-RMIA is throughput-optimum among the polices with RMIA, i.e., it achieves the network capacity region with RMIA, which in turn demonstrates that DIVBAR-RMIA outperforms traditional DIVBAR on the achievable throughput. Moreover, we prove that DIVBAR-MIA performs at least as well as DIVBAR-RMIA with respect to throughput. Simulations also confirm these results.
\end{abstract}

\section{Introduction}
\label{sec: introduction}

{\let\thefootnote\relax\footnote{Part of this work was presented at ICC 2012. The work was financially supported by NSF under Grant 0964479.}}

Wireless multi-hop ad-hoc networks have drawn significant attention in recent years, due to their flexibility and low cost, and their resulting importance in factory automation, sensor networks, security systems, and many other applications. A fundamental problem in such networks is the routing of data packets, i.e., which nodes should transmit which packets in which sequence. Optimum routing to minimize the delay of a single packet flowing through a deterministic multi-hop network has been solved by some classic algorithms, e.g., Dijkstra and Bellman-Ford (see Ref. \cite{Molisch_2011_book}, Chapter 22 and references therein). Based on these algorithms, a straightforward way of routing packets in multi-hop wireless network is to adopt some existing routing methods used in deterministic network with some ad-hoc modifications. The \emph{source routing} protocols, e.g., DSR \cite{Johnson_et_al_2011} and LQSR \cite{Draves_et_al_2004}, and \emph{distance vector routing} protocols, e.g., AODV \cite{Perkins_Royer_1999} and DSDV \cite{Perkins_Bhagwat_1994}, are existing approaches belonging to this category. These approaches predetermine the routing path for each source-destination pair before the actual transmissions, which is based on the assumption that the network is static during the whole delivery process of each packet from the source to destination, and the channel realization on each link is the same as probed beforehand. Unfortunately, however, these assumptions are not desirable and/or possible in many scenarios of multi-hop wireless ad-hoc networks.

On the other hand, the throughput performance becomes an issue when a single stream of packets flows through a network; this has been well-explored by several approaches, such as Ford-Fulkerson algorithm and Preflow-Push algorithm (see Ref. \cite{Kleinberg_Tardos_book_2005}, Chapter 7 and references therein), and Goldberg-Rao algorithm (see \cite{Goldberg_Rao_1998} and references therein), in the case of wired networks. Nevertheless, simultaneous routing of multiple packet streams intended for multiple destinations (i.e., multiple \emph{commodities}) is much more difficult, as different commodities are competing for the limited network resources.

To deal with these issues, several studies focus on the routing in the wireless network with unreliable channels and possible multiple commodities. The ExOR algorithm \cite{ExOR_Biswas_2005} takes advantage of the \emph{broadcast effect}, i.e., the packet being transmitted by a node can be overheard by multiple receiving nodes. After confirming the successful receivers among all the potential receiving nodes after each attempt of transmission, the transmitting node decides the best node among the successful receivers to forward the packet in the future according to the \emph{Expected Transmission Count Metric} (ETX) \cite{ETX_Couto_2003}, which indicates the proximity from each receiving node to the destination node in terms of \emph{forward delivery probability}. As a further improvement, the proactive SOAR algorithm \cite{SOAR_Rozner_2009} also uses ETX as the underling routing metric but leverages the path diversity by certain \emph{adaptive forwarding path selections}. Both ExOR and SOAR have shown better throughput performance than the traditional routing methods, but neither theoretically provides a \emph{throughput-optimum} routing approach for multi-hop, multi-commodity wireless ad-hoc networks.

Throughput maximization can be tackled by \emph{stochastic network optimization}, which involves routing, scheduling and resource allocation in networks without reliable or precisely predictable links but with certain stochastic features. Refs. \cite{Georgiades_et_al_2006}, \cite{Neely_book_2010} systematically analyze this kind of problems by using \emph{Lyapunov drift} analysis originating from control theory, which follows and generalizes \emph{Backpressure} algorithm proposed in \cite{Tassiulas_Emphremides_1992} \cite{Tassiulas_Emphremides_1993}. The backpressure algorithm establishes a \emph{Max-weight-matching} metric for each commodity on each available link that takes into account the local differential \emph{backlogs} (queue lengths or the number of packets of the particular commodity at a node) as well as the channel state of the corresponding link observed in time. The packet of the commodity with the largest metric will be transmitted from each node. Thus, the backpressure algorithm achieves routing without ever designing an explicit route and without requiring centralized information, and therefore, is considered as a very promising approach to stochastic network optimization problems with multiple commodities. The idea of Backpressure routing was later extended to many other communication applications, e.g., power and server allocation in satellite downlink \cite{Neely_satellite_2002}, routing and power allocation in time-varying wireless networks \cite{Neely_2005}, and throughput optimal routing in cooperative two hop parallel relay networks \cite{Yeh_Barry_2007}.

Based on the principle of Backpressure algorithm, \cite{Neely_Rahul_DIVBAR_2009} developed the \emph{Diversity Backpressure} (DIVBAR) algorithm for the routing in multi-hop, multi-commodity wireless ad-hoc networks. Similar to ExOR and SOAR, DIVBAR assumes a network with no reliable or precisely predictable channel state and exploits the broadcast nature of the wireless medium. In general, each node under DIVBAR locally uses the backpressure concept to route packets in the direction of maximum differential backlog. Specifically, each transmitting node under DIVBAR chooses the packet with the optimal commodity to transmit by computing the Max-weight-matching metric, whose factors include the observed differential backlogs and the link success probabilities resulting from the fading channels; after getting the feedbacks from all the potential receiving nodes indicating the successful receptions, the transmitting node let the successful recipient with the largest positive differential backlog get the forwarding responsibility. The superiority of DIVBAR over ExOR or SOAR is that DIVBAR has been theoretically shown to be throughput-optimum in wireless ad-hoc networks subject to the similar assumptions as those of ExOR and SOAR, e.g., unreliable links, no complete channel state information, broadcast effect, and most notably, the assumption that any packet not correctly received by any potential receiving node needs to be completely retransmitted in the future transmission attempts. Here we call the scheme of complete retransmission the \emph{Repetition Transmission Scheme} (REP).

The efficiency of REP can be greatly enhanced by mutual-information accumulation (MIA), where the receiving nodes store \emph{partial information} of the packets that cannot be decoded at the previous transmission attempts. MIA is implemented by using Fountain Codes (or rateless codes), which were introduced by Luby and coworkers in Ref. \cite{Luby_2002}\cite{Byers_et_al_2002}\cite{Sokrollahi_2004}. The transmitter with Fountain codes encodes and transmits the source information in infinitely long code streams, and the receiver can recover the original source information from the portions of the code streams received in an unordered manner, as long as the amount of total accumulated information exceeds the entropy of the source information. Moreover, Fountain codes can work at any SNR, and therefore, the same code design can be used for broadcasting from one transmitter to multiple receivers whose links to the transmitter have different channel gains. At the same time, Fountain codes can even accumulate the partial information from multiple transmitters. Fountain codes have been suggested for various applications, e.g., point-to-point communications with quasi-static and block fading channels \cite{Castura_et_al_point2point_2006}, cooperative communications in single relay networks \cite{Castura_Mao_2007}\cite{Liu_2006}, cooperative communications in two hop multi-relay networks \cite{Molisch_et_al_2007}, incremental redundancy Hybrid-ARQ protocols used for Gaussian collision channel in non-routing settings\cite{Caire_Tuninetti_2001}, and incremental redundancy Hybrid-ARQ protocols used in the downlink scheduling of the MU-MIMO system \cite{Hooman_Caire_et_al_2011}. In these applications, Fountain codes have been shown to enhance robustness, save energy, reduce transmission time and increase throughput. Refs. \cite{Draper_et_al_2011}, \cite{Rahul_Neely_2011} introduce MIA into the routing of multi-hop ad-hoc networks, and have shown that the delay performance can be enhanced with constraint power and bandwidth resources. However, none of above papers touches the throughput performance of ad-hoc networks with MIA.

For multi-hop, multi-commodity wireless ad-hoc networks, the throughput potential might also be increased when implementing MIA instead of REP. An intuitive approach of exploring this problem is to combine MIA with Lyapunov drift analysis, and design a "MIA version" of Backpressure or DIVBAR algorithm. Following this strategy and parallel to our work, Ref. \cite{Yang_et_al_2012} proposed a T-slot routing algorithm and a virtual queue routing algorithm for multi-hop, multi-commodity wireless ad-hoc network with broadcast effect. These two algorithms assume that each link in the network has fixed and reliable transmission rate, and each transmitting node making local decisions can predetermine the local transmitting and forwarding realizations based on the backlog and virtual queue observations.

In this paper, in contrast with Ref. \cite{Yang_et_al_2012}, we explore the multi-hop, multi-commodity routing in the case of unreliable and non-precisely predictable rates. We assume that the network has stationary channel fading, i.e., the distribution of the channel realization remains the same, however, the particular realization changes with time; although no precise channel state information at the transmitter (CSIT) is available, the distributions of the channel realization of each link can be obtained by the transmitter beforehand, i.e., each transmitting node has the average CSIT; the transmitting node can obtain the receiving (decoding) results of all the receiving nodes by some simple feedbacks sent by the receiving nodes through (certain reliable) control channels.

Our contributions of this paper are summarized as follows:
\begin{itemize}
\item We analyze the throughput potential of networks by characterizing the \emph{network capacity region} \cite{Georgiades_et_al_2006}\cite{Neely_book_2010} with \emph{Renewal Mutual Information Accumulation} (RMIA) transmission scheme. Here "Renewal" stands for a clearing operation; and RMIA is the transmission scheme, in which all the receiving nodes accumulate the partial information of a certain packet and try to decode the packet when receiving it, but clear the partial information of a packet every time the corresponding packet is firstly decoded by one or more receiving nodes in the network; note that the receiving node needs not to be the destination of the packet.
\item We prove that the network capacity region with RMIA is strictly larger than the network capacity region with REP (in Ref. \cite{Neely_Rahul_DIVBAR_2009}) under some mild assumptions, which indicates that MIA technique can increase the throughput potential of a network.
\item We propose and analyze two new routing algorithms that combine the concept of DIVBAR with MIA. The first version, DIVBAR-RMIA is implemented with the renewal operation, and is shown to be throughput-optimum among all possible routing algorithms with RMIA transmission scheme. Under the second version, DIVBAR-MIA, all received partial information of a packet remains stored at all the nodes in the network until that packet has reached its destination. We prove that DIVBAR-MIA's throughput performance is at least as good as DIVBAR-RMIA. In sum, both proposed algorithms can achieve larger throughput limits than the original DIVBAR algorithm with REP.
\end{itemize}

The remainder of the paper is organized as follows: Section \ref{sec: network_model} presents the network model and describes the implementation of routing with MIA technique, the timing diagram within one timeslot, and the queuing dynamics. Section \ref{sec: network_capacity_region_RMIA} characterizes the network capacity region with RMIA and compares it with the network capacity region with REP. Section \ref{sec: algorithm_discription} describes the two proposed algorithms: DIVBAR-RMIA and DIVBAR-MIA. Section \ref{sec: performance_anlaysis_main} proves the throughput optimality of DIVBAR-RMIA with RMIA assumption and proves the throughput performance guarantee of DIVBAR-MIA. Section \ref{sec: simulations} presents the simulation results. Section \ref{sec: conclusions} concludes the paper. Mathematical details of the proofs are relegated to Appendices.

\section{Network Model}
\label{sec: network_model}
Consider a stationary wireless ad-hoc network with $N$ nodes, denoted as set $\cal N$, where multiple packet streams indexed as $c=1,\cdots N$ are transmitted, possibly via multi-hop. Categorize all packets in the packet stream destined for a particular node $c$ as \emph{commodity $c$ packets} irrespective of their origin. Each link in the network is denoted by an ordered pair $\left(n,k\right)$, for $n,k \in \cal N$, where $n$ is the transmitting node and $k$ is the receiving node. Exogenous input data arrives randomly to the network in units of packets, all of which have the same fixed amount of information (entropy) denoted as $H_0$.  Packets arriving at each node are stored in a queue waiting to be forwarded, except at the destination, where they leave the network immediately upon arrival/decoding. The transmission power of each node is constant.

Time is slotted and normalized into integer units $\tau=0,1,2,3,\cdots$. The timeslot length is assumed to be equal to the coherence time of the channels, so that we can adopt the common block-fading model: within a timeslot duration, instantaneous channel gains are constant, while they are i.i.d. (independent and identically distributed) across timeslots, for each link. Average channel state information (CSI) of each link is known locally, i.e., at the node from which the link is emanating; however, instantaneous CSI (i.e., channel gains for a specific timeslot) are never known at any transmitting node. The exogenous packet arrival rate $a_n^{\left( c \right)}\left( \tau \right)$ is i.i.d. across timeslots and is upper bounded by a constant value ${A_{\max }}$. When a packet is transmitted by a node $n$ in each timeslot, it can be simultaneously overheard by multiple neighbor nodes ("multi-cast effect") represented by set ${\cal K}_n$. In this network model, a transmission of a packet over a link $\left(n,k\right)$ can be interpreted as a process, in which a new copy of the packet is being created in the the receiving node $k$ while the original copy of packet is retained in the transmitting node $n$, and correspondingly, multi-cast effect indicates that multiple copies of the same packet can be created at multiple receiving nodes simultaneously. However, in this case, at most one successful receiving node is finally allowed to get the responsibility of forwarding the packet in the future and keep the received copy of packet, i.e., if defining $b_{nk}^{\left(c\right)}\left(\tau\right)$ as the number of packets of commodity $c$ that flow from node $n$ to node $k \in {\cal K}_n$ in timeslot $\tau$, then $b_{nk}^{\left(c\right)}\left(\tau\right) \in \left\{0,1\right\}$ and $\sum_{c \in \cal N} {\sum_{k \in {{\cal K}_n}} {b_{nk}^{\left( c \right)}\left( \tau \right)} }  \le 1$, for $\forall n\in \cal N$. Here, the \emph{flow rate} $b_{nk}^{\left(c\right)}\left(\tau\right)$ depends on three factors: first, the decision of choosing commodity $c$ to transmit; second, the success of the reception over link $\left(n,k\right)$; third, the decision of assigning the forwarding responsibility. After each forwarding decision is made among the nodes having a complete copy of the packet (including the transmitting node and receiving nodes), only the node that gets the forwarding responsibility (possibly being retained by the transmitting node after making the forwarding decision) can keep the packet, while others discard their copies.

Based on this network model, our goal is to design a routing algorithm that can support an exogenous input rate as large as possible, while subject to a possible tradeoff with delay.

\subsection{Routing with Mutual Information Accumulation Technique}
\label{subsec: routing_with_MIA}

Ref. \cite{Neely_Rahul_DIVBAR_2009} analyzes the routing algorithms implemented based on REP, i.e., for each transmission, the packet either is successfully received at another node, or has to be completely re-transmitted in a later timeslot. As has been described in Section \ref{sec: introduction}, we suggest to avoid the inefficiencies of complete retransmission by enabling the Mutual Information Accumulation (MIA) technique into the transmission scheme with the help of using Fountain codes.

In our scenario, we assume that each link uses a capacity-achieving coding scheme, so that a packet is received correctly in timeslot $\tau$ if the amount of partial information of the packet received by the end of timeslot $\tau$ exceeds the entropy of the packet $H_0$, i.e., a successful transmission from node $n$ to node $k$ in timeslot $\tau$ occurs when $\log_2 \left(1+ \gamma_{nk}\left(\tau\right)\right)+I_k\left(\tau\right)\geq H_0$, where $\gamma\left(\tau\right)$ is the SNR in timeslot $\tau$, whose distribution depends on the average channel state of link $\left(n,k\right)$; $I_k\left(\tau\right)$ is the amount of partial information of the corresponding packet already accumulated in the receiving node $k$ by timeslot $\tau-1$. Moreover, despite that each receiving node may simultaneously overhear the signals transmitted from multiple neighbor nodes, we assume that there is no inter-channel interference among these signals and the successful receiving of each signal is independent of the signals transmitted through other links.\footnote[1]{While this assumption is not practically realizable in wireless scenarios unless we use orthogonal channels, it is a standard assumption in the literature of stochastic network optimizations for wireless networks \cite{Neely_2005}\cite{Neely_Rahul_DIVBAR_2009}.}

As will be shown in Section \ref{sec: network_capacity_region_RMIA}, the throughput potential of the network can be increased by adopting MIA technique instead of REP in the transmissions, since MIA essentially increases the success probability of the transmissions over each link. Here the intuition is: with MIA technique, each transmission might not have to transmit a whole packet in each timeslot but take advantage of the information already accumulated at the receiving node. Only the differential amount of information is required for the decoding at the receiver. This will increase the success probability of each transmission, and therefore increases the average transmission rate.

To implement routing with MIA technique, each node has set up two kinds of queues: the \textit{compact packet queue} (CPQ) and \textit{partial packet queue} (PPQ). CPQs store the packets that have already been decoded and are categorized by packets' commodities; while the pieces of partial information stored in PPQ are distinguished by the packets they belong to. As soon as the partial information of a specific packet accumulated in PPQ of current node exceeds the entropy of that packet, the packet is decoded and moved out of PPQ, and then put into CPQ if the current node gets the forwarding responsibility from the transmitting node.

In this paper, we propose two versions of a transmission scheme based on the MIA technique: Renewal Mutual Information Accumulation (RMIA) transmission and Mutual Information Accumulation (MIA) transmission. With RMIA transmission scheme and multicast effect, each receiving node accumulates the partial information of the packet in the timeslots when the packet gets the transmission opportunity; as soon as one or more receiving node firstly decode the packet, the transmission of the packet from node $n$ stops, and all the partial information of the packet already accumulated in the receiving nodes is cleared, which is called \emph{renewal operation}. The timeslot when the first decoding occurs is called the \emph{first decoding timeslot}; the set of successful receiving nodes is called the \emph{first successful receiver set}. In contrast, with MIA transmission scheme, instead of clearing the partial information after the corresponding packet is decoded, each receiving node retains the partial information of the packet and possibly uses it in the future decoding, when the receiving node overhears another copy of the same packet in later transmissions. The partial information is only cleared when the packet has reached its destination.

For each transmitting node in the network under an arbitrary routing policy with MIA or RMIA, define the timeslots that are used to transmit the same copy of packet from one node as one \textit{epoch}. During each epoch, each receiving node keeps accumulating the partial information of a particular packet until the transmitting node stops transmitting the copy of this packet, which is the end of the epoch. Note that the timeslots making up one epoch of a particular packet may not be contiguous timeslots. Thus, under an arbitrary policy with RMIA or MIA, the timeslots' allocation for different commodities and epoch allocation is shown as Fig. \ref{fig_slots_allocation}, in which the notation ${\rm{Epoch}}_{n,i}^{\left(c\right)}$ denotes the $i$th epoch that is used to transmit commodity $c$ by node $n$; in the silent timeslots, node $n$ does not transmit commodity packets. Then the difference between RMIA and MIA can be expressed as: with RMIA, each receiving node $k \in {\cal K}_n$ implements a renewal operation at the end of each epoch for each transmitting node $n$; while with MIA, the partial information is retained at the end of each epoch for each transmitting node and will possibly facilitate the decoding in later epochs, and the partial information of a particular packet is cleared only when this packet is delivered to its destination.

\begin{figure}
        \includegraphics[height=4.3cm]{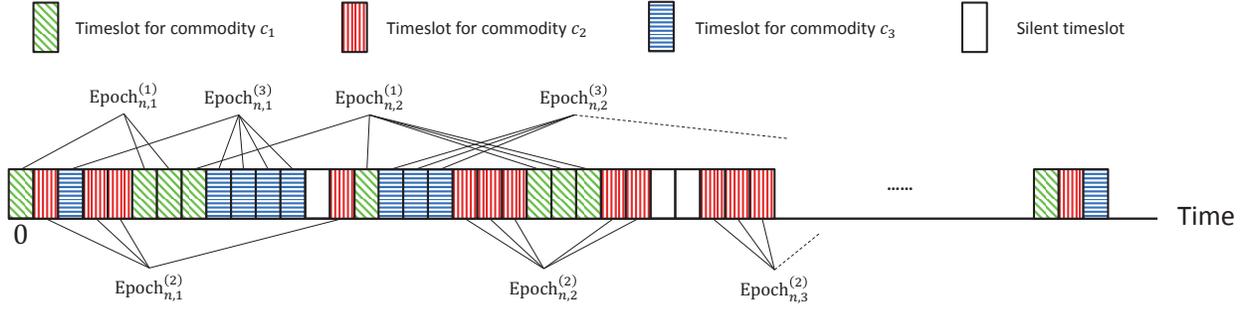}
        \centering{
        \caption{Slot allocation over different commodities and formation of epochs for node $n$ under an arbitrary policy}
        \label{fig_slots_allocation}}
\end{figure}

\subsection{Timing Diagram in One Timeslot and Queuing Dynamics}
\label{subsec: queuing_dynamic_time_diagram}
\begin{figure}
        \includegraphics[height=3.5cm]{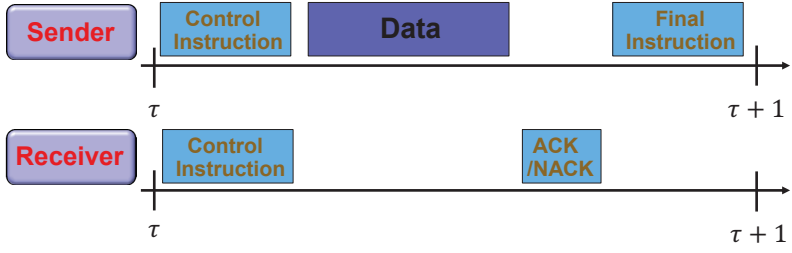}
        \centering{
        \caption{Timing diagram of the working protocol within one timeslot}
        \label{fig_timing_diagram}}
\end{figure}

The timing diagram of the communication protocol between each pair of sending node and receiving node within one timeslot is illustrated in Fig. \ref{fig_timing_diagram}. As is shown in this figure, at the beginning of each timeslot $\tau$, the transmitting node (sender) and receiving node (receiver) exchange their control instructions which include the backlog information of CPQs. Then the transmitting node makes decisions about which commodity to transmit or whether to keep silent in timeslot $\tau$. After the decision is made, the data transmission starts and lasts for a fixed time interval (less than the timeslot length), during which coded bits based on a packet with entropy $H_0$ are being transmitted. After the data transmission period ends, each receiving node sends an ACK/NACK signal back to the transmitting node indicating if the packet is successfully decoded by the receiving node or not (as will be shown in Section \ref{sec: algorithm_discription}, each receiving node under the proposed DIVBAR-MIA algorithm sends two kinds of ACK/NACK signals back to the transmitting node). The transmitting node gathers all the ACK/NACKs from all the receiving nodes and, based on the feedbacks, makes the decision on which successful receiver should get the forwarding responsibility. Finally, the forwarding decision is multicast through a final instruction signal to all the receiving nodes, and correspondingly all the successful receivers, except the one getting the forwarding responsibility, abandon their decoded copies. Note that the final step of the working protocol assumes that each packet always has a single copy being routed in the network.

The queueing dynamics over each timeslot is based on the above timing diagram. Let $Q_n^{\left(c\right)}\left(\tau\right)$ represent the backlog of the CPQ of commodity $c$ in node $n$ at the beginning of timeslot $\tau$. The backlog of each commodity in each node is updated over each timeslot as follows:
\begin{equation}
Q_n^{\left( c \right)}\left( {\tau  + 1} \right) \le \max \left\{ {Q_n^{\left( c \right)}\left( \tau  \right) - \sum\limits_{k \in {{\cal K}_n}} {b_{nk}^{\left( c \right)}} \left( \tau  \right),0} \right\} + \sum\limits_{k \in {{\cal K}_n}} {b_{kn}^{\left( c \right)}} \left( \tau  \right) + a_n^{\left( c \right)}\left( \tau  \right),\ {\rm{for}}\ \forall n\in \cal N,
\label{eq_basic_queuing_dynamics}
\end{equation}
where the term $\sum_{k \in {{\cal K}_n}} {b_{nk}^{\left( c \right)}} \left( \tau  \right)$ represents the output rate; $\sum_{k \in {{\cal K}_n}} {b_{kn}^{\left( c \right)}} \left( \tau  \right)$ represents the endogenous input rate flowing from the neighbor nodes; $a_n^{\left(c\right)}\left(\tau\right)$ is the exogenous input rate of commodity $c$ arriving at node $n$ in timeslot $\tau$. The expression in (\ref{eq_basic_queuing_dynamics}) is an inequality instead of an equality because the actual valid endogenous input rate may be less than $\sum_{k \in {{\cal K}_n}} {b_{kn}^{\left( c \right)}} \left( \tau  \right)$. This occurs if a neighbor node $k\in {\cal K}_n$ has no valid packet of commodity $c$ to send (its CPQ of commodity $c$ is empty), but the decision made on node $k$ is to send a commodity $c$ packet, so node $k$ just sends a null packet that is counted into $\sum_{k \in {{\cal K}_n}} {b_{kn}^{\left( c \right)}} \left( \tau  \right)$ but is not counted in the CPQ backlog. On the other hand, for the output rate $\sum_{k \in {{\cal K}_n}} {b_{nk}^{\left( c \right)}} \left( \tau  \right)$, this issue is solved by the $\max\left\{\cdot\right\}$ operation, which guarantees that the backlog $Q_n^{\left( c \right)}\left( {\tau  + 1} \right)$ can never falls below zero. Here we should note that the value of $b_{nk}^{\left(c\right)}\left(\tau\right)$ is determined after the forwarding decision is made by node $n$ based on ACK/NACK feedbacks from all the receiving nodes in ${\cal K}_n$.

\section{Network Capacity Region with Renewal Mutual Information Accumulation}
\label{sec: network_capacity_region_RMIA}
In this section, we characterize the throughput potential of a stationary wireless network. Following from the definitions in Ref. \cite{Neely_Rahul_DIVBAR_2009}, for a multi-hop, multi-commodity network, let $\left(\lambda_n^{\left(c\right)}\right)$ represent the matrix of exogenous time average input rates, where each entry $\lambda_n^{\left(c\right)}$ represents the exogenous time average input rate of commodity $c$ entering source node $n$, in units of packet/slot. Let $Y_n^{\left(c\right)}\left(t\right)$ represent the number of packets with source node $n$ that have been successfully delivered to destination node $c$ within the first $t$ timeslots. Then a routing algorithm is \emph{rate stable} if
\begin{equation}
\mathop {\lim }\limits_{t \to \infty } \frac{{Y_n^{\left( c \right)}\left( t \right)}}{t} = \lambda _n^{\left( c \right)},\ {\rm{for}}\ n,c\in \cal N.
\end{equation}

With the above definitions, Ref. \cite{Neely_Rahul_DIVBAR_2009} defines the \emph{network capacity region} as the set of all exogenous input rate matrices $\left(\lambda_n^{\left(c\right)}\right)$ that can be stably supported by the network using certain rate stable routing algorithms. In this paper, we inherit the concept of network capacity region and use it to describe the throughput potential of a network. However, considering the effect of transmission scheme on the throughput performance, we also specify the network capacity region with different transmission schemes. For example, all the algorithms discussed in Ref. \cite{Neely_Rahul_DIVBAR_2009} are based on repetition transmission scheme (REP) assumption, and therefore, we specify the network capacity region defined in Ref. \cite{Neely_Rahul_DIVBAR_2009} as \emph{network capacity region with REP} (or REP network capacity region) in this paper, denoted as $\Lambda_{\rm{REP}}$. In our work, we define \emph{network capacity region with RMIA} (or RMIA network capacity region), denoted as $\Lambda_{\rm{RMIA}}$, as the set of all exogenous input rate matrices that can be stably supported by the network using certain rate stable routing algorithms with RMIA transmission scheme.

To prepare for the later proof of the main theorems, we firstly derive some stochastic properties for the amount of information transmitted through an arbitrary link $\left(n,k\right)$ in each timeslot, denoted as $R_{nk}\left(\tau\right)$, which is continuously distributed over $\left[0,\infty\right)$ and is i.i.d. across timeslots. Let $F_{R_{nk}}\left(x\right)$ and $f_{R_{nk}}\left(x\right)$ respectively represent the cdf (cumulative distribution function) and pdf (probability density function) of $R_{nk}\left(\tau\right)$. Let $F_{R_{nk}}^{\left(m\right)}\left(x\right)$ represent the cdf of $\sum_{\tau  = 1}^m {{R_{nk}}\left( \tau  \right)}$, where $F_{R_{nk}}^{\left(1\right)}\left(x\right) = F_{R_{nk}}\left(x\right)$; let $F_{R_{nk}}^{\left(0\right)}\left(x\right)=1$. Now we claim a lemma, whose results will be used in the proofs of later theorems:
\begin{lemma}
\label{lemma: flowing_rate_property}
If $0< F_{R_{nk}}\left(H_0\right)<1$, and $R_{nk}\left(\tau\right)$ is continuously distributed on $\left[0,\infty\right)$, then we have the following relations:
\begin{equation}
F_{{R_{nk}}}^{\left( m \right)}\left( {{H_0}} \right) < F_{{R_{nk}}}^{\left( {m - 1} \right)}\left( {{H_0}} \right){F_{{R_{nk}}}}\left( {{H_0}} \right),\ {\rm{for}}\ m \ge 2;
\label{eq_cdf_information_perslot1}
\end{equation}
\begin{equation}
F_{{R_{nk}}}^{\left( m \right)}\left( {{H_0}} \right) < {\left[ {{F_{{R_{nk}}}}\left( {{H_0}} \right)} \right]^m},{\rm{\ for\ }} m \ge 2;
\label{eq_cdf_information_perslot2}
\end{equation}
\begin{equation}
F_{{R_{nk}}}^{\left( m \right)}\left( {{H_0}} \right) < F_{{R_{nk}}}^{\left( {m'} \right)}\left( {{H_0}} \right), {\ \rm{for}}\ m > m'\ge0.
\label{eq_cdf_information_perslot3}
\end{equation}
\end{lemma}
\begin{proof}
If $m\ge 2$, we have
\begin{align}
F_{{R_{nk}}}^{\left( m \right)}\left( {{H_0}} \right) &= \Pr \left\{ {\sum\limits_{\tau  =   1}^{m - 1} {{R_{nk}}\left( \tau  \right)}  + {R_{nk}}\left( m \right) < {H_0}} \right\}\nonumber\\
&= \int\limits_0^{{H_0}} {F_{{R_{nk}}}^{\left( {m - 1} \right)}\left( {{H_0} - x} \right){f_{{R_{nk}}}}\left( x \right)dx} \nonumber\\
& < F_{{R_{nk}}}^{\left( {m - 1} \right)}\left( {{H_0}} \right){F_{{R_{nk}}}}\left( {{H_0}} \right),
\label{eq_cdf_information_perslot1_1}
\end{align}
Based on (\ref{eq_cdf_information_perslot1_1}), if continuing the same procedure on $F_{{R_{nk}}}^{\left( {m - 1} \right)}\left( {{H_0}} \right),\cdots,F_{{R_{nk}}}^{\left( {2} \right)}\left( {{H_0}} \right)$, we further get
\begin{equation}
F_{{R_{nk}}}^{\left( m \right)}\left( {{H_0}} \right) < F_{{R_{nk}}}^{\left( {m - 1} \right)}\left( {{H_0}} \right){F_{{R_{nk}}}}\left( {{H_0}} \right) < F_{{R_{nk}}}^{\left( {m - 2} \right)}\left( {{H_0}} \right){\left[ {{F_{{R_{nk}}}}\left( {{H_0}} \right)} \right]^2} <  \cdots  < {\left[ {{F_{{R_{nk}}}}\left( {{H_0}} \right)} \right]^m},\ {\rm{for}}\ m\ge 2,
\end{equation}
and by including the case of $F_{R_{nk}}\left(H_0\right)<F_{R_{nk}}^{\left(0\right)}\left(H_0\right)=1$, we further get
\begin{equation}
F_{{R_{nk}}}^{\left( m \right)}\left( {{H_0}} \right) < F_{{R_{nk}}}^{\left( {m'} \right)}\left( {{H_0}} \right), {\ \rm{if}}\ m > m'\ge 0.
\end{equation}
\end{proof}

\subsection{The network capacity region with RMIA}
\label{subsec: capacity_region_RMIA_main}
In this subsection, we aim to characterizes the network capacity region with RMIA by a stationary randomized policy with RMIA transmissions scheme. In Ref. \cite{Neely_Rahul_DIVBAR_2009}, a stationary randomized policy with REP is defined as: in each timeslot, each node $n$ uses a fixed probability to choose each commodity to transmit, and a fixed probability to forward the decoded packet to each node within the successful receiver set known through ACK/NACKs, where these fixed probabilities are independent of the backlog states. The underlying assumption behind the definition is the REP transmission scheme.

Likewise, the stationary randomized policy with RMIA is defined to have two similar characteristics: each node $n$ uses a fixed probability to choose each commodity to transmit in each timeslot; once a packet of commodity $c$ transmitted by node $n$ is firstly decoded by one or more receiving nodes, node $n$ uses a fixed probability to forward the decoded packet to each successful receiver. The major difference of the stationary randomized policy with RMIA from the REP version lies on the receiver side, where the decoding under RMIA takes advantage of the partial information accumulated during previous transmission attempts.

To begin with, we re-state the characterization of the network capacity region with REP, which was derived in Ref. \cite{Neely_Rahul_DIVBAR_2009}:
\begin{thm}
\label{cap_region_thm}
The network capacity region (with REP) $\Lambda_{\rm{REP}}$ consists of all the exogenous time average input rate matrices $\left( {\lambda _n^{\left( c \right)}} \right)$ for which there exists a stationary randomized policy that chooses probabilities $\alpha _n^{**\left( c \right)}$, $\theta _{nk}^{**\left( c \right)}\left( {{\Omega _n}} \right)$, and forms the time average flow rate $b_{nk}^{**\left( c \right)}$ (in unit of packet/slot), for all nodes $ n,k,c \in \cal N$ and all nonempty subsets ${{\Omega _n}}$ included by node $n$'s neighbor nodes set ${\cal K}_n$, such that:
\begin{equation}
\label{cap_region_const1}
b_{nk}^{**\left( c \right)} \ge 0,b_{cn}^{**\left( c \right)} = 0,b_{nn}^{**\left( c \right)} = 0,{\rm{\ for\ }}n \ne c,
\end{equation}
\begin{equation}
\label{cap_region_const2}
\sum\limits_{k\in {\cal K}_n} {b_{kn}^{**\left( c \right)}}  + \lambda _n^{\left( c \right)} \le \sum\limits_{k\in {\cal K}_n} {b_{nk}^{**\left( c \right)}} ,{\rm{\ for}}\ n \ne c,
\end{equation}
\begin{equation}
\label{cap_region_const3}
b_{nk}^{**\left( c \right)} \le \alpha _n^{**\left( c \right)}\sum\limits_{\Omega_n: {\Omega _n} \subseteq {{\cal K}_n}} {q_{n,{\Omega _n}}^{\rm{rep}}\theta _{nk}^{**\left( c \right)}\left( {{\Omega _n}} \right)},
\end{equation}
where $\alpha _n^{**\left( c \right)}$ is the probability that node $n$ decides to transmit a packet of commodity $c$ in each timeslot; ${q_{n,{\Omega _n}}^{\rm{rep}}}$ is the probability that $\Omega_n$ is the successful receiver set for a packet transmitted by node $n$ with REP; $\theta _{nk}^{**{\left( c \right)}}\left( {{\Omega _n}} \right)$ is the conditional probability that node $n$ forwards a packet of commodity $c$ to node $k$, given that the successful receiver set is $\Omega_n$.
\end{thm}
In Theorem \ref{cap_region_thm}, the superscript $**$ of a variable indicates that the value of this variable is related to the implemented policy, while the variable $q_{n,\Omega_n}^{\rm{rep}}$ does not depend on the policy but depends on the adopted REP transmission scheme. The detailed proof of Theorem \ref{cap_region_thm} is shown in Ref. \cite{Neely_Rahul_DIVBAR_2009}.

Here we clarify that the summation notation $\sum_{{\Omega _n}:{\Omega _n} \in {{\cal K}_n}} {} $ in (\ref{cap_region_const3}) means the summation over all possible subsets $\Omega_n$ included by the neighbor set ${\cal K}_n$. In this work, from now on, the summation notation with form $\sum_{x:{\rm{\ Expression}}\left( x \right)} {}$ means the summation over all possible $x$ satisfying Expression$\left(x\right)$; in contrast, if the summation notation has a form similar as $\sum_{x \in {{\Phi}}} {} $ without the specification of summing index variable before a colon, by default, it means the summation is over $x$ satisfying $x \in \Phi$.

Theorem \ref{cap_region_thm} is based on the REP assumption. However, the following corollary makes an analogous statement, namely that a stationary randomized policy achieves the network capacity region also holds true with RMIA (however, note that only the \textit{structures} of the solutions are similar, while the actual \textit{values} of the flow rate etc. are different; note that the there is superscript $*$ instead of $**$ on the variables in Corollary \ref{cap_region_corr}).
\begin{corollary}
\label{cap_region_corr}
 With RMIA transmission scheme, the network capacity region $\Lambda_{\rm{RMIA}}$ consists of all the exogenous time average input rate matrices $\left( {\lambda _n^{\left( c \right)}} \right)$ for which there exists a stationary randomized policy that chooses probability $\alpha _n^{*\left( c \right)}$, $\theta _{nk}^{*\left( c \right)}\left( {{\Omega _n}} \right)$ and forms the time average flow rate $b_{nk}^{*\left( c \right)}$ (in unit of packet/slot), for all nodes $ n,k,c \in \cal N$, and all the subsets ${{\Omega _n}}$ included by node $n$'s neighbors ${\cal K}_n$, such that constraints with the similar structures as (\ref{cap_region_const1})-(\ref{cap_region_const3}) are satisfied:
\begin{equation}
\label{cap_region_const1_coll}
b_{nk}^{*\left( c \right)} \ge 0,b_{cn}^{*\left( c \right)} = 0,b_{nn}^{*\left( c \right)} = 0,{\rm{\ for\ }}n \ne c,
\end{equation}
\begin{equation}
\label{cap_region_const2_coll}
\sum\limits_{k\in {\cal K}_n} {b_{kn}^{*\left( c \right)}}  + \lambda _n^{\left( c \right)} \le \sum\limits_{k\in {\cal K}_n} {b_{nk}^{*\left( c \right)}} ,{\rm{for}}\ n \neq c,
\end{equation}
\begin{equation}
\label{cap_region_const3_coll}
b_{nk}^{*\left( c \right)} \le \alpha _n^{*\left( c \right)}\beta _n^{{\rm{rmia}}}\sum\limits_{\Omega_n: {\Omega _n} \subseteq {{\cal K}_n}} {q_{n,{\Omega _n}}^{{\rm{rmia}}}\theta _{nk}^{*\left( c \right)}\left( {{\Omega _n}} \right)},
\end{equation}
where $\alpha _n^{*\left( c \right)}$ is the probability that node $n$ decides to transmit a packet of commodity $c$ in each timeslot; $\beta _n^{{\rm{rmia}}}$ is a constant representing the inverse value of the expected epoch length for node $n$ with RMIA; ${q_{n,{\Omega _n}}^{\rm{rmia}}}$ is the probability that $\Omega_n$ is the first successful receiver set for node $n$ in each epoch with RMIA; ${\theta _{nk}^{*\left( c \right)}\left( {{\Omega _n}} \right)}$ is the conditional probability that node $n$ forwards a packet of commodity $c$ to node $k$, given that the first successful receiver set is $\Omega_n$.
\end{corollary}

Corollary \ref{cap_region_corr} is proven by first showing that the given constraints (\ref{cap_region_const1_coll})-(\ref{cap_region_const3_coll}) are necessary for all the policies that support any $\left(\lambda_n^{\left(c\right)}\right)$ in the RMIA network capacity region $\Lambda_{\rm{RMIA}}$, which is shown in Appendix \ref{appendix: necessity_corollary_1}.

As for the sufficiency part, we need to show that if there exists a stationary randomized policy, under which the time average flow rates over the links together with any input rate matrix $\left(\lambda_n^{\left(c\right)}\right)$ within $\Lambda_{\rm{RMIA}}$ satisfy the constraints (\ref{cap_region_const1_coll})-(\ref{cap_region_const3_coll}), then the stationary randomized policy can stably support $\left(\lambda_n^{\left(c\right)}\right)$. The sufficiency part is proven in Appendix \ref{appendix: proof_sufficiency_corollary1}.

The proof of the sufficiency part bears some resemblance to the theoretical analysis in Ref. \cite{Neely_Rahul_DIVBAR_2009}, which analyzes the upper bound of the \emph{one-timeslot Lyapunov drift} of the stationary randomized policy with REP. As shown in Ref. \cite{Neely_Rahul_DIVBAR_2009}, although one-timeslot Lyapunov drift analysis works well with the assumption that the flow rate on each link is i.i.d. across timeslots, which is guaranteed with REP, it is not enough to analyze the statistical properties of the related metrics under the stationary randomized policy with RMIA. In fact, because of using the MIA technique, the flow rate under the stationary randomized policy on each link is no longer i.i.d. across timeslots but is related to the partial information already accumulated in the receiving nodes, which is the result of the previous history. However, with the renewal operation implemented at the end of each epoch for each node, there still exists some underlying good statistical properties in the metrics under the stationary randomized policy with RMIA. As will be shown in proving the sufficiency part of Corollary \ref{cap_region_corr}, for each transmitting node and for each commodity, the lengths of epochs under the stationary randomized policy are i.i.d.. Based on this property, the network throughput analysis can be done by using \emph{$d$-timeslot average Lyapunov drift}, which is defined as follows:
\begin{equation}
\frac{1}{d}\sum\limits_{n,c} {\mathbb{E}_\omega\left\{ {\left. {{{\left( {Q_n^{\left( c \right)}\left( {{t_0} + d} \right)} \right)}^2} - {{\left( {Q_n^{\left( c \right)}\left( {{t_0}} \right)} \right)}^2}} \right|\;{\bf{Q}}\left( {{t_0}} \right)} \right\}},
\label{eq_d-slot_lyapunov_drift}
\end{equation}
where $d$ is a positive interval length (in unit of timeslot); $t_0$ is an arbitrary timeslot; the vector ${\bf{Q}}\left(t_0\right)$ represents the backlog state of the network in timeslot $t_0$; $\mathbb{E}_\omega$ is the expectation operation taken with respect to $\omega$, which is the elementary random event defined for the network over the whole time horizon.

In order to explain the expectation operator $\mathbb{E}_\omega$ in (\ref{eq_d-slot_lyapunov_drift}), we need to define the probability space, denoted as $\left({\cal S}, {\cal F}, \mathbb{P}\right)$, for the network over the whole time horizon. Specifically, we firstly define the following sample spaces and their corresponding $\sigma$-fields:
\begin{itemize}
\item Let ${\cal S}_{\rm{ch}}$ represent the sample space consisting of all possible channel realization sample paths over the whole time horizon in the network; let the set of all possible channel realization events to be the corresponding $\sigma$-field, denoted as ${\cal F}_{\rm{ch}}$.
\item Let ${\cal S}_{\rm{arr}}$ represent the sample space consisting of all possible exogenous input arrival sample paths over the whole time horizon to network; let the set of all possible exogenous arrival events to be the corresponding $\sigma$-field, denoted as ${\cal F}_{\rm{arr}}$.
\item Among all possible policies (with transmission scheme discrimination) that can be implemented on the network over the whole time horizon, index them as $1, 2, 3, \cdots$. Let ${\cal S}_{j,\rm{dec}}$ represent the sample space consisting of the sample paths of all possible decision outcomes under the $j$th policy over the whole time horizon in the network; let the set of all possible decision events to be the corresponding $\sigma$-field, denoted as ${\cal F}_{j,\rm{dec}}$. The definitions of ${\cal S}_{j,\rm{dec}}$ and ${\cal F}_{j,\rm{dec}}$ are for the analysis of the randomized policies, such as the stationary randomized policy introduced in this work. The non-randomized policies, such as DIVBAR-RMIA and DIVBAR-MIA proposed in this work, can also be dominated by this general framework.
\end{itemize}
With the above definitions, the total sample space $\cal S$ for the whole network over the whole time horizon under all possible policies is defined as
\begin{equation}
{\cal S} = {\cal S}_{\rm{ch}} \times {\cal S}_{\rm{arr}} \times {\cal S}_{1,\rm{dec}} \times {\cal S}_{2,\rm{dec}} \times {\cal S}_{3,\rm{dec}} \times\cdots;
\end{equation}
and the total $\sigma$-field for the whole network over the whole time horizon under all possible policies can be defined as
\begin{equation}
{\cal F} = {\cal F}_{\rm{ch}} \times {\cal F}_{\rm{arr}} \times {\cal F}_{1,\rm{dec}} \times {\cal F}_{2,\rm{dec}} \times {\cal F}_{3,\rm{dec}} \times\cdots.
\end{equation}
For the probability measure $\mathbb{P}$, which is a function mapping from $\cal F$ to $\left[0,1\right]$, is determined by the statistics of the channel realizations of the links in the network over the whole time horizon, the statistics of the exogenous input arrivals to the network over the whole time horizon, the statistics of the decisions made by all possible policies, the queueing dynamics, and the constraints imposed on the routing.

The definition of the probability space $\left({\cal S}, {\cal F}, \mathbb{P}\right)$ sets up a unified framework for the analysis of all the random variables, such as the flow rate, backlog state, etc., under any policy. In other words, any random variable $X$ in the routing of the network under any policy is a function mapping from $\cal S$ to the set of real numbers $\mathbb{R}$ with the property that $\left\{\omega \in {\cal S}: X\left(\omega\right)\le x\right\} \in \cal F$, and the variable's distribution function satisfies: $F_X\left(x\right)=\mathbb{P}\left(X\le x\right)$. Therefore, from a general point of view, the expectation operation should be taken with respect to $\omega$, which is the elementary event in $\cal S$. From now on, if not specifically clarified, the expectation operator $\mathbb{E}$ represents $\mathbb{E}_\omega$ for notational simplification.

Based on the probability space $\left({\cal S}, {\cal F}, \mathbb{P}\right)$, going back to (\ref{eq_d-slot_lyapunov_drift}) and following the idea of Lyapunov drift analysis, the key metric should arise from the upper bound of the $d$-timeslot average Lyapunov drift. The intuition behind using the $d$-timeslot average Lyapunov drift is two fold:
\begin{itemize}
\item If $d$ takes a large value, the time interval starting from timeslot $t_0$ to timeslot $t_0+d-1$ has a high probability of including a large number of epochs for each transmitting node. In other words, the key metric for each transmitting node in the upper bound of the $d$-timeslot average Lyapunov drift can be interpreted as the average of the metrics over single epochs, which are within the $d$ timeslots interval. With the i.i.d. property of the epoch lengths, the key metric in the upper bound of the $d$-timeslot average Lyapunov drift shows convergence.
\item Timeslot $t_0$ and timeslot $t_0+d-1$ may be located in the middle of certain epochs, i.e., the $d$ timeslots interval may not contain an integer multiple of epochs. Instead, there might be some "marginal" interval at the beginning or at the end of the $d$ timeslots interval, and the distribution of the "marginal" interval length is difficult to characterize. However, since the Lyapunov drift an averaged over $d$ timeslots, the effect of the "marginal" interval becomes negligible as $d$ grows large.
\end{itemize}

\subsection{Network capacity region: RMIA versus REP}
\label{subsec: capacity_region_RMIA_vs_REP_main}
Comparing with REP, RMIA potentially increases the success probability of the transmission attempts over each wireless link in the network. With this fact, we can intuitively comprehend that the network has an "enhanced ability" of delivering the arriving packets. Based on this intuition, we also predict that the network capacity can also be enlarged by using RMIA instead of REP, which is organized as the following theorem:
\begin{thm}
\label{thm: capacity_comparison}
Let the amount of information transmitted over each link in each timeslot be continuously distributed over $\left[0,\infty\right)$, and the corresponding cdf $F_{R_{nk}}\left(x\right)$ satisfies: $0 < {F_{{R_{nk}}}}\left( {{x}} \right) < 1$, where $x>0$. The RMIA network capacity region $\Lambda_{\rm{RMIA}}$ is strictly larger than the REP network capacity region $\Lambda_{\rm{REP}}$, i.e., $\Lambda_{\rm{RMIA}}  \supset {\Lambda_{\rm{REP}}}$.
\end{thm}

The detailed proof of Theorem \ref{thm: capacity_comparison} is given in Appendix \ref{appendix: capacity comparison}.

The assumption: $0 < {F_{{R_{nk}}}}\left( {{x}} \right) < 1$, $\forall n,k \in \cal N$, guarantees that the claim in Theorem \ref{thm: capacity_comparison} holds true for any possible (positive) value of packet entropy. With this mild assumption, Theorem \ref{thm: capacity_comparison} demonstrates that the network has a non-zero potential of further increasing its throughput by using transmission scheme RMIA instead of REP. In other words, Theorem \ref{thm: capacity_comparison} shows an opportunity to develop a routing algorithm with RMIA that has a better throughput performance than the original throughput optimal algorithm with REP that achieves $\Lambda_{\rm{REP}}$.

\section{Diversity Backpressure Routing Algorithms with Mutual Information Accumulation}
\label{sec: algorithm_discription}
In Ref. \cite{Neely_Rahul_DIVBAR_2009}, the Diversity Backpressure (DIVBAR) routing algorithm is proposed for wireless ad-hoc network and has been shown to be throughput optimal among all possible algorithms with REP assumption. In this section, we use the DIVBAR algorithm as a reference and develop two routing algorithms with MIA technique: DIVBAR-RMIA and DIVBAR-MIA, in order to further enhance the throughput performance. Both proposed algorithms work in the similar manner as the original DIVBAR algorithm. That is, each node $n$ in the network dynamically make routing decisions based on the observation of the \emph{backlog state}, which is the backlog information of CPQs in node $n$ and in the neighbor nodes $k\in{\cal K}_n$. The key difference between DIVBAR-RMIA and DIVBAR-MIA is: DIVBAR-RMIA clears the partial information of each packet in the network, whenever one copy of the packet being transmitted from node $n$ is firstly decoded by one or more receiving nodes in ${\cal K}_n$; DIVBAR-MIA retains the partial information until the corresponding packet is delivered to the destination.

\subsection{Diversity Backpressure Routing with Renewal Mutual Information Accumulation (DIVBAR-RMIA)}
\label{subsec: DIVBAR-RMIA}
We summarize the DIVBAR-RMIA algorithm for each transmitting node $n$ in its $i$th epoch as the following steps, where the notation in the form of $\hat x$ means that the value of the variable $x$ is specifically determined by DIVBAR-RMIA:
\begin{enumerate}
\item In the starting timeslot $u_{n,i}$ of each epoch $i$ for the transmitting node $n$, node $n$ observes the CPQ backlog of each commodity $c\in \cal N$ in each of its potential receiver $k \in {\cal K}_n$. Combining with its own backlog of CPQ of each commodity, node $n$ computes the d\emph{ifferential backlog coefficient} as follows:
    \begin{equation}
    \hat W_{nk}^{\left( c \right)}\left( {{u_{n,i}}} \right) = \max \left\{ {\hat Q_n^{\left( c \right)}\left( {{u_{n,i}}} \right) - \hat Q_k^{\left( c \right)}\left( {{u_{n,i}}} \right),0} \right\}.
    \label{eq_differential_backlog_main}
    \end{equation}
    \label{step: differential backlog_main_RMIA}
\item For each commodity $c$, the potential receivers in ${\cal K}_n$ are ranked according to their corresponding $W_{nk}^{\left( c \right)}\left( {{u_{n,i}}} \right)$ weights sorted in descending order. We define ${\cal {\hat R}}_{nk}^{\rm{high},\left(c\right)}\left(u_{n,i}\right)$ and ${\cal {\hat R}}_{nk}^{\rm{low},\left(c\right)}\left(u_{n,i}\right)$ respectively as the set of the receivers $j \in {\cal K}_n$ with higher and lower rank than receiver $k$ in timeslot $u_{n,i}$, i.e.,
    \begin{align}
    \hat W_{nj}^{\left( c \right)}\left( {{u_{n,i}}} \right) &\ge \hat W_{nk}^{\left( c \right)}\left( {{u_{n,i}}} \right){\rm{\ for\ }}\forall j \in {\cal {\hat R}}_{nk}^{{\rm{high,}}\left( c \right)}\left( {{u_{n,i}}} \right);\nonumber\\
    \hat W_{nj}^{\left( c \right)}\left( {{u_{n,i}}} \right) &\le \hat W_{nk}^{\left( c \right)}\left( {{u_{n,i}}} \right){\rm{\ for\ }}\forall j \in {\cal {\hat R}}_{nk}^{{\rm{low,}}\left( c \right)}\left( {{u_{n,i}}} \right).
    \end{align}
    \label{step: define_priority_set_RMIA_main}
\item Define $\hat \varphi_{nk}^{\left(c\right)}\left(i\right)$ as the probability that a packet of commodity $c$ is firstly decoded by the receiving node $k \in {\cal K}_n$ in the first decoding timeslot of epoch $i$, while the receiving nodes in set ${\cal {\hat R}}_{nk}^{\rm{high},\left(c\right)}\left(u_{n,i}\right)$ do not successfully decode, i.e., node $k$ has the highest priority among the successful receivers in the first successful receiver set.
    \label{step: define_hiararchy_set_RMIA_main}
\item Define $\hat { c}_{n}\left(i\right)$ as the optimal commodity that maximizes the following backpressure metric:
    \begin{equation}
    \sum\limits_{k \in {{\cal K}_n}} {\hat W_{nk}^{\left( c \right)}\left( {{u_{n,i}}} \right)\hat \varphi _{nk}^{\left( c \right)}\left(i\right)}.
    \label{eq_backpressure_metric_main}
    \end{equation}
    Define $\hat \Xi_{n}\left(i\right)$ as the resulting maximum value:
    \begin{equation}
    {\hat \Xi _n}\left(i\right) = \sum\limits_{k \in {{\cal K}_n}} {\hat W_{nk}^{\left( {{{{\hat c}}_n}\left(i\right)} \right)}\left( {{u_{n,i}}} \right)\hat \varphi _{nk}^{\left( {{{ {\hat c}}_n}\left(i\right)} \right)}\left(i\right)}.
    \label{eq_maximum_metric_RMIA_main}
    \end{equation}
    \label{step: compute_metric_DIVBAR-RMIA_main}
\item If $\hat \Xi_n\left(i\right)>0$, node $n$ chooses a packet at the head of the CPQ of commodity $\hat c_n\left(i\right)$ to transmit in the current timeslot $\tau$, where $\tau= u_{n,i}$. Else node $n$ starts transmitting a null packet.
    \label{step: choose_commodity_to_transmit_RMIA_main}
\item After the transmission in the current timeslot $\tau$, where $\tau \ge u_{n,i}$, each receiver $k\in {\cal K}_n$ sends ACK/NACK back to node $n$ indicating whether node $k$ decoded the packet or not.
    \label{step: generating_ACK_RMIA_main}
\item After gathering all the ACK/NACK feedbacks from all the receiving nodes of ${\cal K}_n$ in the current timeslot $\tau$, where $\tau \ge u_{n,i}$, node $n$ checks if there is any receiving node that decoded the packet in timeslot $\tau$.  If yes, timeslot $\tau$ is the ending timeslot of current epoch, and the algorithm goes to step \ref{step: forwarding_packet}); if not, there's no more operation left in the current timeslot, and the algorithm goes to step \ref{step: keep_transmitting}) for the next timeslot $\tau+1$;
    \label{step: judge_first_decoding_set}
\item Knowing the fact that none of the receivers in set ${\cal K}_n$ successfully decoded the packet being transmitted in the previous timeslot $\tau-1$, node $n$ keeps transmitting the packet in the current timeslot $\tau$, where $\tau\ge u_{n,i}+1$. On the receiver side, each receiving node keeps accumulating the partial information of the packet. Then the algorithm goes back to step \ref{step: generating_ACK_RMIA_main}).
    \label{step: keep_transmitting}
\item After knowing that there is at least one successful receiving node in timeslot $\tau$, node $n$ firstly checks if the packet being transmitted is a null packet or a valid packet. In the former case, the algorithm directly goes to step \ref{step: epoch_end}). In the later case, node $n$ shifts the forwarding responsibility of the decoded packet to the successful receiver $k$ with the largest positive differential backlog coefficient $\hat W_{nk}^{\left(\hat c_n\left(i\right)\right)}\left(u_{n,i}\right)$, while it retains the forwarding responsibility if none of successful receiver has positive differential backlog. Then the algorithm goes to step \ref{step: epoch_end}).
    \label{step: forwarding_packet}
\item If the decoded packet is a null packet, all the copies and partial information of the null packet in the transmitting and receiving nodes are cleared. If the decoded packet is valid, only the node which has the forwarding responsibility keeps the packet while other nodes clear either the partial information or complete copies of the packet. Then the current epoch of transmitting node $n$ ends.
    \label{step: epoch_end}
\end{enumerate}

In the above summary of the DIVBAR-RMIA algorithm, the value of the probability $\hat \varphi_{nk}^{\left(c\right)}\left(i\right)$ can be computed with the knowledge of the distribution of the amount of information transmitted per timeslot $R_{nk}\left(\tau\right)$. Note that $R_{nk}\left(\tau\right)$ is i.i.d. across timeslots. Define $T_n\left(i\right)$ as the number of timeslots in the $i$th epoch. We compute the value of $\hat \varphi_{nk}^{\left(c\right)}\left(i\right)$ as follows:
\begin{align}
\hat \varphi _{nk}^{\left( {{c_n}} \right)}\left( i \right)&= \sum\limits_{m = 1}^\infty  {\Pr \left( {{\rm{node\ }}k{\rm{\ has\ the\ highest\ priority\ in\ the\ first\ successful\ receiver\ set,\ }}{T_n}\left( i \right) = m} \right)}\nonumber\\
&= \sum\limits_{m = 1}^\infty  {\prod\limits_{j \in {\cal {\hat R}}_{nk}^{{\rm{high,}}\left( c \right)}\left( {{u_{n,i}}} \right)} {\Pr \left( {\sum\limits_{\tau  = 1}^m {{{ R}_{nj}}\left( \tau  \right)}  < {H_0}} \right)} }\cdot \prod\limits_{j \in {\cal {\hat R}}_{nk}^{{\rm{low,}}\left( c \right)}\left( {{u_{n,i}}} \right)} {\Pr \left( {\sum\limits_{\tau  = 1}^{m - 1} {{R_{nj}}\left( \tau  \right)}  < {H_0}} \right)}\nonumber\\
&\ \ \ \ \ \ \ \ \ \ \cdot \Pr \left( {\sum\limits_{\tau  = 1}^{m - 1} {{R_{nk}}\left( \tau  \right)}  < {H_0} \le \sum\limits_{\tau  = 1}^m {{R_{nk}}\left( \tau  \right)} } \right)\nonumber\\
&= \sum\limits_{m = 1}^\infty  {\prod\limits_{j \in {\cal {\hat R}}_{nk}^{{\rm{high,}}\left( c \right)}\left( {{u_{n,i}}} \right)} {F_{{R_{nj}}}^{\left( m \right)}\left( {{H_0}} \right)}  \cdot \prod\limits_{j \in {\cal {\hat R}}_{nk}^{{\rm{low,}}\left( c \right)}\left( {{u_{n,i}}} \right)} {F_{{R_{nj}}}^{\left( {m - 1} \right)}\left( {{H_0}} \right)}  \cdot \left[ {F_{{R_{nk}}}^{\left( {m - 1} \right)}\left( {{H_0}} \right) - F_{{R_{nk}}}^{\left( m \right)}\left( {{H_0}} \right)} \right]} ,
\label{eq_probability_for_choosing_commodity_RMIA_main}
\end{align}
where $F_{R_{nk}}^{\left(m\right)}\left(x\right)$ can be computed iteratively:
\begin{equation}
F_{{R_{nk}}}^{\left( m \right)}\left( x \right) = \int\limits_0^{{H_0}} {F_{{R_{nk}}}^{\left( {m - 1} \right)}\left( {x - y} \right){f_{{R_{nk}}}}\left( y \right)dy},\ {\rm{for}}\ k \in {\cal K}_n.
\end{equation}

Moreover, according to the above description, the DIVBAR-RMIA algorithm is a distributed algorithm because, when implementing the algorithm, each node $n$ only requires queue backlog information of its neighbor nodes (potential receivers) and the link success probabilities of those neighbor nodes over the transmission periods of different epoch lengths, which depends on the value of $F_{R_{nk}}^{\left(m\right)}\left(H_0\right)$. Since the network is stationary, the cdf functions $F_{R_{nk}}^{\left(m\right)}\left(x\right)$ can be computed off line, and each node $n$ in the network can compute and store the values of $F_{R_{nk}}^{\left(m\right)}\left(H_0\right)$ with different epoch length $m$ for each potential receiver $k \in {\cal K}_n$ before the algorithm starts. During the algorithm's running time, in the starting timeslot $u_{n,i}$ of each epoch for the transmitting node $n$, node $n$ just uses the sorted differential backlog $\hat W_{nk}\left(u_{n,i}\right)$ obtained from backlog observation and the stored values of $F_{R_{nk}}^{\left(m\right)}\left(H_0\right)$ to compute the value $\hat \varphi_{nk}^{\left(c\right)}\left(i\right)$ according to (\ref{eq_probability_for_choosing_commodity_RMIA_main}), and further compute the metric $\hat \Xi_n\left(i\right)$ according to (\ref{eq_maximum_metric_RMIA_main}) and find the corresponding optimal commodity $\hat c_n\left(i\right)$ to transmit.

Furthermore, with contiguous timeslots consisting of each epoch $i$, the transmission decision made by node $n$ during epoch $i$ is based on the backlog observation in the starting timeslot $u_{n,i}$ of the epoch and remains the same through the whole epoch. On other other hand, node $n$ is not only a transmitting node but also simultaneously a receiving node respective to each of its neighbor nodes $k\in {\cal K}_n$ as well. Note that node $n$ and node $k$ can have non-synchronized epochs due to the different channel conditions on their emanating links. Therefore, it's possible that the backlogs in node $n$ and node $k$ are updated in the middle of the epoch $i$ for node $n$, when an epoch for node $k$ ends and a packet is forwarded from node $k$ to node $n$. In this case, the original backlog state observation made by node $n$ in timeslot $u_{n,i}$ is no longer "fresh" since the update of the backlogs in node $n$ and node $k$, which in turn indicates that the transmission decision made by node $n$ for the remaining part of epoch $i$ is based on an outdated backlog observation. However, as will be shown in Theorem \ref{thm: DIVBAR_MIA_throughput_optimal} later, the throughput performance is not affected by the outdated backlog observation.

\subsection{Diversity Backpressure Routing with Mutual Information Accumulation (DIVBAR-MIA)}
\label{subsec: DIVBAR-MIA}
DIVBAR-RMIA clears the partial information on all nodes as soon as the corresponding packet is firstly decoded by one or more receiving nodes. However, the eliminated partial information could be useful for the future decoding, and therefore, the clearance operation may be a waste of "resource".

In contrast, DIVBAR-MIA uses the same strategy of choosing commodity to transmit and choosing the successful receiving node to forward the decoded packet, but retains the partial information of the corresponding packet until the packet reaches its destination. Moreover, we consider here a version of DIVBAR-MIA that is synchronized with DIBAR-RMIA in the sense that it is set to perform in synchronized epochs with DIVBAR-RMIA, i.e., each node starts to transmit a new copy of packet only when it would also start to transmit a new copy of packet under DIVBAR-RMIA. Thus, DIVBAR-MIA results in that the potential receivers of each transmitting node $n$ may have already accumulated certain amount of partial information of a packet before node $n$ starts transmitting the packet. With the property of having synchronized epochs with DIVBAR-RMIA, DIVBAR-MIA guarantees that the successful receiver set at the end of each epoch of node $n$ includes the first successful receiver set under DIVBAR-RMIA. In other words, the set of decoding nodes might be enlarged at the end of each epoch by using DIVBAR-MIA instead of using DIVBAR-RMIA.

Nevertheless, here "being synchronized with DIVBAR-RMIA" does not mean that we have to run DIVBAR-RMIA simultaneously with DIVBAR-MIA. Instead, DIVBAR-MIA only needs the information of the epochs' starting and ending time under DIVBAR-RMIA. This can be done in each receiving node by keeping checking if the receiving node has accumulated enough amount of partial information to decode the packet being transmitted, without the help of the \emph{pre-accumulated partial information} respective to the starting timeslot of the current epoch starts.

Before summarizing the DIVBAR-MIA algorithm, define $p_{n,i}$ as the index of the packet which is being transmitted by node $n$ in its $i$th epoch. Let $I_{k,p_{n,i}}^{\rm{pre}}\left(u_{n,i}\right)$ represent the amount of pre-accumulated partial information of packet $p_{n,i}$ respective to timeslot $u_{n,i}$ stored in node $k$. Additionally, define $I_{nk}^{\rm{rmia}}\left(u_{n,i},\tau\right)$ as the amount of partial information of packet $p_{n,i}$ purely accumulated by the transmissions in the time interval from timeslot $u_{n,i}$ to the timeslot $\tau$.

The definitions of $I_{k,p_{n,i}}^{\rm{pre}}\left(u_{n,i}\right)$ and $I_{nk,p_{n,i}}^{\rm{rmia}}\left(u_{n,i},\tau\right)$ indicate that in timeslot $\tau$, under RMIA, the receiving node $k$ is only allowed to decode the packet being transmitted using the amount of partial information $I_{nk,p_{n,i}}^{\rm{rmia}}\left(u_{n,i},\tau\right)$, while under MIA, the receiving node $k$ is allowed to decode the packet being transmitted by using the amount of partial information $I_{k,p_{n,i}}^{\rm{pre}}\left(u_{n,i}\right)+I_{nk,p_{n,i}}^{\rm{rmia}}\left(u_{n,i},\tau\right)$. When implementing DIVBAR-MIA, node $k$ is assumed to be able to distinguish the two kinds of partial information belonging to the same packet being transmitted.

Now we summarize DIVBAR-MIA algorithm for each transmitting node $n$ in its $i$th epoch as the following steps, where the notation in the form of $\hat {\hat x}$ means that the value of the variable $x$ is specifically determined by DIVBAR-MIA:
\begin{enumerate}
\item On the transmitter side, node $n$ executes the similar steps as Step \ref{step: differential backlog_main_RMIA})-\ref{step: choose_commodity_to_transmit_RMIA_main}) in the algorithm summary of DIVBAR-RMIA but based on the backlog observations ${\bf{\hat{\hat Q}}}\left(u_{n,i}\right)$ under DIVBAR-MIA, in which the differential backlog coefficient $\hat{\hat{W}}_{nk}^{\left(c\right)}\left(u_{n,i}\right)$, the node sets $\hat {\hat {{\cal R}}}_{nk}^{\rm{high},\left(c\right)}\left(u_{n,i}\right)$ and $\hat {\hat {{\cal R}}}_{nk}^{\rm{low},\left(c\right)}\left(u_{n,i}\right)$, the probability value $\hat {\hat \varphi}_{nk}^{\left(c\right)}\left(i\right)$, the commodity $\hat{\hat{c}}_n\left(i\right)$ that might be chosen to transmit, and the resulting backpressure metric $\hat{\hat{\Xi}}_n\left(i\right)$ are computed.
\item On the receiver side, after the data transmission in timeslot $\tau$, where $\tau \ge u_{n,i}$, each receiving node sends two feedback signals back to node $n$: $\rm (ACK/NACK)_{MIA}$ and $\rm (ACK/NACK)_{RMIA}$. $\rm (ACK/NACK)_{MIA}$ indicates whether node $k$ successfully decodes the packet with MIA, which is true if ${I_{k,{p_{n,i}}}}\left( {{u_{n,i}}} \right) + I_{nk,p_{n,i}}^{\rm{rmia}}\left(u_{n,i},\tau\right)\ge H_0$; $\rm (ACK/NACK)_{RMIA}$ indicates whether the partial information accumulated at node $k$ purely during the current epoch has been enough to decode the packet, which is true if $I_{nk,p_{n,i}}^{\rm{rmia}}\left(u_{n,i},\tau\right)\ge H_0$.
    \label{step: forming_ACKs_MIA_main}
\item After gathering all the $\rm (ACK/NACK)_{MIA}$ and $\rm (ACK/NACK)_{RMIA}$ feedbacks from all the receiving nodes in the current timeslot $\tau$, where $\tau \ge u_{n,i}$, node $n$ firstly check $\rm (ACK/NACK)_{RMIA}$ feedbacks to confirm if there is any receiving node that firstly accumulates enough information $I_{nk,p_{n,i}}^{\rm{rmia}}\left(u_{n,i},\tau\right)$ which exceeds $H_0$. If there is such a receiving node, timeslot $\tau$ should be the ending timeslot of current epoch. Then the algorithm goes to step \ref{step: choose_forwarding_node_MIA_main}). If none of the receiving node has accumulated enough amount of partial information $I_{nk,p_{n,i}}^{\rm{rmia}}\left(u_{n,i},\tau\right)$ to decode the packet, there is no more operation left in the current timeslot, and the algorithm goes to step \ref{step: keep_transmitting_MIA_main}) for the next timeslot $\tau+1$.
\item Knowing the fact that none of the receiving nodes in the set ${\cal K}_n$ has accumulated enough amount of partial information $I_{nk,p_{n,i}}^{\rm{rmia}}\left(u_{n,i},\tau-1\right)$ up to the previous timeslot $\tau-1$, node $n$ keeps transmitting the same packet in the current timeslot $\tau$, where $\tau \ge u_{n,i}+1$. Then the algorithm goes back to step \ref{step: forming_ACKs_MIA_main}).
    \label{step: keep_transmitting_MIA_main}
\item After knowing that there is at least one receiving node that has accumulated enough amount of partial information $I_{nk,p_{n,i}}^{\rm{rmia}}\left(u_{n,i},\tau\right)$ to decode the packet, node $n$ checks if the transmitted packet is a null packet or a valid one. In the former case, the algorithm directly goes to step \ref{step: epoch_end_MIA_main}). In the later case, node $n$ further checks the gathered $\rm (ACK/NACK)_{MIA}$ feedbacks, based on which node $n$ shifts forwarding responsibility of the decoded packet to the successful receiver $k$ with the largest positive differential backlog coefficient $\hat {\hat W}_{nk}^{\left(\hat {\hat c}_n\right)}\left(u_{n,i}\right)$, while it retains the forwarding responsibility if none of the successful receivers has positive differential backlog. Then the algorithm goes to step \ref{step: epoch_end_MIA_main}).
    \label{step: choose_forwarding_node_MIA_main}
\item If the decoded packet is null, any partial information or copies of the null packet in the transmitting and receiving nodes are cleared. If the decoded packet is valid, the nodes having a complete copy of the packet delete their copies except the node which gets the forwarding responsibility, while other receiving nodes retain their incomplete partial information of the packet. Then the current epoch of node $n$ ends.
    \label{step: epoch_end_MIA_main}
\end{enumerate}

Additional to the steps in each epoch shown above, after the packet $p_{n,i}$ is delivered to its destination, all the partial information belonging to packet $p_{n,i}$ stored in the network is cleared in order to free up the memory.

According to above algorithm summary, like DIVBAR-RMIA, DIVBAR-MIA is also a distributed algorithm, because when implementing the algorithm, each node $n$ only requires CPQs' backlog information of its neighbor nodes and the link success probabilities of those neighbor nodes over different epoch lengths. Here note that in the starting timeslot $u_{n,i}$ of each epoch $i$ for node $n$, node $n$ under DIVBAR-MIA does not observe the partial information already accumulated before timeslot $u_{n,i}$ and makes decisions to choose commodity to transmit according to the same strategy as DIVBAR-RMIA. However, on the receiver side, the pre-accumulated partial information can be used by node $k \in {\cal K}_n$ to facilitate the decoding of the packet $p_{n,i}$.

\section{Performance Analysis}
\label{sec: performance_anlaysis_main}
In this section, the performances of DIVBAR-RMIA and DIVBAR-MIA are evaluated. As will be shown in Theorem \ref{thm: DIVBAR_MIA_throughput_optimal}, we first prove the throughput optimality of DIVBAR-RMIA among all possible algorithms with RMIA, i.e., DIVBAR-RMIA can support any (time average) input rate matrix within the RMIA network capacity region $\Lambda_{\rm{RMIA}}$. Secondly, as will be shown in Theorem \ref{thm: DIVBAR-RMIA_vs_DIVBAR-MIA} and Corollary \ref{corollary: DIVBAR-MIA_vs_DIVBAR-RMIA}, we prove that DIVBAR-MIA's throughput performance is at least as good as DIVBAR-RMIA because it is also shown to be able to support any input rate matrix within $\Lambda_{\rm{RMIA}}$.

Theorem \ref{thm: DIVBAR_MIA_throughput_optimal}, Theorem \ref{thm: DIVBAR-RMIA_vs_DIVBAR-MIA} and Corollary \ref{corollary: DIVBAR-MIA_vs_DIVBAR-RMIA}, together with Corollary \ref{cap_region_corr} and Theorem \ref{thm: capacity_comparison} proposed in Section \ref{sec: network_capacity_region_RMIA} indicate the logic of showing that DIVBAR-RMIA and DIVBAR-MIA outperform the original DIVBAR algorithm with REP in terms of throughput performance, which can be summarized into four steps: (i) characterize the network capacity region with RMIA $\Lambda_{\rm{RMIA}}$; (ii) show that $\Lambda_{\rm{RMIA}}$ is strictly larger than the $\Lambda_{\rm{REP}}$; (iii) show that DIVBAR-RMIA can support any input rate matrix within $\Lambda_{\rm{RMIA}}$, which in turn shows that DIVBAR-RMIA outperforms the the original DIVBAR algorithm which reaches the whole $\Lambda_{\rm{REP}}$; (iv) DIVBAR-MIA can support any input rate matrix that can be supported by DIVBAR-RMIA. Note that we do not claim throughput optimality of DIVBAR-MIA among all policies with MIA. Also note that the implementation of DIVBAR-MIA we consider (using synchronized epochs with DIVAR-RMIA) is not the only possible implementation; rather it is chosen because its superiority can be proven exactly.

\subsection{Throughput optimality of DIVBAR-RMIA among all possible policies with RMIA}
\label{subsec: throughput_optimality_DIVBAR-RMIA}
In this subsection, our goal is to analyze the throughput performance of DIVBAR-RMIA algorithm and show that it is throughput optimal among all possible policies with RMIA.

To begin with, we need to analyze the backpressure metric under DIVBAR-RMIA over a single epoch. Firstly, let $\cal P$ represent the set of all possible policies under which the epochs for each transmitting node consist of contiguous timeslots. This definition demonstrates that DIVBAR-RMIA belongs to $\cal P$ and all the policies in $\cal P$ have synchronized epochs for each transmitting node. Then define ${Z_n}\left( {i,{\bf{\hat Q}}\left( {{u_{n,i}}} \right)} \right)$ as the following backpressure metric over the $i$th epoch under a policy in $\cal P$:
\begin{equation}
\label{eq_single_epoch_metric_P}
{Z_n}\left( {i,{\bf{\hat Q}}\left( {{u_{n,i}}} \right)} \right) = \sum\limits_c {\sum\limits_{\tau  = {u_{n,i}}}^{{u_{n,i + 1}} - 1} {\sum\limits_{k \in {{\cal K}_n}} {b_{nk}^{\left( c \right)}\left( \tau  \right)\left[ {\hat Q_n^{\left( c \right)}\left( {{u_{n,i}}} \right) - \hat Q_k^{\left( c \right)}\left( {{u_{n,i}}} \right)} \right]} } },
\end{equation}
where ${\bf{\hat Q}}\left(\tau\right)$ is the backlog state under DIVBAR-RMIA in timeslot $\tau$. With the definitions of ${\cal P}$ and ${\bf{\hat Q}}\left(\tau\right)$, we propose Lemma \ref{lemma: characterize_metric_DIVBAR-RMIA} as follows to characterize DIVBAR-RMIA over a single epoch and explain the origin of the backpressure metric shown as (\ref{eq_backpressure_metric_main}) in step \ref{step: compute_metric_DIVBAR-RMIA_main}) of the DIVBAR-RMIA algorithm summary. This lemma is essential to the proof of the throughput optimality of DIVBAR-RMIA among all the policies with RMIA.

\begin{lemma}
\label{lemma: characterize_metric_DIVBAR-RMIA}
For each node $n$, the metric $\mathbb{E}\left\{ {\left. {{Z_n}\left( {i,{\bf{\hat Q}}\left( {{u_{n,i}}} \right)} \right)} \right|{\bf{\hat Q}}\left( {{u_{n,i}}} \right)} \right\}$ under an arbitrary policy within the restricted policy set ${\cal P}$ is upper bounded as follows:
\begin{equation}
\mathbb{E}\left\{ {\left. {{Z_n}\left( {i,{\bf{\hat Q}}\left( {{u_{n,i}}} \right)} \right)} \right|{\bf{\hat Q}}\left( {{u_{n,i}}} \right)} \right\}\le \hat \Xi_n\left(i\right),
\label{eq_metric_single_epoch_comparison_DIVBAR-RMIA_main}
\end{equation}
where the upper bound $\hat \Xi_n\left(i\right)$ is the metric value shown in (\ref{eq_maximum_metric_RMIA_main}). The upper bound is achieved if implementing the DIVBAR-RMIA algorithm.
\end{lemma}
A detailed proof of Lemma \ref{lemma: characterize_metric_DIVBAR-RMIA} is given in Appendix \ref{appendix: single_epoch_opt_DIVBAR-RMIA}.

After characterizing the backpressure metric under DIVBAR-RMIA over a single epoch by Lemma \ref{lemma: characterize_metric_DIVBAR-RMIA}, we start proving the throughput optimality of DIVBAR-RMIA, which is to show that strong stability can be achieved under DIVBAR-RMIA whenever the input rate matrix is within the RMIA network capacity region $\Lambda_{\rm{RMIA}}$. With this goal, we propose the following theorem:
\begin{thm}
\label{thm: DIVBAR_MIA_throughput_optimal}
DIVBAR-RMIA is throughput optimal with the RMIA assumption: DIVBAR-RMIA can stably support any exogenous input rate matrix within network capacity region $\Lambda_{\rm{RMIA}}$, i.e., for an exogenous input rate matrix $\left(\lambda_n^{\left(c\right)}\right)$, if $\exists \varepsilon >0$ satisfying $\left(\lambda_n^{\left(c\right)}+\varepsilon\right)\in \Lambda_{\rm{RMIA}}$, then there exists an integer $D>0$, such that the mean time average backlog of the whole network can be upper bounded as follows:
\begin{equation}
\mathop {\lim \sup }\limits_{t \to \infty } \frac{1}{t}\sum\limits_{\tau  = 0}^{t - 1} {\sum\limits_{n,c} {\mathbb{E}\left\{ {\hat Q_n^{\left( c \right)}\left( \tau  \right)} \right\}} }  \le \frac{2\left[B\left(D\right)+C\left(D\right)\right]}{\varepsilon },
\label{eq_strong_stability_hat_main}
\end{equation}
when implementing the DIVBAR-RMIA algorithm, where
\begin{equation}
B\left(D\right)=N^2D\left[1+\left(N+A_{\rm{max}}\right)^2\right];\ C\left(D\right)=4ND\left(N+A_{\rm{max}}+1\right);\nonumber
\end{equation}
$\hat Q_n^{\left(c\right)}\left(\tau\right)$ is the backlog (CPQ backlog) of commodity $c$ in node $n$ in timeslot $\tau$ under the DIVBAR-RMIA algorithm.
\end{thm}

A detailed proof of Theorem \ref{thm: DIVBAR_MIA_throughput_optimal} is given in Appendix \ref{appendix: DIVBAR-RMIA_optimallity}. The proof follows the similar intuitions summarized for the case of stationary randomized policy in Section \ref{sec: network_capacity_region_RMIA} and adopts $d$-timeslot Lyapunov drift to analyze the key metrics under DIVBAR-RMIA and related intermediate policies. As $d$ grows large, on the one hand, the key metrics under the related metrics converge, and on the other hand, the effect of the "marginal" interval due to the arbitrary starting timeslot $t_0$ gradually vanishes.

\subsection{Throughput performance of DIVBAR-MIA}
\label{subsec: DIVBAR-MIA_throughput_performance}
In contrast with the renewal operation under DIVBAR-RMIA, each receiving node under DIVBAR-MIA retains the partial information of each packet until the packet reaches its final destination, and the epochs of DIVBAR-MIA is set to be synchronized with epochs under DIVBAR-RMIA. Therefore, with the help of the pre-accumulated information in the receiving nodes of ${\cal K}_n$, the successful receiver set at the end of each epoch under DIVBAR-MIA should include the first successful receiver set under DIVBAR-RMIA. This intuition indicates that the throughput performance of DIVBAR-MIA should be at least as good as DIVBAR-RMIA. With this intuition, the following theorem is proposed:

\begin{thm}
\label{thm: DIVBAR-RMIA_vs_DIVBAR-MIA}
DIVBAR-MIA can support any input rate matrix within the RMIA network capacity region, i.e., for an exogenous input rate matrix $\left(\lambda_n^{\left(c\right)}\right)$, if $\exists \varepsilon >0$ satisfying $\left(\lambda_n^{\left(c\right)}+\varepsilon\right)\in \Lambda_{\rm{RMIA}}$, then there exists an integer $D>0$, such that the mean time average backlog of the whole network can be upper bounded as follows:
\begin{equation}
\mathop {\lim \sup }\limits_{t \to \infty } \frac{1}{t}\sum\limits_{\tau  = 0}^{t - 1} {\sum\limits_{n,c} {\mathbb{E}\left\{ {\hat {\hat Q}_n^{\left( c \right)}\left( \tau  \right)} \right\}} }  \le \frac{2\left[B\left(D\right)+C\left(D\right)\right]}{\varepsilon },
\label{eq_strong_stability_hat_main}
\end{equation}
when implementing the DIVBAR-MIA algorithm, where
\begin{equation}
B\left(D\right)=N^2D\left[1+\left(N+A_{\rm{max}}\right)^2\right];\ C\left(D\right)=4ND\left(N+A_{\rm{max}}+1\right);\nonumber
\end{equation}
$\hat {\hat Q}_n^{\left(c\right)}\left(\tau\right)$ is the backlog (CPQ backlog) of commodity $c$ in node $n$ in timeslot $\tau$ under the DIVBAR-MIA algorithm.
\end{thm}

The detailed proof of Theorem \ref{thm: DIVBAR-RMIA_vs_DIVBAR-MIA} is given in Appendix \ref{appendix: DIVBAR-MIA_vs_DIVBAR-RMIA}, and the proof strategy is similar to that of Theorem \ref{thm: DIVBAR_MIA_throughput_optimal}.

Based on Theorem \ref{thm: DIVBAR-RMIA_vs_DIVBAR-MIA}, we can further compare DIVBAR-MIA and DIVBAR-RMIA in the throughput performance by showing the following corollary:
\begin{corollary}
The throughput performance of DIVBAR-MIA is at least as good as DIVBAR-RMIA, i.e., DIVBAR-MIA can stably support any input rate matrix that can be supported by DIVBAR-RMIA.
\label{corollary: DIVBAR-MIA_vs_DIVBAR-RMIA}
\end{corollary}
\begin{proof}
Theorem \ref{thm: DIVBAR_MIA_throughput_optimal} shows that any input rate matrix within $\Lambda_{\rm{RMIA}}$ can be stably supported by DIVBAR-RMIA. On the other hand, $\Lambda_{\rm{RMIA}}$ is defined as the set of input rate matrices that can be stably supported by all possible algorithms with RMIA. Combining the two aspects, $\Lambda_{\rm{RMIA}}$ is exactly the set of input rate matrices that can be supported by DIVBAR-RMIA.

According to Theorem \ref{thm: DIVBAR-RMIA_vs_DIVBAR-MIA}, DIVBAR-MIA breaks the RMIA assumption and is shown to be able to support any input rate matrix within $\Lambda_{\rm{RMIA}}$, which indicates that any input rate matrix that can be stably supported by DIVBAR-RMIA can also be stably supported by DIVBAR-MIA, i.e., the throughput performance of DIVBAR-MIA is at least as good as DIVBAR-RMIA
\end{proof}

\section{Simulations}
\label{sec: simulations}
\begin{figure}
        \includegraphics[height=6cm]{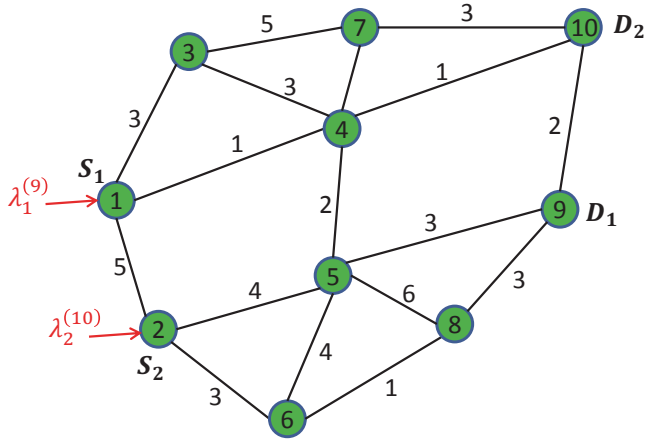}
        \centering{
        \caption{The ad-hoc network being simulated}
        \label{simulation_network}}
\end{figure}

\begin{figure}
        \includegraphics[height=8cm]{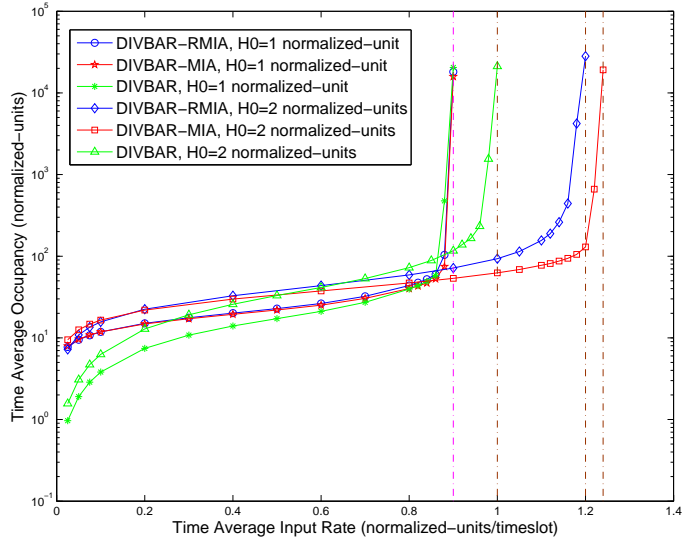}
        \centering{
        \caption{Throughput performance comparison among DIVBAR-MIA, DIVBAR-RMIA and DIVBAR algorithms with different packet lengths}
        \label{throughput_comparison}}
\end{figure}

Example simulations are carried out in the ad-hoc wireless network shown in Fig. \ref{simulation_network}. All the links in the network are independent non-interfering links, each of which is subject to Rayleigh fading (independent among links and timeslots), while the average channel states are static. The number on each link represents the mean SNR value (linear scale) over that link; the time average exogenous input rates $\lambda_1^{\left(9\right)}$ and $\lambda_2^{\left(10\right)}$ are set to be the same.

Simulations are conducted comparing throughput performance of the three algorithms: DIVBAR-MIA, DIVBAR-RMIA, and regular DIVBAR (with REP). Fig. \ref{throughput_comparison} shows the time average occupancy (total time average backlog in the network measured in \emph{normalized-units}) vs. exogenous time average input rate measured in normalized-units/timeslot. Here a normalized-unit has to be long enough (contain sufficient number of bits) to allow the application of a capacity achieving code. The maximum supportable throughput corresponds to the input rate at which the occupancy goes towards very large values (due to a finite number of simulation time, it does not approach infinity in our simulations). As is shown in the figure, the throughput under the DIVBAR-MIA algorithm is generally the largest among the three algorithms; the throughput under DIVBAR-RMIA algorithm is smaller than that of DIVBAR-MIA; the throughput under both algorithms are larger than that of the regular DIVBAR algorithm. These observations are in line with the theoretical analysis.

The simulation of the throughput comparison is carried out under different packet entropy conditions. The entropy contained in each packet is denoted by $H_0$ as is shown in the figure. When $H_0=1$ normalized-unit, Fig. \ref{throughput_comparison} shows that the throughput under the three algorithms are nearly identical. This phenomenon is caused by the fact that the packet length is generally small compared to the transmission ability of the links in the network. Therefore nodes in the network can usually achieve a successful transmission over a link at the first attempt, which results in that (R)MIA has little benefit. However, as $H_0$ increases to 2 normalized-units, the success probability in a single attempt decreases. Nodes under regular DIVBAR increase the chance of successful transmission just through trying more times, while DIVBAR-MIA and DIVBAR-RMIA accumulate information in each attempt, which will facilitate the future transmissions. Thus the throughput difference between DIVBAR and DIVBAR-(R)MIA becomes obvious.

\section{Conclusions}
\label{sec: conclusions}
In this paper, we proposed two distributed routing algorithms: DIVBAR-RMIA and DIVBAR-MIA, which exploit mutual information accumulation technique as the physical layer transmission scheme for the routing in multi-hop, multi-commodity wireless ad-hoc networks with unreliable links. After setting up a proper network model, including designing the queue structure of each network node to implement MIA or RMIA, and the working diagram within each timeslot, we analyzed the network's throughput potential with MIA technique by characterizing and analyzing the network capacity region with RMIA, and proved that the proposed two algorithms have superior throughput performance compared to the original DIVBAR with REP. Simulation results confirmed the throughput performance enhancement.

\section*{Acknowledgements}
The authors would like to thank Prof. Michael J. Neely for many helpful discussions.

\appendices
\section{Policies used in the theoretical analysis}
\label{appendix: policy_list}
\begin{figure}
        \includegraphics[height=15cm]{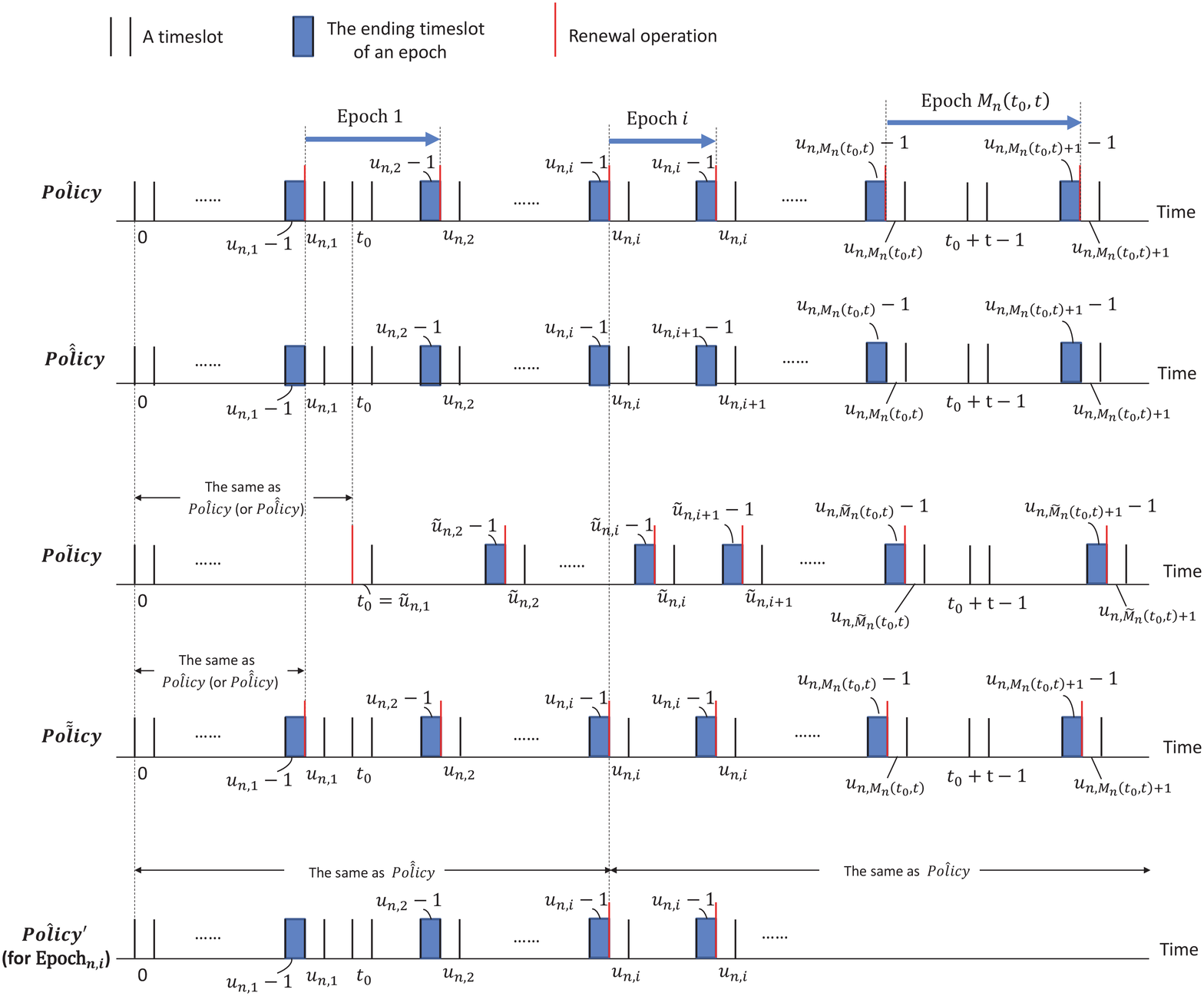}
        \centering{
        \caption{Epochs for each transmitting node $n$ in the network under $\hat{Policy}$, $\hat{\hat{Policy}}$, $\tilde{Policy}$, $\tilde{\tilde{Policy}}$, and $\hat{Policy'}$}
        \label{fig_epoch relations among several policies}}
\end{figure}

\begin{enumerate}[a)]
\item \textbf{$Policy^{**}$}: (firstly defined in Theorem \ref{cap_region_thm}) the stationary randomized policy with REP that can support any exogenous input rate matrix $\left(\lambda_n^{\left(c\right)}\right)$ within $\Lambda_{\rm{REP}}$.
\item \textbf{$Policy^*$}: (firstly defined in Corollary \ref{cap_region_corr}) the stationary randomized policy with RMIA that can support any exogenous input rate matrix $\left({\lambda}_n^{\left(c\right)}\right)$ within $\Lambda_{\rm{RMIA}}$.
\item \textbf{$Policy^1$}: (firstly defined in the proof of Theorem \ref{thm: capacity_comparison}) the stationary randomized policy with RMIA that forms the same flow rate matrix as $Policy^{**}$ with REP.
\item \textbf{$Policy^2$}: (firstly defined in the proof of Theorem \ref{thm: capacity_comparison}) the stationary randomized policy with RMIA that supports the increased input rate matrix $\left({\lambda'}_n^{\left(c\right)}\right)$ being located within $\Lambda_{\rm{RMIA}}$ but outside of $\Lambda_{\rm{REP}}$.
\item \textbf{$\hat{Policy}$}: (firstly defined in the proof of Theorem \ref{thm: DIVBAR_MIA_throughput_optimal}) the DIVBAR-RMIA policy.
\item \textbf{${Policy'}^{*}$}: (firstly defined in the proof of Theorem \ref{thm: DIVBAR_MIA_throughput_optimal} (or Theorem \ref{thm: DIVBAR-RMIA_vs_DIVBAR-MIA})) the intermediate policy with the following properties: it is the same as DIVBAR-RMIA (or DIVBAR-MIA) in the interval from timeslot $0$ to timeslot $t_0-1$; starting from timeslot $t_0$ without using the pre-accumulated partial information, it is the same as $Policy^*$ starting from timeslot $0$.
\item \textbf{$\tilde{Policy}$}: (firstly defined in the proof of Theorem \ref{thm: DIVBAR_MIA_throughput_optimal} (or Theorem \ref{thm: DIVBAR-RMIA_vs_DIVBAR-MIA})) the intermediate policy with the following properties: it is the same as DIVBAR-RMIA (or DIVBAR-MIA) in the interval from timeslot $0$ to timeslot $t_0-1$; starting from timeslot $t_0$ without using the pre-accumulated partial information, each node chooses the commodity to transmit according to the maximization of a backpressure metric, and keeps transmitting the packets of the chosen commodity in later timeslots with RMIA, and forwards each decoded packet to the successful receiver with the largest positive differential backlog observed in timeslot $t_0$.
\item \textbf{$\tilde{Policy^*}$}: (firstly defined in the proof of Theorem \ref{thm: DIVBAR_MIA_throughput_optimal} (or Theorem \ref{thm: DIVBAR-RMIA_vs_DIVBAR-MIA})) the intermediate policy with the following properties:  it is the same as DIVBAR-RMIA (or DIVBAR-MIA) in the interval from timeslot $0$ to timeslot $t_0-1$; starting from timeslot $t_0$ without using the pre-accumulated partial information, each node uses the same probabilities as $Policy^{*}$ to choose commodities to transmit but uses the same strategy as $\tilde{Policy}$ to forward the decoded packets.
\item \textbf{$\tilde{\tilde{Policy}}$}: (firstly defined in the proof of Theorem \ref{thm: DIVBAR_MIA_throughput_optimal} (or Theorem \ref{thm: DIVBAR-RMIA_vs_DIVBAR-MIA})) the intermediate and non-causal policy with the following properties: the epochs for each node have contiguous timeslots; for each node, it is the same as DIVBAR-RMIA (or DIVBAR-MIA) in the interval from timeslot $0$ to $u_{n,1}-1$ ($u_{n,1}$ is the starting timeslot of the epoch for the transmitting node $n$ that includes timeslot $t_0$); starting from timeslot $u_{n,1}$ without using the pre-accumulated partial information, each node chooses the same commodity to transmit as that chosen by $\tilde{Policy}$ in timeslot $t_0$ ($u_{n,1}\le t_0$), and keeps transmitting the packets of the chosen commodity during the later timeslots with RMIA, and forwards each decoded packet to the receiver with the largest differential backlog formed under DIVBAR-RMIA (or DIVBAR-MIA) in timeslot $t_0$.
\item \textbf{$\hat{\hat{Policy}}$}: (firstly defined in the proof of Theorem \ref{thm: DIVBAR-RMIA_vs_DIVBAR-MIA}): the DIVBAR-MIA policy.
\item \textbf{$\hat{Policy'}$}: (firstly defined in the proof of Theorem \ref{thm: DIVBAR-RMIA_vs_DIVBAR-MIA}) the intermediate policy with the property that: for each node $n$, it is the same as DIVBAR-MIA from timeslot $0$ to timeslot $u_{n,i}-1$, where $u_{n,i}$ is the starting timeslot of a particular epoch $i$ for node $n$ under DIVBAR-MIA, while starting from timeslot $u_{n,i}$ without using the pre-accumulated partial information, it is the same as DIVBAR-RMIA.
\end{enumerate}

Here we clarify that the intermediate policies $\tilde{Policy}$, $\tilde{\tilde{Policy}}$, $\tilde{Policy^*}$ are specified by the timeslot $t_0$, which can be arbitrary and is the starting timeslot of the interval for the Lyapunov drift analysis. However, here we omit putting the notation $t_0$ on these policies in order to simplify the notations. Similarly, although $\hat{Policy'}$ is designed for each epoch, the epoch index notation is omitted for the same reason.

Fig. \ref{fig_epoch relations among several policies} shows the epochs for each node $n$ under several policies whose epochs for each transmitting node consist of contiguous timeslots. In Fig. \ref{fig_epoch relations among several policies}, $u_{n,i}$ is the starting timeslot of the $i$th epoch of $\hat{Policy}$, $\hat{\hat{Policy}}$, and $\tilde{\tilde{Policy}}$ for node $n$ counting from the epoch that includes timeslot $t_0$; $\tilde{u}_{n,i}$ is the starting timeslot of the $i$th epoch under $\tilde{Policy}$ for node $n$, where $\tilde{u}_{n,1}=t_0$.

\section{Proof of the necessity part of Corollary \ref{cap_region_corr}}
\label{appendix: necessity_corollary_1}
Most parts of the necessity proof are similar to the proof of Theorem \ref{cap_region_thm} given in Ref. \cite{Neely_Rahul_DIVBAR_2009}. We start from the assumption that, for an input rate matrix $\left( {\lambda _{nk}^{\left( c \right)}} \right)$ in the network capacity region $\Lambda_{\rm{RMIA}}$, there exists a policy that can support it.

Same as in Ref. \cite{Neely_Rahul_DIVBAR_2009}, we define a \emph{unit} as a copy of a packet. Two units are said to be \emph{distinct} if they are copies of different original packets. When a packet is successfully transmitted from one node to another, we say that the original unit is retained in the transmitting node while a copy of the unit is created in the new node. After the forwarding decision of one transmission is made, only one of all the non-distinct units is kept, either in the transmitting node or in one of the successful receiving nodes.

Let $A_n^{\left(c\right)}\left(t\right)$ represent the total number of commodity $c$ distinct units that exogenously arrive at node $n$ during the first $t$ timeslots. Define $Y_n^{\left(c\right)}\left(t\right)$ as the total number of distinct units with source node $n$ and commodity $c$ that are delivered to the destination up to time $t$. Because of the assumption that the policy is \textit{rate stable}, for any node $n$ and commodity $c$, the delivery rate is equal to the input rate:
\begin{equation}
\label{eq_rate_stable}
\mathop {\lim }\limits_{t \to \infty } \frac{{Y_n^{\left( c \right)}\left( t \right)}}{t} = \mathop {\lim }\limits_{t \to \infty } \frac{{A_n^{\left( c \right)}\left( t \right)}}{t} = \lambda _n^{\left( c \right)}\rm{\ with\ prob.\ 1}.
\end{equation}
Let ${\cal U}_j^{\left(c\right)}\left(t\right)$ be the set of distinct units that are the first to reach their destination $c$ from the source node $j$ during the first $t$ timeslots. Define $G_{nk}^{\left(c\right)}\left(t\right)$ to be the total number of units of commodity $c$ within the set $\bigcup\limits_{j \in \cal N} {{\cal U}_j^{\left( c \right)}\left( t \right)}$ that are forwarded from node $n$ to node $k$ within the first $t$ timeslots. Then for node $n$ and commodity $c$, it follows that
\begin{equation}
\label{in_and_out}
Y_n^{\left( c \right)}\left( t \right) + \sum\limits_{k: k \in {{\cal K}_n}} {G_{kn}^{\left( c \right)}\left( t \right)}  = \sum\limits_{k: k \in {{\cal K}_n}} {G_{nk}^{\left( c \right)}\left( t \right)} ,{\rm{\ for\ }}n \ne c.
\end{equation}
Now define the following variables for all nodes $n,k \in \cal N$ and all commodities $c\in \cal N$:
\begin{itemize}
    \item $\alpha _n^{\left( c \right)}\left( {t} \right)$: the number of times node $n$ decides to transmit the units of commodity $c$ from the head of its CPQ during the first $t$ timeslots.
    \item $\beta_{n}^{\rm{rmia},\left(c\right)}\left(t\right)$: the number of epochs of commodity $c$ for node $n$ with RMIA that end during the first $t$ timeslots.
    \item $q_{n,{\Omega _n}}^{\rm{rmia},\left( c \right)}\left( {t} \right)$: the number of units of commodity $c$ sent by node $n$ and successfully received by the set of nodes $\Omega_n$ with RMIA during the first $t$ timeslots.
    \item $\theta _{nk}^{\left( c \right)}\left( {{\Omega _n},t} \right)$: the number of times the units of commodity $c$ within $\bigcup\limits_{j \in \cal N} {{\cal U}_j^{\left( c \right)}\left( t \right)}$ are forwarded from node $n$ to node $k$ during the first $t$ timeslots, given that the successful receiver set is $\Omega_n$.
\end{itemize}
Then we have
\begin{equation}
\label{time_average_relation}
\frac{{G_{nk}^{\left( c \right)}\left( t \right)}}{t} = \frac{{\alpha _n^{\left( c \right)}\left( t \right)}}{t}\frac{{\beta _n^{\rm{rmia},\left( c \right)}\left( t \right)}}{{\alpha _n^{\left( c \right)}\left( t \right)}}\sum\limits_{{\Omega _n}} {\frac{{q_{n,{\Omega _n}}^{\rm{rmia},\left( c \right)}\left( t \right)}}{{\beta _n^{\rm{rmia},\left( c \right)}\left( t \right)}}\frac{{\theta _{nk}^{\left( c \right)}\left( {{\Omega _n},t} \right)}}{{q_{n,{\Omega _n}}^{\rm{rmia},\left( c \right)}\left( t \right)}}} ,
\end{equation}
where we define ${0 \mathord{\left/{\vphantom {0 0}} \right.\kern-\nulldelimiterspace} 0}\buildrel \Delta \over = 0$ for terms on the right hand side of the above equation, and the definition is also effective for the similar forms of formulas in other parts of the paper. Note firstly that for all $t$, we have:
\begin{equation}
\label{omega_theta_ratio_bound}
0 \le \frac{{\alpha _n^{\left( c \right)}\left( t \right)}}{t} \le 1,\;0 \le \frac{{\theta _{nk}^{\left( c \right)}\left( {{\Omega _n},t} \right)}}{{q_{n,{\Omega _n}}^{\rm{rmia},\left( c \right)}\left( t \right)}} \le 1.
\end{equation}
Moreover, we mainly need to show that $\mathop {\lim }\limits_{t \to \infty } {{\beta _n^{\rm{rmia},\left( c \right)}\left( t \right)} \mathord{\left/
 {\vphantom {{\beta _n^{\left( c \right)}\left( t \right)} {\alpha _n^{\left( c \right)}\left( t \right)}}} \right.
 \kern-\nulldelimiterspace} {\alpha _n^{\left( c \right)}\left( t \right)}}$ and $\mathop {\lim }\limits_{t \to \infty } {{q_{n,{\Omega _n}}^{\rm{rmia},\left( c \right)}\left( t \right)} \mathord{\left/{\vphantom {{q_{n,{\Omega _n}}^{\rm{rmia},\left( c \right)}\left( t \right)} {\beta _n^{\rm{rmia},\left( c \right)}\left( t \right)}}} \right.\kern-\nulldelimiterspace} {\beta _n^{\rm{rmia},\left( c \right)}\left( t \right)}}$ exist and are well-defined.

Let $T_n^{\left( c \right)}\left( i \right)$ represent the epoch length of the $i$th unit of commodity $c$. First note that the expectation of $T_n^{\left( c \right)}\left( i \right)$ exists because
\begin{align}
\mathbb{E}\left\{ {{T_n^{\left(c\right)}}\left( i \right)} \right\} &= \sum\limits_{m = 1}^\infty  {\Pr \left\{ {{T_n^{\left(c\right)}}\left( i \right) \ge m} \right\}} \nonumber\\
&= \sum\limits_{m = 1}^\infty  {\prod\limits_{j \in {{\cal K}_n}} {F_{nj}^{\left( {m - 1} \right)}\left( {{H_0}} \right)} } \nonumber\\
&\buildrel (a) \over < \sum\limits_{m = 1}^\infty  {\prod\limits_{j \in {{\cal K}_n}} {{{\left[ {{F_{{R_{nj}}}}\left( {{H_0}} \right)} \right]}^{m - 1}}} } = \frac{1}{{\prod\limits_{j \in {{\cal K}_n}} {\left[ {1 - {F_{{R_{nj}}}}\left( {{H_0}} \right)} \right]} }} < \infty ,
\label{eq_finite_epoch_length}
\end{align}
where the step (a) is due to (\ref{eq_cdf_information_perslot2}) in Lemma \ref{lemma: flowing_rate_property}.

Since ${\beta _n^{\rm{rmia},\left( c \right)}\left( t \right)}$ units have been transmitted during the ${\alpha_n^{\left( c \right)}\left( t \right)}$ timeslots, we have
\begin{equation}
\sum\limits_{i = 1}^{\beta _n^{\rm{rmia},\left( c \right)}\left( t \right)} {T_n^{\left( c \right)}\left( i \right)}  \le \alpha _n^{\left( c \right)}\left( t \right) < \sum\limits_{i = 1}^{\beta _n^{\rm{rmia},\left( c \right)}\left( t \right) + 1} {T_n^{\left( c \right)}\left( i \right)}.
\end{equation}
Then it follows that
\begin{equation}
\label{beta_ratio_omega_limit_proof}
\frac{1}{{\frac{1}{{\beta _n^{\rm{rmia},\left( c \right)}\left( t \right)}}\sum\limits_{i = 1}^{\beta _n^{\rm{rmia},\left( c \right)}\left( t \right) + 1} {T_n^{\left( c \right)}\left( i \right)} }} < \frac{{\beta _n^{\rm{rmia},\left( c \right)}\left( t \right)}}{{\alpha _n^{\left( c \right)}\left( t \right)}} \le \frac{1}{{\frac{1}{{\beta _n^{\rm{rmia},\left( c \right)}\left( t \right)}}\sum\limits_{i = 1}^{\beta _n^{\rm{rmia},\left( c \right)}\left( t \right)} {T_n^{\left( c \right)}\left( i \right)} }}.
\end{equation}
Because of the renewal operation at the end of each epoch, $\left\{{T_n^{\left( c \right)}\left( i \right)}:i\ge 1\right\}$ are i.i.d. over epochs and $\mathbb{E}\left\{ {T_n^{\left( c \right)}\left( i \right)} \right\}$ is well defined as is shown in (\ref{eq_finite_epoch_length}). According to the law of large numbers, the denominators of the lower and upper bounds in (\ref{beta_ratio_omega_limit_proof}) therefore approach $\mathbb{E}\left\{ {T_n^{\left( c \right)}\left( i \right)} \right\}$ with probability 1 as $t\rightarrow \infty $. Then it follows that
\begin{equation}
\label{beta_ratio_omega_limit}
\mathop {\lim }\limits_{t \to \infty } \frac{{\beta _n^{\rm{rmia},\left( c \right)}\left( t \right)}}{{\alpha _n^{\left( c \right)}\left( t \right)}} = \frac{1}{{\mathbb{E}\left\{ {T_n^{\left( c \right)}\left( i \right)} \right\}}}\buildrel \Delta \over =\beta _n^{\rm{rmia}}{\rm{\ with\ prob.\ }}1.
\end{equation}
Here the notation $\beta _n^{\rm{rmia}}$ does not have the superscript $c$ because $\mathbb{E}\left\{ {T_n^{\left( c \right)}\left( i \right)} \right\}$ is the same for all commodities. $\beta _n^{\rm{rmia}}$ can be intuitively interpreted as the "average frequency" of the epochs for the RMIA transmission scheme.

Moreover, for ${{q_{n,{\Omega _n}}^{\rm{rmia},\left( c \right)}\left( t \right)} \mathord{\left/{\vphantom {{q_{n,{\Omega _n}}^{\rm{rmia},\left( c \right)}\left( t \right)} {\beta _n^{\rm{rmia},\left( c \right)}\left( t \right)}}} \right.\kern-\nulldelimiterspace} {\beta _n^{\rm{rmia},\left( c \right)}\left( t \right)}}$, with the renewal operation, the decoding sets for each transmitting node $n$ across epochs are identically and independently distributed. Then we get by the strong law of large numbers that
\begin{equation}
\label{eq_q_ratio_limit}
\mathop {\lim }\limits_{t \to \infty } \frac{{q_{n,{\Omega _n}}^{\rm{rmia},\left( c \right)}\left( t \right)}}{{\beta _n^{\rm{rmia},\left( c \right)}\left( t \right)}} = q_{n,{\Omega _n}}^{{\rm{rmia}}}{\rm{\ with\ prob.\ }}1,
\end{equation}
where ${q_{n,{\Omega _n}}^{{\rm{rmia}}}}$ is the probability that $\Omega_n$ is the first successful receiver set at the end of each epoch. ${q_{n,{\Omega _n}}^{{\rm{rmia}}}}$ is the same for all commodities.

Define $g_{nk}^{\left( c \right)}\left( t \right) \buildrel \Delta \over = {{G_{nk}^{\left( c \right)}\left( t \right)} \mathord{\left/{\vphantom {{G_{nk}^{\left( c \right)}\left( t \right)} t}} \right.\kern-\nulldelimiterspace} t}$, and note that:
\begin{equation}
\label{flow_variable_g_bounded}
0 \le g_{nk}^{\left( c \right)}\left( t \right) \le 1,{\rm{\ }}g_{cn}^{\left( c \right)}\left( t \right) = g_{nn}^{\left( c \right)}\left( t \right) = 0.
\end{equation}

Since the constraints defined in (\ref{omega_theta_ratio_bound}) show closed and bounded regions with finite dimensions, a subsequence $t_i$ must exists, over which the individual terms in (\ref{omega_theta_ratio_bound}) and (\ref{flow_variable_g_bounded}) converge to constants $\alpha _n^{*\left( c \right)}$ and $\theta _{nk}^{*\left( c \right)}\left( {{\Omega _n}} \right)$. Furthermore, suppose there is a stationary randomized policy that decides to transmit a commodity $c$ unit every timeslot with a fixed probability $\alpha_n^{*\left( c \right)}$ and chooses node $k$ to get the forwarding responsibility with a fixed conditional probability $\theta _{nk}^{*\left( c \right)}\left( {{\Omega _n}} \right)$, given that the set of nodes $\Omega_n$ successfully received the packet. In this case, according to the law of large numbers, the values $\alpha_n^{*\left( c \right)}$ and $\theta _{nk}^{*\left( c \right)}\left( {{\Omega _n}} \right)$ are the limit values over the whole timeslot sequence $\left\{t\right\}$, i.e., the converging subsequence $\left\{t_i\right\}$ becomes $\left\{t\right\}$, and therefore it follows that
\begin{equation}
\label{omega_limit}
\mathop {\lim }\limits_{{t} \to \infty } \frac{{\alpha _n^{\left( c \right)}\left( {{t}} \right)}}{{{t}}} = \alpha _n^{*\left( c \right)}\rm{\ with\ prob.\ 1},
\end{equation}
\begin{equation}
\label{theta_limit}
\mathop {\lim }\limits_{{t} \to \infty } \frac{{\theta _{nk}^{\left( c \right)}\left( {{\Omega _n},{t}} \right)}}{{q_{n,{\Omega _n}}^{\rm{rmia},\left( c \right)}\left( {{t}} \right)}} = \theta _{nk}^{*\left( c \right)}\left( {{\Omega _n}} \right)\rm{\ with\ prob.\ 1},
\end{equation}
\begin{equation}
\label{eq_g_limit}
\mathop {\lim }\limits_{{t} \to \infty } g_{nk}^{\left( c \right)}\left( {{t}} \right) = b_{nk}^{*\left( c \right)}\rm{\ with\ prob.\ 1}.
\end{equation}
Combining (\ref{flow_variable_g_bounded}) and (\ref{eq_g_limit}), we have
\begin{equation}
b_{nk}^{*\left( c \right)} \ge 0,b_{cn}^{*\left( c \right)} = 0,b_{nn}^{*\left( c \right)} = 0,{\rm{\ for\ }}n \ne c.
\end{equation}
Furthermore, dividing both sides of (\ref{in_and_out}) by $t$ and using the results of (\ref{eq_rate_stable}) and (\ref{eq_g_limit}) yields:
\begin{equation}
\lambda _n^{\left( c \right)} + \sum\limits_{k \in {{\cal K}_n}} {b_{kn}^{*\left( c \right)}}  = \sum\limits_{k \in {{\cal K}_n}} {b_{nk}^{*\left( c \right)}} ,{\rm{\ for\ }}n \ne c.
\end{equation}
Likewise, according to (\ref{beta_ratio_omega_limit}), (\ref{eq_q_ratio_limit}), and (\ref{omega_limit})-(\ref{eq_g_limit}), taking the limit $t\rightarrow\infty$ in (\ref{time_average_relation}) yields:
\begin{equation}
\label{eq_average_rate_with_s}
b_{nk}^{*\left( c \right)} = \alpha _n^{*\left( c \right)}\beta _n^{{\rm{rmia}}}\sum\limits_{\Omega_n: {\Omega _n} \subseteq {{\cal K}_n}} {q_{n,{\Omega _n}}^{{\rm{rmia}}}\theta _{nk}^{*\left( c \right)}\left( {{\Omega _n}} \right)}.
\end{equation}

Thus, the RMIA network capacity region guarantees the existence of a stationary randomized policy with fixed probabilities $\alpha _n^{*\left( c \right)}$ and ${\theta _{nk}^{*\left( c \right)}\left( {{\Omega _n}} \right)}$, such that, for $\forall n,k,c \in \cal N$, it satisfies the constraints shown as (\ref{cap_region_const1_coll})-(\ref{cap_region_const3_coll}).

\section{Proof of the sufficiency part of Corollary \ref{cap_region_corr}}
\label{appendix: proof_sufficiency_corollary1}
To prepare for the proof, we firstly propose the following two lemmas:
\begin{lemma}
\label{lemma: uniformly convergence}
With RMIA transmission scheme, $\forall \varepsilon  > 0$, there exists an integer $D_{nk}^{\left(c\right)}>0$ such that, for $\forall t_0\ge 0$ and $\forall t \ge D_{nk}^{\left(c\right)}>0$ ($t_0$ and $t$ are integers), the mean time average flow rate of commodity $c$ through link $\left(n,k\right)$ over the interval from timeslot $t_0$ to timeslot $t_0+t-1$ satisfies:
\begin{equation}
\label{general_time_average}
\left| {\frac{1}{t}\sum\limits_{\tau  = {t_0}}^{{t_0} + t - 1} {\mathbb{E}\left\{ {b_{nk}^{\left( c \right)}\left( \tau  \right)} \right\}}  - b_{nk}^{\left( c \right)}} \right| \le \varepsilon.
\end{equation}
\end{lemma}

The proof of Lemma \ref{lemma: uniformly convergence} is shown in Appendix \ref{appendix: uniformly_convergence} and is non-trivial due to the requirement that the value of $D_{nk}^{\left(c\right)}$ does not depend on $t_0$, which is an arbitrary timeslot index.

\begin{lemma}
\label{lemma: strong_stability}
If there exists a constant $\varepsilon>0$ and integer $d>0$ such that for each timeslot $t_0$ and the backlogs ${\bf{Q}}\left(t_0\right)$ of the network in timeslot $t_0$, the $d$-timeslot average Lyapunov drift satisfies:
\begin{equation}
\label{eq_D_step_lyapunov_drift}
\frac{1}{d}\sum\limits_{n,c} {\mathbb{E}\left\{ {\left. {{{\left( {Q_n^{\left( c \right)}\left( {{t_0} + d} \right)} \right)}^2} - {{\left( {Q_n^{\left( c \right)}\left( {{t_0}} \right)} \right)}^2}} \right|\;{\bf{Q}}\left( {{t_0}} \right)} \right\}}  \le B\left(d\right) - \varepsilon \sum\limits_{n,c} {Q_n^{\left( c \right)}\left( {{t_0}} \right)},
\end{equation}
where $B\left(d\right)$ is a constant related to $d$. Then the mean time average backlog of the whole network satisfies:
\begin{equation}
\label{eq_mean_time_average_backlog_bounded}
\mathop {\lim \sup }\limits_{t \to \infty } \frac{1}{t}\sum\limits_{\tau  = 0}^{t - 1} {\sum\limits_{n,c} {\mathbb{E}\left\{ {Q_n^{\left( c \right)}\left( \tau  \right)} \right\}} }  \le \frac{B\left(d\right)}{\varepsilon }.
\end{equation}
\end{lemma}

The proof of Lemma \ref{lemma: strong_stability} is shown in Appendix \ref{appendix: stong_stability}.

With Lemma \ref{lemma: uniformly convergence} and Lemma \ref{lemma: strong_stability}, the main strategy of proving the sufficiency part of Corollary \ref{cap_region_corr} is to show that the network can achieve strong stability with the input rate matrix $\left(\lambda_n^{\left(c\right)}\right)$, if there exists a stationary randomized policy and a constant $\varepsilon>0$, such that the flow rates under the stationary randomized policy, denoted as $Policy^*$, and the input rate matrix $\left(\lambda_n^{\left(c\right)}+\varepsilon\right)$ satisfy (\ref{cap_region_const2_coll}), i.e.,
\begin{equation}
\label{cap_region_const4_coll}
\sum\limits_{k\in {\cal K}_n} {b_{kn}^{*\left( c \right)}}  + \left(\lambda _n^{\left( c \right)} + \varepsilon\right) \le \sum\limits_{k\in {\cal K}_n} {b_{nk}^{*\left( c \right)}} ,{\rm{\ for}}\ n \neq c.
\end{equation}

Start by extending one-timeslot queueing dynamic shown as (\ref{eq_basic_queuing_dynamics}) to $t$-timeslot queueing dynamic under $Policy^*$:
\begin{equation}
\label{eq_new_D_step_queuing_dynamics}
Q_n^{*\left( c \right)}\left( {{t_0} + t} \right) \le \max \left\{ {Q_n^{*\left( c \right)}\left( {{t_0}} \right) - \sum\limits_{\tau  = {t_0}}^{{t_0} + t - 1} {\sum\limits_{k \in {{\cal K}_n}} {b_{nk}^{\left( *c \right)}\left( \tau  \right)} } } , 0\right\} + \sum\limits_{\tau  = {t_0}}^{{t_0} + t - 1} {\sum\limits_{k \in {{\cal K}_n}} {b_{kn}^{*\left( c \right)}\left( \tau  \right)} }  + \sum\limits_{\tau  = {t_0}}^{{t_0} + t - 1} {a_n^{\left( c \right)}\left( \tau  \right)},
\end{equation}
where $t_0\ge0$; $t\ge1$. Here we claim that, through the whole paper, $t_0$ and $t$ are always used as timeslot index and the number of timeslots, respectively, and therefore always take integer values. Square both sides of (\ref{eq_new_D_step_queuing_dynamics}) and sum it over $n,c\in\left\{1,\cdots,N\right\}$ and divide it by $t$ to get
\begin{align}
\label{eq_square_inequality1}
\frac{{\sum\limits_{n,c} {{{\left( {Q_n^{*\left( c \right)}\left( {{t_0} + t} \right)} \right)}^2}} }}{t} \le &\frac{{\sum\limits_{n,c} {{{\left( {Q_n^{*\left( c \right)}\left( {{t_0}} \right)} \right)}^2}} }}{t} \nonumber\\
&+ \frac{1}{t}\sum\limits_{n,c} {\left[ {{{\left[ {\sum\limits_{\tau  = {t_0}}^{{t_0} + t - 1} {\sum\limits_{k \in {K_n}} {b_{nk}^{*\left( c \right)}\left( \tau  \right)} } } \right]}^2} + {{\left[ {\sum\limits_{\tau  = {t_0}}^{{t_0} + t - 1} {\sum\limits_{k \in {K_n}} {b_{kn}^{*\left( c \right)}\left( \tau  \right)} }  + \sum\limits_{\tau  = {t_0}}^{{t_0} + t - 1} {a_n^{\left( c \right)}\left( \tau  \right)} } \right]}^2}} \right]}\nonumber\\
 &- 2\sum\limits_{n,c} {Q_n^{*\left( c \right)}\left( {{t_0}} \right)\left[ {\frac{1}{t}\sum\limits_{\tau  = {t_0}}^{{t_0} + t - 1} {\sum\limits_{k \in {{\cal K}_n}} {b_{nk}^{*\left( c \right)}\left( \tau  \right)} }  - \frac{1}{t}\sum\limits_{\tau  = {t_0}}^{{t_0} + t - 1} {\sum\limits_{k \in {{\cal K}_n}} {b_{kn}^{*\left( c \right)}\left( \tau  \right)} }  - \frac{1}{t}\sum\limits_{\tau  = {t_0}}^{{t_0} + t - 1} {a_n^{\left( c \right)}\left( \tau  \right)} } \right]}.
\end{align}
Under an arbitrary policy, $\sum\limits_{k \in {{\cal K}_n}} {b_{nk}^{\left( c \right)}\left( \tau  \right)}  \le 1$ and $a_n^{\left( c \right)}\left( \tau  \right) \le {A_{\max }}$, and therefore we can bound the second and third square terms on the right hand side of the above inequality, which is also valid for any other policy:
\begin{align}
&\frac{1}{t}\sum\limits_{n,c} {\left[ {{{\left[ {\sum\limits_{\tau  = {t_0}}^{{t_0} + t - 1} {\sum\limits_{k \in {K_n}} {b_{nk}^{*\left( c \right)}\left( \tau  \right)} } } \right]}^2} + {{\left[ {\sum\limits_{\tau  = {t_0}}^{{t_0} + t - 1} {\sum\limits_{k \in {K_n}} {b_{kn}^{*\left( c \right)}\left( \tau  \right)} }  + \sum\limits_{\tau  = {t_0}}^{{t_0} + t - 1} {a_n^{\left( c \right)}\left( \tau  \right)} } \right]}^2}} \right]} \le {N^2}t\left[ {1 + {{\left( {N + {A_{\max }}} \right)}^2}} \right] \buildrel \Delta \over = B\left(t\right).
\label{eq_B_value}
\end{align}
Then take conditional expectations on both sides of (\ref{eq_square_inequality1}), given the backlog state observation ${\bf Q}^{*}\left(t_0\right)$, to get
\begin{align}
\label{eq_corrl_lyapunov_drift1}
&\frac{1}{t}\sum\limits_{n,c} {\mathbb{E}\left\{ {\left. {{{\left( {Q_n^{*\left( c \right)}\left( {{t_0} + t} \right)} \right)}^2} - {{\left( {Q_n^{*\left( c \right)}\left( {{t_0}} \right)} \right)}^2}} \right|\;{{\bf{Q}}^*}\left( {{t_0}} \right)} \right\}}  \le B\left( t \right)\nonumber\\
&\ \ \ \ -2\sum\limits_{n,c} {Q_n^{*\left( c \right)}\left( {{t_0}} \right)\mathbb{E}\left\{ {\left. {\sum\limits_{k \in {{\cal K}_n}} {\frac{1}{t}\sum\limits_{\tau  = {t_0}}^{{t_0} + t - 1} {b_{nk}^{*\left( c \right)}\left( \tau  \right)} }  - \sum\limits_{k \in {{\cal K}_n}} {\frac{1}{t}\sum\limits_{\tau  = {t_0}}^{{t_0} + t - 1} {b_{nk}^{*\left( c \right)}\left( \tau  \right)} }  - \frac{1}{t}\sum\limits_{\tau  = {t_0}}^{{t_0} + t - 1} {a_n^{\left( c \right)}\left( \tau  \right)} } \right|{{\bf{Q}}^*}\left( {{t_0}} \right)} \right\}}.
\end{align}
On the right hand side of (\ref{eq_corrl_lyapunov_drift1}), the expectation term can be rewritten as follows:
\begin{align}
\label{eq_rewrite_average_rate_term}
&\ \ \ \ {\mathbb{E}\left\{ {\left. {\sum\limits_{k \in {{\cal K}_n}} {\frac{1}{t}\sum\limits_{\tau  = {t_0}}^{{t_0} + t - 1} {b_{nk}^{*\left( c \right)}\left( \tau  \right)} }  - \sum\limits_{k \in {{\cal K}_n}} {\frac{1}{t}\sum\limits_{\tau  = {t_0}}^{{t_0} + t - 1} {b_{nk}^{*\left( c \right)}\left( \tau  \right)} }  - \frac{1}{t}\sum\limits_{\tau  = {t_0}}^{{t_0} + t - 1} {a_n^{\left( c \right)}\left( \tau  \right)} } \right|{{\bf{Q}}^*}\left( {{t_0}} \right)} \right\}}\nonumber\\
&=\sum\limits_{k \in {\mathcal{K}_n}} {\left[ {\frac{1}{t}\sum\limits_{\tau  = {t_0}}^{{t_0} + t - 1} {\mathbb{E}\left\{ {b_{nk}^{\left( c \right)}\left( \tau  \right)} \right\}}  - b_{nk}^{\left( c \right)}} \right]}  - \sum\limits_{k \in {\mathcal{K}_n}} {\left[ {\frac{1}{t}\sum\limits_{\tau  = {t_0}}^{{t_0} + t - 1} {\mathbb{E}\left\{ {b_{kn}^{\left( c \right)}\left( \tau  \right)} \right\}}  - b_{kn}^{\left( c \right)}} \right]}\nonumber\\
&\ \ \ + \left[ {\sum\limits_{k \in {\mathcal{K}_n}} {b_{nk}^{\left( c \right)}}  - \sum\limits_{k \in {\mathcal{K}_n}} {b_{kn}^{\left( c \right)}}  - \lambda _n^{\left( c \right)}} \right] + \left[ {\lambda _n^{\left( c \right)} - \frac{1}{t}\sum\limits_{\tau  = {t_0}}^{{t_0} + t - 1} {\mathbb{E}\left\{ {a_n^{\left( c \right)}\left( \tau  \right)} \right\}} } \right],
\end{align}
where ${\bf{Q}}^*\left(t_0\right)$ is dropped from the given condition of the expectation because the decisions made under a stationary randomized policy are independent of the backlog observations. To further manipulate (\ref{eq_rewrite_average_rate_term}), due to Lemma \ref{lemma: uniformly convergence}, for link $\left(n,k\right)$, there exists an integer $D_{nk}^{\left(c\right)}>0$, such that, $\forall t\ge D_{nk}^{\left(c\right)}$ and $\forall t_0\ge 0$,
\begin{equation}
\label{eq_time_average_flow rate_bound}
\left| {\frac{1}{t}\sum\limits_{\tau  = {t_0}}^{{t_0} + t - 1} {\mathbb{E}\left\{ {b_{nk}^{\left( c \right)}\left( \tau  \right)} \right\}}  - b_{nk}^{\left( c \right)}} \right| \le \frac{\varepsilon }{{8N}}.
\end{equation}
Moreover, because $a_{n}^{\left(c\right)}\left(t\right)$ are i.i.d and bounded over timeslots, according to the law of large numbers and dominated convergence property (see Ref. \cite{domenate_convergence_Wong_Hajek_1985}, Chapter 1), there exists an integer $D_{a,n}^{\left(c\right)}$ such that, $\forall t\ge D_{a,n}^{\left(c\right)}$ and $\forall t_0\ge 0$,
\begin{equation}
\label{eq_time_average_input_rate_bound}
\left| {\frac{1}{t}\sum\limits_{\tau  = {t_0}}^{{t_0} + t - 1} {\mathbb{E}\left\{ {a_{nk}^{\left( c \right)}\left( \tau  \right)} \right\}}  - \lambda _n^{\left( c \right)}} \right| \le \frac{\varepsilon }{4}.
\end{equation}
Then choosing $D^* = \max \left\{ {D_{nk}^{\left( c \right)},D_{kn}^{\left( c \right)},D_{a,n}^{\left( c \right)}:\ n,k,c \in \cal N} \right\}$, and plugging (\ref{cap_region_const4_coll}), (\ref{eq_time_average_flow rate_bound}) and (\ref{eq_time_average_input_rate_bound}) into (\ref{eq_rewrite_average_rate_term}), it follows that, $\forall t\ge D^*$ and $\forall t_0\ge 0$,
\begin{align}
\label{eq_bound_average_rate_term}
&\ \ \ {\mathbb{E}\left\{ {\left. {\sum\limits_{k \in {{\cal K}_n}} {\frac{1}{t}\sum\limits_{\tau  = {t_0}}^{{t_0} + t - 1} {b_{nk}^{*\left( c \right)}\left( \tau  \right)} }  - \sum\limits_{k \in {{\cal K}_n}} {\frac{1}{t}\sum\limits_{\tau  = {t_0}}^{{t_0} + t - 1} {b_{nk}^{*\left( c \right)}\left( \tau  \right)} }  - \frac{1}{t}\sum\limits_{\tau  = {t_0}}^{{t_0} + t - 1} {a_n^{\left( c \right)}\left( \tau  \right)} } \right|{{\bf{Q}}^*}\left( {{t_0}} \right)} \right\}}\nonumber\\
&\ge  - \frac{\varepsilon }{{8N}}N - \frac{\varepsilon }{{8N}}N + \varepsilon  - \frac{\varepsilon }{{4}} = \frac{\varepsilon }{2}.
\end{align}
Plugging (\ref{eq_bound_average_rate_term}) back into (\ref{eq_corrl_lyapunov_drift1}), we get that there is an integer $D^*>0$ such that, $\forall t_0 \ge 0$ and $\forall t\ge D^*$,
\begin{equation}
\frac{1}{t}\sum\limits_{n,c} {\mathbb{E}\left\{ {\left. {{{\left( {Q_n^{*\left( c \right)}\left( {{t_0} + t} \right)} \right)}^2} - {{\left( {Q_n^{*\left( c \right)}\left( {{t_0}} \right)} \right)}^2}} \right|\;{{\bf{Q}}^*}\left( {{t_0}} \right)} \right\}}  \le B\left( t \right)-\varepsilon\sum\limits_{n,c} {Q_n^{*\left( c \right)}\left( {{t_0}} \right)}.
\end{equation}
When $t=D^*$, the $D^*$ slots Lyapunov drift is upper bounded as follows:
\begin{equation}
\label{eq_corr1_lyapunov_drift}
\frac{1}{D^*}\sum\limits_{n,c} {\mathbb{E}\left\{ {\left. {{{\left( {Q_n^{*\left( c \right)}\left( {{t_0} + D^*} \right)} \right)}^2} - {{\left( {Q_n^{*\left( c \right)}\left( {{t_0}} \right)} \right)}^2}} \right|\;{{\bf{Q}}^*}\left( {{t_0}} \right)} \right\}}  \le B\left( D^* \right)-\varepsilon\sum\limits_{n,c} {Q_n^{*\left( c \right)}\left( {{t_0}} \right)}.
\end{equation}

The $D^*$-timeslot average Lyapunov drift bound shown on the right hand side of (\ref{eq_corr1_lyapunov_drift}) satisfies the condition required by Lemma \ref{lemma: strong_stability}, and therefore we can conclude that
\begin{equation}
\label{eq_mean_time_average_backlog_bounded_policy_star}
\mathop {\lim \sup }\limits_{t \to \infty } \frac{1}{t}\sum\limits_{\tau  = 0}^{t - 1} {\sum\limits_{n,c} {\mathbb{E}\left\{ {Q_n^{\left( c \right)}\left( \tau  \right)} \right\}} }  \le \frac{B\left(D^*\right)}{\varepsilon },
\end{equation}
which shows \emph{strong stability}.

\section{Proof of Theorem \ref{thm: capacity_comparison}}
\label{appendix: capacity comparison}
The proof of this theorem includes two aspects: first, all the input rate matrices in $\Lambda_{\rm{REP}}$ are included in $\Lambda_{\rm{RMIA}}$, which is shown in Subsection \ref{subsec_capacity_comparison_sub1}; second, there exists at least one input rate matrix that is in $\Lambda_{\rm{RMIA}}$ but is outside of $\Lambda_{\rm{REP}}$, which is shown in Subsection \ref{subsec_capacity_comparison_sub2}.

\subsection{All input rate matrices within $\Lambda_{\rm{REP}}$ are included in $\Lambda_{\rm{RMIA}}$}
\label{subsec_capacity_comparison_sub1}
For this part of the proof, the goal is to show that, for any input rate matrix $\left(\lambda_n^{\left(c\right)}\right)$ within $\Lambda_{\rm{REP}}$, there must exist a stationary randomized policy with RMIA that can stably support $\left(\lambda_n^{\left(c\right)}\right)$. Thus, if we can find a stationary randomized policy with RMIA that can also form the time average flow rate matrix $\left(b_{nk}^{**\left(c\right)}\right)$, then it naturally follows that $\left(\lambda_n^{\left(c\right)}\right)$ can be supported. To be specific, we claim as follows:

\emph{Let $Policy^{**}$ represent the stationary randomized policy with REP that can support any input rate matrix within $\Lambda_{\rm{REP}}$ and forms a flow rate matrix $\left(b_{nk}^{**\left(c\right)}\right)$, then there must exist another stationary randomized policy with RMIA, denoted as $Policy^{1}$, such that the flow rate matrix $\left(b_{nk}^{**\left(c\right)}\right)$ can also be formed. Here node $n$ under $Policy^{**}$ with REP uses probability $\alpha_n^{**\left(c\right)}$ to choose commodity $c$ to transmit and probability $\theta_{nk}^{**\left(c\right)}\left(\Omega_n\right)$ to choose node $k$ as the forwarding node, given the successful receiver set $\Omega_n$; $Policy^1$ also uses probability $\alpha_n^{**\left(c\right)}$ to choose commodity $c$ to transmit but uses $\theta_{nk}^{1\left(c\right)}\left(\Omega_n\right)$ to choose node $k$ as the forwarding node, given the successful receiver set $\Omega_n$, where
\begin{equation}
\theta _{nk}^{1\left( c \right)}\left( {{\Omega _n}} \right) = \sum\limits_{{\Psi _n}:{\Psi _n} \subseteq {\Omega _n},k \in {\Psi _n}} {\frac{{q_{n,{\Psi _n},{\Omega _n}}^{{\rm{rep}},{\rm{rmia}}}}}{{q_{n,{\Omega _n}}^{{\rm{rmia}}}}}\theta _{nk}^{**\left( c \right)}\left( {{\Psi _n}} \right)},\ \Omega_n \neq \emptyset.
\label{eq_forwarding_prob_policy1}
\end{equation}
In (\ref{eq_forwarding_prob_policy1}), $q_{n,\Omega_n}^{\rm{rmia}}$ represents the probability that the first successful receiver set is $\Omega_n$ at the end of each epoch when node $n$ is transmitting with RMIA; $q_{n,\Psi_n,\Omega_n}^{\rm{rep,rmia}}$ represents the probability that, when node $n$ is transmitting, $\Omega_n$ is the first successful receiver set at the ending timeslot of each epoch if with RMIA, while in the same timeslot, the set of nodes $\Psi_n$ (possibly empty) decode the unit if with REP.}

Now we start proving the claim above. Since both $Policy^*$ and $Policy^1$ use the same probabilities $\left\{\alpha_n^{**\left(c\right)}\right\}$ to choose commodities to transmit, the ultimate goal of showing the existence of $Policy^1$ with RMIA relies on whether we can find a set of probability values $\left\{\theta_{nk}^{1\left(c\right)}\left(\Omega_n\right)\right\}$ used to choose forwarding nodes under RMIA, such that, combined with $\left\{\alpha_n^{**\left(c\right)}\right\}$, the flow matrix $\left(b_{nk}^{**\left(c\right)}\right)$ can be formed.

With the same probabilities $\left\{\alpha_n^{**\left(c\right)}\right\}$ used to choose commodities to transmit by $Policy^*$ and $Policy^1$, we can make a further assumption that, for each commodity $c$, $Policy^1$ uses exactly the same set of timeslots to transmit commodity $c$ units as $Policy^{**}$. Admittedly, this assumption seems to be stronger than the assumption that the two policies just use the same probabilities $\left\{\alpha_n^{**\left(c\right)}\right\}$. However, since we only aim to show the existence of $Policy^1$ based on the existence of $Policy^{**}$, this stronger assumption does not matter. In fact, under the constraint that the two policies are both stationary randomized, making decision on which commodity to transmit at the beginning of each timeslot does not depend on the physical layer transmission scheme used, and therefore, we can imagine that the step of choosing commodities to transmit can be firstly done without specification of the two policies. After that, node $n$ starts transmitting and forwarding the packet with the chosen commodity respectively according to $Policy^{**}$ or $Policy^1$.

With the assumption of the same set of timeslots used for transmitting each commodity under $Policy^{**}$ and $Policy^1$, on the one hand, we can guarantee that the channel realizations in the set of timeslots under $Policy^{**}$ and $Policy^1$ used for transmitting each commodity packets are exactly the same. On the other hand, with RMIA, each receiving node in each timeslot not only uses the information transmitted in the current timelsot to decode the packet being transmitted, but also uses the partial information possibly accumulated during previous timeslots as well, while with REP, each receiving node can not use this "inheritance" from the previous timeslots. Therefore, these two facts yield the following two conclusions:
\begin{enumerate}
\item the set of decoding timeslots (successful receiving timeslots) under $Policy^{**}$ with REP must be a subset of first decoding timeslots under $Policy^1$ with RMIA;
\item the successful receiver set in each decoding timeslot under $Policy^{**}$ with REP must be a subset of the first successful receiver set in each first decoding timeslot under $Policy^1$ with RMIA.
\end{enumerate}

Based on the above two conclusions, an intuitive approach of letting the flow rates under a policy with RMIA to be the same as the flow rates under $Policy^{**}$ with REP is: with the transmitting timeslots determined by $\left\{\alpha_n^{**\left(c\right)}\right\}$, each node $n$ with RMIA treats the set of decoding nodes $\Omega_n$ in each first-decoding timeslot as if the set of decoding nodes is $\Psi_n$ ($\Psi_n \subseteq \Omega_n$) when transmitting with REP and uses probability $\theta_{nk}^{**\left(c\right)}\left(\Psi_n\right)$ to forward the decoded packet to node $k \in \Psi_n$. Here $\Psi_n$ is possibly empty and $\theta_{nk}^{**\left(c\right)}\left(\Psi_n\right)$ is correspondingly $0$. After averaging the forwarding probabilities for each receiving node over the time horizon, $Policy^1$ can be formed.

Following this intuition, in a first-decoding timeslot of one epoch under RMIA, let ${\cal A}_{n,\Psi_n,\Omega_n}^{\rm{rep,rmia}}$ represent the event that the first successful receiver set is $\Omega_n$ ($\Omega_n\neq \emptyset$), while the successful receiver set is $\Psi_n$ if with REP. Then let $q_{n,\Psi_n,\Omega_n}^{{\rm{rep,rmia}}\left(c\right)}\left(t\right)$ represent the number of epochs with RMIA in the first $t$ timeslots, in which ${\cal A}_{n,\Psi_n,\Omega_n}^{\rm{rep,rmia}}$ happens; let $q_{n,{\Psi _n}}^{{\rm{rep}},\left(c\right)}\left( t \right)$ represent the number of timeslots of transmitting commodity $c$ with REP, in which the successful receiver set is $\Psi_n$ ($\Psi_n\neq \emptyset$), in the first $t$ timeslots of the time horizon. Then it naturally follows that
\begin{equation}
q_{n,{\Psi _n}}^{{\rm{rep}},\left(c\right)}\left( t \right) = \sum\limits_{{\Omega _n}:{\Psi _n} \subseteq {\Omega _n}} {q_{n,{\Psi _n},{\Omega _n}}^{{\rm{rep,rmia}},\left(c\right)}\left( t \right)},\ \Psi_n \neq \emptyset.
\label{eq_q_relation1}
\end{equation}
Additionally, by the law of large numbers, with REP, the ratio of the timeslots with decoding set $\Psi_n$ over the total number of transmission slots converges to the probability that $\Psi_n$ is the decoding set in each timeslot, i.e.,
\begin{equation}
\mathop {\lim }\limits_{t \to \infty } \frac{{q_{n,{\Psi _n}}^{{\rm{rep}},\left(c\right)}\left( t \right)}}{{\alpha _n^{**\left( c \right)}\left( t \right)}} = q_{n,{\Psi _n}}^{{\rm{rep}}},{\rm{\ with\ prob}}{\rm{.\ }}1,\ \Psi_n \neq \emptyset.
\label{eq_q_relation3}
\end{equation}
Moreover, we have
\begin{equation}
\sum\limits_{{\Omega _n}:{\Psi _n} \subseteq {\Omega _n}} {\frac{{q_{n,{\Psi _n},{\Omega _n}}^{{\rm{rep}},{\rm{rmia}},\left(c\right)}\left( t \right)}}{{\alpha _n^{**\left( c \right)}\left( t \right)}}}  = \sum\limits_{{\Omega _n}:{\Psi _n} \subseteq {\Omega _n}} {\frac{{\beta _n^{{\rm{rmia}},\left(c\right)}\left( t \right)}}{{\alpha _n^{**\left( c \right)}\left( t \right)}}} \frac{{q_{n,{\Psi _n},{\Omega _n}}^{{\rm{rep}},{\rm{rmia}},\left(c\right)}\left( t \right)}}{{\beta _n^{{\rm{rmia}},\left(c\right)}\left( t \right)}},\ \Psi_n \neq \emptyset
\label{eq_q_relation4}
\end{equation}
where, as defined in the proof of Corollary \ref{cap_region_corr}, ${\beta _n^{{\rm{rmia}}}\left( t \right)}$ is the number of epochs within the $\alpha_n^{\left(c\right)}\left(t\right)$ transmission timeslots, and we have ${{\mathop {\lim }\limits_{t \to \infty } \beta _n^{{\rm{rmia}},\left(c\right)}\left( t \right)} \mathord{\left/{\vphantom {{\mathop {\lim }\limits_{t \to \infty } \beta _n^{{\rm{rmia}},\left(c\right)}\left( t \right)} {\alpha _n^{**\left( c \right)}\left( t \right)}}} \right.\kern-\nulldelimiterspace} {\alpha _n^{**\left( c \right)}\left( t \right)}} = {1 \mathord{\left/{\vphantom {1 {\mathbb{E}\left\{ {{T_n}} \right\}}}} \right.\kern-\nulldelimiterspace} {\mathbb{E}\left\{ {{T_n}} \right\}}} \buildrel \Delta \over = \beta _n^{{\rm{rmia}}}$ with probability 1. Additionally, since the sets of decoding nodes $\Psi_n$ and $\Omega_n$ only depend on the channel realizations of the timeslots in the current epoch and are independent of other epochs, the occurrence of the event that the successful receiver set is $\Psi_n$ if with REP and the first successful receiver set is $\Omega_n$ if with RMIA at the end of each epoch has a fixed probability, denoted as ${q_{n,{\Psi _n},{\Omega _n}}^{{\rm{rep}},{\rm{rmia}}}}$. Then according to the law of large numbers, we have
\begin{equation}
\mathop {\lim }\limits_{t \to \infty } \frac{{q_{n,{\Psi _n},{\Omega _n}}^{{\rm{rep}},{\rm{rmia}},\left(c\right)}\left( t \right)}}{{\beta _n^{{\rm{rmia},\left(c\right)}}\left( t \right)}} = q_{n,{\Psi _n},{\Omega _n}}^{{\rm{rep}},{\rm{rmia}}},{\rm{\ with\ prob}}{\rm{.\ }}1,\ \Omega_n \neq \emptyset.
\label{eq_q_relation4.5}
\end{equation}
In (\ref{eq_q_relation4.5}), we don't need to set $\Psi_n$ nonempty, but $\Omega_n \neq \emptyset$ is sufficient because this convergence also holds true for the case that $\Psi_n=\emptyset,\ \Omega_n \neq \emptyset$; $q_{n,{\Psi _n},{\Omega _n}}^{{\rm{rep}},{\rm{rmia}}}$ does not have a superscript $c$ because this probability value is the same for all commodities. Then going back to (\ref{eq_q_relation4}), let $t\rightarrow \infty$ on both sides, we have
\begin{equation}
\mathop {\lim }\limits_{t \to \infty } \sum\limits_{{\Omega _n}:{\Psi _n} \subseteq {\Omega _n}} {\frac{{q_{n,{\Psi _n},{\Omega _n}}^{{\rm{rep}},{\rm{rmia}},\left(c\right)}\left( t \right)}}{{\alpha _n^{**\left( c \right)}\left( t \right)}}}  = \sum\limits_{{\Omega _n}:{\Psi _n} \subseteq {\Omega _n}} {\beta _n^{{\rm{rmia}}}q_{n,{\Psi _n},{\Omega _n}}^{{\rm{rep}},{\rm{rmia}}}},\ {\rm{with\ prob.\ }}1,\ \Omega_n \neq \emptyset.
\label{eq_q_relation5}
\end{equation}
Divide both sides of (\ref{eq_q_relation1}) by $\alpha_n^{**\left(c\right)}\left(t\right)$ and combine it with (\ref{eq_q_relation5}) and (\ref{eq_q_relation4}), we get
\begin{equation}
q_{n,{\Psi _n}}^{{\rm{rep}}} = \sum\limits_{{\Omega _n}:{\Psi _n} \subseteq {\Omega _n}} {\beta _n^{{\rm{rmia}}}q_{n,{\Psi _n},{\Omega _n}}^{{\rm{rep}},{\rm{rmia}}}},\ \Psi_n \neq \emptyset.
\label{eq_q_relation6}
\end{equation}

As is shown in Theorem \ref{cap_region_thm}, $Policy^{**}$ respectively uses probabilities $\left\{\alpha_n^{**\left(c\right)}\right\}$ and $\left\{\theta_{nk}^{**\left(c\right)}\left(\Omega_n\right)\right\}$ to choose commodities to transmit and the forwarding nodes, and form flow rate matrix $\left(b_{nk}^{**\left(c\right)}\right)$ supporting arbitrary input rate matrix $\left(\lambda_n^{\left(c\right)}\right)$ within $\Lambda_{\rm{REP}}$. According to Theorem \ref{cap_region_thm}, we have, $\forall \left(n,k\right) \in {\cal N}$,
\begin{align}
b_{nk}^{**\left( c \right)} = \alpha _n^{**\left( c \right)}\sum\limits_{{\Psi _n}:k \in {\Psi _n}} {q_{n,{\Psi _n}}^{{\rm{rep}}}\theta _{nk}^{**\left( c \right)}\left( {{\Psi _n}} \right)}.
\label{eq_in_thm2_average_flowing_rate_rep1}
\end{align}

Plug (\ref{eq_q_relation6}) into (\ref{eq_in_thm2_average_flowing_rate_rep1}) and it follows that
\begin{align}
b_{nk}^{**\left( c \right)} &  = \alpha _n^{**\left( c \right)}\sum\limits_{{\Psi _n}:k \in {\Psi _n}} {\ \sum\limits_{{\Omega _n}:{\Psi _n} \subseteq {\Omega _n}} {\beta _n^{{\rm{rmia}}}q_{n,{\Psi _n},{\Omega _n}}^{{\rm{rep}},{\rm{rmia}}}} }\  \theta _{nk}^{**\left( c \right)}\left( {{\Psi _n}} \right) \nonumber\\
&= \alpha _n^{**\left( c \right)}\beta _n^{{\rm{rmia}}}\sum\limits_{{\Omega _n}:k \in {\Omega _n}} {q_{n,{\Omega _n}}^{{\rm{rmia}}}\sum\limits_{{\Psi _n}:{\Psi _n} \subseteq {\Omega _n},k \in {\Psi _n}} {\frac{{q_{n,{\Psi _n},{\Omega _n}}^{{\rm{rep}},{\rm{rmia}}}}}{{q_{n,{\Omega _n}}^{{\rm{rmia}}}}}\theta _{nk}^{**\left( c \right)}\left( {{\Psi _n}} \right)} }.
\label{eq_in_thm2_average_flowing_rate1}
\end{align}
In (\ref{eq_in_thm2_average_flowing_rate1}), ${q_{n,{\Omega _n}}^{{\rm{rmia}}}}$ is the probability that the first successful receiver set is $\Omega_n$ at the end of one epoch and can be put into the denominator because this probability value must be non-zero, which can be shown as follows: let $T_n$ represent the epoch length for transmitting node $n$ under RMIA; let ${\cal A}_{n,\Omega_n}^{\rm{rmia}}$ represent the event that $\Omega_n$ is the first successful receiver set when node $n$ is transmitting with RMIA, and then we have
\begin{align}
q_{n,{\Omega _n}}^{{\rm{rmia}}} &= \Pr \left( {{\cal A}_{n,{\Omega _n}}^{{\rm{rmia}}}} \right) = \sum\limits_{m = 1}^\infty  {\Pr \left( {{T_n} = m,{\cal A}_{n,{\Omega _n}}^{{\rm{rmia}}}} \right)} \nonumber\\
&= \sum\limits_{m = 1}^\infty\  {\prod\limits_{j:j \in {\Omega _n}} {\Pr \left( {\sum\limits_{i = 1}^{m - 1} {{R_{nj}}\left( i \right)}  < {H_0},\sum\limits_{i = 1}^m {{R_{nj}}\left( i \right)}  \ge {H_0}} \right)} \prod\limits_{j:j \notin {\Omega _n},j \in {{\cal K}_n}} {\Pr \left( {\sum\limits_{i = 1}^m {{R_{nj}}\left( i \right)}  < {H_0}} \right)} } \nonumber\\
&= \sum\limits_{m = 1}^\infty\  {\prod\limits_{j:j \in {\Omega _n}} {\left[ {F_{{R_{nj}}}^{\left( {m - 1} \right)}\left( {{H_0}} \right) - F_{{R_{nj}}}^{\left( m \right)}\left( {{H_0}} \right)} \right]} \prod\limits_{j:j \notin {\Omega _n},j \in {K_n}} {F_{{R_{nj}}}^{\left( m \right)}\left( {{H_0}} \right)} }>0,
\label{eq_q_omega_RMIA}
\end{align}
where, according to (\ref{eq_cdf_information_perslot2}) in Lemma \ref{lemma: flowing_rate_property}, we have ${F_{{R_{nj}}}^{\left( {m - 1} \right)}\left( {{H_0}} \right) - F_{{R_{nj}}}^{\left( m \right)}\left( {{H_0}} \right)}>0$. Then combined with the fact that ${F_{{R_{nj}}}^{\left( m \right)}\left( {{H_0}} \right)}>0$, (\ref{eq_q_omega_RMIA}) therefore indicates that $q_{n,{\Omega _n}}^{{\rm{rmia}}}>0$.

If we define
\begin{equation}
\theta _{nk}^{1\left( c \right)}\left( {{\Omega _n}} \right) = \sum\limits_{{\Psi _n}:{\Psi _n} \subseteq {\Omega _n},k \in {\Psi _n}} {\frac{{q_{n,{\Psi _n},{\Omega _n}}^{{\rm{rep}},{\rm{rmia}}}}}{{q_{n,{\Omega _n}}^{{\rm{rmia}}}}}\theta _{nk}^{**\left( c \right)}\left( {{\Psi _n}} \right)},
\label{eq_the_value_of_theta1}
\end{equation}
then by plugging (\ref{eq_the_value_of_theta1}) into (\ref{eq_in_thm2_average_flowing_rate1}), the flow rate can be expressed as
\begin{equation}
b_{nk}^{**\left( c \right)} = \alpha _n^{**\left( c \right)}\beta _n^{{\rm{rmia}}}\sum\limits_{{\Omega _n}:k \in {\Omega _n}} {q_{n,{\Omega _n}}^{{\rm{rmia}}}\theta _{nk}^{1\left( c \right)}\left( {{\Omega _n}} \right)}.
\label{eq_in_thm2_average_flowing_rate2}
\end{equation}
According to (\ref{eq_in_thm2_average_flowing_rate2}), if $\theta_{nk}^{1\left(c\right)}\left(\Omega_n\right)$ under $Policy^1$ with RMIA takes value shown as (\ref{eq_the_value_of_theta1}), the flow rate matrix $\left(b_{nk}^{**\left(c\right)}\right)$ can be formed. Then the later part of the proof is to check whether the set of probability values of $\left\{\theta_{nk}^{1\left(c\right)}\left(\Omega_n\right)\right\}$ shown in (\ref{eq_the_value_of_theta1}) are valid forwarding probability values under $Policy^1$ with RMIA.

Note that the only criterion of checking the validity of forwarding probabilities for a transmission node is to check if the summation of forwarding probabilities of all the successful receiving nodes, given a successful receiver set $\Psi_n$, is no larger than 1. Under $Policy^{**}$, whose forwarding probabilities $\left\{\theta_{nk}^{**\left(c\right)}\left(\Psi_n\right)\right\}$ are valid, we have, $\forall n \in \cal N$,
\begin{equation}
\sum\limits_{k:k \in {\Psi _n}} {\theta _{nk}^{**\left( c \right)}\left( {{\Psi _n}} \right)}  \le 1,\ {\Psi _n} \subseteq {{\cal K}_n}.
\label{eq_valid_forwarding_prob1}
\end{equation}
With RMIA, going back to the definition of ${q_{n,{\Psi _n},{\Omega _n}}^{{\rm{rep}},{\rm{rmia}},\left(c\right)}\left( t \right)}$, firstly we have
\begin{equation}
q_{n,{\Omega _n}}^{{\rm{rmia}},\left(c\right)}\left( t \right) = \sum\limits_{{\Psi _n}:\Psi  \subseteq {\Omega _n}} {q_{n,{\Psi _n},{\Omega _n}}^{{\rm{rep}},{\rm{rmia}},\left(c\right)}\left( t \right)} ,\ {\Omega _n} \ne \emptyset,
\label{eq_q_relation8}
\end{equation}
where, in retrospect, $q_{n,{\Omega _n}}^{{\rm{rmia}},\left(c\right)}\left( t \right)$ is the number of epochs that the first successful receiver set is $\Omega_n$ during the first $t$ timeslots. Additionally, by the law of large numbers, we have
\begin{equation}
\mathop {\lim }\limits_{t \to \infty } \frac{{q_{n,{\Omega _n}}^{{\rm{rmia}},\left(c\right)}\left( t \right)}}{{\alpha _n^{**\left( c \right)}\left( t \right)}} = \mathop {\lim }\limits_{t \to \infty } \frac{{\beta _n^{{\rm{rmia}},\left(c\right)}\left( t \right)}}{{\alpha _n^{**\left( c \right)}\left( t \right)}}\frac{{q_{n,{\Omega _n}}^{{\rm{rmia}},\left(c\right)}\left( t \right)}}{{\beta _n^{{\rm{rmia}},\left(c\right)}\left( t \right)}} = \beta _n^{{\rm{rmia}}}q_{n,{\Omega _n}}^{{\rm{rmia}}},\ \rm{with\ prob.\ 1,}\ \Omega_n \neq \emptyset,
\label{eq_q_relation9}
\end{equation}
Moreover, with the same reasoning, we have
\begin{align}
\mathop {\lim }\limits_{t \to \infty } \sum\limits_{{\Psi _n}:{\Psi _n} \subseteq {\Omega _n}} {\frac{{q_{n,{\Psi _n},{\Omega _n}}^{{\rm{rep}},{\rm{rmia}}\left(c\right)}\left( t \right)}}{{\alpha _n^{**\left( c \right)}\left( t \right)}}}&= \mathop {\lim }\limits_{t \to \infty } \sum\limits_{{\Psi _n}:{\Psi _n} \subseteq {\Omega _n}} {\frac{{\beta _n^{{\rm{rmia}}\left(c\right)}\left( t \right)}}{{\alpha _n^{**\left( c \right)}\left( t \right)}}} \frac{{q_{n,{\Psi _n},{\Omega _n}}^{{\rm{rep}},{\rm{rmia}}\left(c\right)}\left( t \right)}}{{\beta _n^{{\rm{rmia}}\left(c\right)}\left( t \right)}} \nonumber\\
&= \beta _n^{{\rm{rmia}}}\sum\limits_{{\Psi _n}:{\Psi _n} \subseteq {\Omega _n}} {q_{n,{\Psi _n},{\Omega _n}}^{{\rm{rep}},{\rm{rmia}}}},\  \rm{with\ prob.\ 1,}\ \Omega_n \neq \emptyset.
\label{eq_q_relation10}
\end{align}
Divide both sides of (\ref{eq_q_relation8}) by $\alpha_n^{**\left(c\right)}\left(t\right)$, and combine it with (\ref{eq_q_relation9}) and (\ref{eq_q_relation10}), it follows that
\begin{equation}
q_{n,{\Omega _n}}^{{\rm{rmia}}} = \sum\limits_{{\Psi _n}:{\Psi _n} \subseteq {\Omega _n}} {q_{n,{\Psi _n},{\Omega _n}}^{{\rm{rep}},{\rm{rmia}}}},\ \Omega_n \neq \emptyset.
\label{eq_q_relation11}
\end{equation}

With (\ref{eq_valid_forwarding_prob1}) and (\ref{eq_q_relation11}), the check is done by summing (\ref{eq_the_value_of_theta1}) over $k \in \Omega_n$:
\begin{align}
\sum\limits_{k:k \in {\Omega _n}} {\theta _{nk}^{1\left( c \right)}\left( {{\Omega _n}} \right)}  &= \sum\limits_{{\Psi _n}:{\Psi _n} \subseteq {\Omega _n},{\Psi _n} \ne \emptyset } {\frac{{q_{n,{\Psi _n},{\Omega _n}}^{{\rm{rep}},{\rm{rmia}}}}}{{q_{n,{\Omega _n}}^{{\rm{rmia}}}}}\sum\limits_{k:k \in {\Psi _n}} {\theta _{nk}^{**\left( c \right)}\left( {{\Psi _n}} \right)} } \nonumber\\
&\le \sum\limits_{{\Psi _n}:{\Psi _n} \subseteq {\Omega _n},{\Psi _n}\ne \emptyset} {\frac{{q_{n,{\Psi _n},{\Omega _n}}^{{\rm{rep}},{\rm{rmia}}}}}{{q_{n,{\Omega _n}}^{{\rm{rmia}}}}}} \nonumber\\
&= \frac{{\sum\limits_{{\Psi _n}:{\Psi _n} \subseteq {\Omega _n},{\Psi _n} \ne \emptyset } {q_{n,{\Psi _n},{\Omega _n}}^{{\rm{rep}},{\rm{rmia}}}} }}{{\sum\limits_{{\Psi _n}:{\Psi _n} \subseteq {\Omega _n},{\Psi _n} \ne \emptyset } {q_{n,{\Psi _n},{\Omega _n}}^{{\rm{rep}},{\rm{rmia}}}}  + q_{n,\emptyset ,{\Omega _n}}^{{\rm{rep}},{\rm{rmia}}}}} \le 1.
\label{eq_q_relation12}
\end{align}
According to the result shown as (\ref{eq_q_relation12}), the set of forwarding probability values under $Policy^1$ shown by (\ref{eq_the_value_of_theta1}) are valid forwarding probability values.

Thus, as long as $Policy^{**}$ with REP forms a flow rate matrix $\left(b_{nk}^{**\left(c\right)}\right)$ and supports $\left(\lambda_n^{\left(c\right)}\right) \in \Lambda_{\rm{REP}}$, there must exist a $Policy^{1}$ with RMIA that can form the same flow rate matrix $\left(b_{nk}^{**\left(c\right)}\right)$ and supports $\left(\lambda_n^{\left(c\right)}\right)$.

\subsection{There exists at least one exogenous input rate matrix which is within $\Lambda_{\rm{RMIA}}$ but outside of $\Lambda_{\rm{REP}}$}
\label{subsec_capacity_comparison_sub2}
The proof strategy of this subsection is to do some modifications on $Policy^1$ and develop a new stationary randomized policy with RMIA, denoted as $Policy^2$, such that it can support the input rate increase $\Delta_{n_0}^{\left(c_0\right)}$ additional to the $\left(n_0,c_0\right)$th dimension of the input rate matrix $\left(\lambda_n^{\left(c\right)}\right)\in \Lambda_{\rm{REP}}$ and forms a new input rate matrix $\left({\lambda'}_n^{\left(c\right)}\right)$ outside of $\Lambda_{\rm{REP}}$ satisfying
\begin{align}
\label{eq_enlarged_input_rate_matrix}
{\lambda '}_{{n_0}}^{\left( {{c_0}} \right)} = \lambda _{{n_0}}^{\left( {{c_0}} \right)} + \Delta_{n_0}^{\left(c_0\right)} ,{\rm{\ }}\Delta_{n_0}^{\left(c_0\right)}  &> 0,\ \lambda_{n_0}^{\left(c_0\right)}>0;\nonumber\\
{\lambda '}_n^{\left( c \right)} = \lambda _n^{\left( c \right)},{\rm{\ }}\left( {n,c} \right) &\ne \left( {{n_0},{c_0}} \right),
\end{align}
where $\left(n_0,c_0\right)$ is a chosen node-commodity pair, such that $\lambda_{n_0}^{\left(c_0\right)}>0$.

The proof is divided into two parts: deriving the flow rate increase and lower bounding the flow increase.

\subsubsection{\textbf{Deriving the flow rate increase}}
In this part of the proof, we exploit that if it is possible to further increase the forwarding probability $\theta_{nk}^{1\left(c\right)}\left(\Omega_n\right)$ under $Policy^1$ with RMIA, the flow rate over a link can further increase. We follow the intuition that the same flow rates as those under $Policy^{**}$ with REP can be achieved with RMIA if each node $n$ treats the set of decoding nodes $\Omega_n$ with RMIA as $\Psi_n$ if with REP, assuming the same transmission timeslots. According to this strategy, in the first decoding timeslots of commodity $c_0$ packets when ${\cal A}_{n,\emptyset,\Omega_n}^{\rm{rep,rmia}}$ happens, i.e., the first successful receiver set is $\Omega_n$ (nonempty) with RMIA while the successful receiver set is $\emptyset$ with REP, the transmitting node $n$ retains the packet, which might be a waste of opportunities. However, if node $n$ uses probability 1 to forward the decoded packet to node $k\in \Omega_n$ in those timeslots, the time average flow rate through link $\left(n,k\right)$ can increase. The average flow increase can be computed as follows:
\begin{align}
\Delta _{nk}^{\left( {{c_0}} \right)} &= \mathop {\lim }\limits_{t \to \infty } \sum\limits_{{\Omega _n}:k \in {\Omega _n}} {\frac{{q_{n,\emptyset ,{\Omega _n}}^{{\rm{rep}},{\rm{rmia}},\left(c_0\right)}\left( t \right) \cdot 1}}{t}}  \nonumber\\
&= \mathop {\lim }\limits_{t \to \infty } \sum\limits_{{\Omega _n}:k \in {\Omega _n}} {\frac{{\alpha _n^{**\left( c_0 \right)}\left( t \right)}}{t}\frac{{\beta _n^{{\rm{rmia}},\left(c_0\right)}\left( t \right)}}{{\alpha _n^{**\left( c_0 \right)}\left( t \right)}}\frac{{q_{n,\emptyset ,{\Omega _n}}^{{\rm{rep}},{\rm{rmia}},\left(c_0\right)}\left( t \right)}}{{\beta _n^{{\rm{rmia}},\left(c_0\right)}\left( t \right)}}} \nonumber\\
&= \alpha _n^{**\left( c_0 \right)}\beta _n^{{\rm{rmia}}}\sum\limits_{{\Omega _n}:k \in {\Omega _n}} {q_{n,\emptyset ,{\Omega _n}}^{{\rm{rep}},{\rm{rmia}}}}\ \rm{with\ prob.\ 1.}
\label{eq_flow_increase5}
\end{align}
Then based on (\ref{eq_in_thm2_average_flowing_rate1}) and (\ref{eq_flow_increase5}), the increased flow rate can be computed as
\begin{align}
b_{nk}^{**\left( c_0 \right)} + \Delta _{nk}^{\left( {{c_0}} \right)} & =\alpha _n^{**\left( c_0 \right)}\sum\limits_{{\Omega _n}:k \in {\Omega _n}} {\beta _n^{{\rm{rmia}}}q_{n,{\Omega _n}}^{{\rm{rmia}}}\left[ {\sum\limits_{{\Psi _n}:{\Psi _n} \subseteq {\Omega _n},k \in {\Psi _n}} {\frac{{q_{n,{\Psi _n},{\Omega _n}}^{{\rm{rep}},{\rm{rmia}}}}}{{q_{n,{\Omega _n}}^{{\rm{rmia}}}}}\theta _{nk}^{**\left( c_0 \right)}\left( {{\Psi _n}} \right)}  + \frac{{q_{n,\emptyset ,{\Omega _n}}^{{\rm{rep}},{\rm{rmia}}}}}{{q_{n,{\Omega _n}}^{{\rm{rmia}}}}}} \right]}\nonumber\\
&\buildrel \Delta \over = \alpha _n^{**\left( c_0 \right)}\sum\limits_{{\Omega _n}:k \in {\Omega _n}} {\beta _n^{{\rm{rmia}}}q_{n,{\Omega _n}}^{{\rm{rmia}}}\theta _{nk}^{1'\left( c_0 \right)}\left( {{\Omega _n}} \right)},
\end{align}
where $\theta _{nk}^{1'\left( c_0 \right)}\left( {{\Omega _n}} \right)$ is the equivalent overall forwarding probability of the new stationary randomized policy:
\begin{equation}
\theta _{nk}^{1'\left( c_0 \right)}\left( {{\Omega _n}} \right) = \sum\limits_{{\Psi _n}:{\Psi _n} \subseteq {\Omega _n},k \in {\Psi _n}} {\frac{{q_{n,{\Psi _n},{\Omega _n}}^{{\rm{rep}},{\rm{rmia}}}}}{{q_{n,{\Omega _n}}^{{\rm{rmia}}}}}\theta _{nk}^{**\left( c_0 \right)}\left( {{\Psi _n}} \right)}  + \frac{{q_{n,\emptyset ,{\Omega _n}}^{{\rm{rep}},{\rm{rmia}}}}}{{q_{n,{\Omega _n}}^{{\rm{rmia}}}}}= \theta _{nk}^{1\left( c_0 \right)}\left( {{\Omega _n}} \right) + \frac{{q_{n,\emptyset ,{\Omega _n}}^{{\rm{rep}},{\rm{rmia}}}}}{{q_{n,{\Omega _n}}^{{\rm{rmia}}}}}.
\label{eq_forwarding_prob_1prime1}
\end{equation}
Additionally, note that for each transmitting node, the flow increase shown by (\ref{eq_flow_increase5}) is for one chosen receiving node $k$, while the flow rates to other receiving nodes do not change, i.e.,
\begin{equation}
\theta _{nj}^{1'\left( c_0 \right)}\left( {{\Omega _n}} \right) = \theta _{nj}^{1\left( c_0 \right)}\left( {{\Omega _n}} \right) = \sum\limits_{{\Psi _n}:{\Psi _n} \subseteq {\Omega _n},j \in {\Psi _n}} {\frac{{q_{n,{\Psi _n},{\Omega _n}}^{{\rm{rep}},{\rm{rmia}}}}}{{q_{n,{\Omega _n}}^{{\rm{rmia}}}}}\theta _{nj}^{**\left( c_0 \right)}\left( {{\Psi _n}} \right)},\ j\in {\cal K}_n,\ j \neq k.
\label{eq_forwarding_prob_1prime2}
\end{equation}

With (\ref{eq_forwarding_prob_1prime1}) and (\ref{eq_forwarding_prob_1prime2}), now we start checking the validity of the new forwarding probabilities as follows:
\begin{align}
\sum\limits_{j:j \in {\Omega _n}} {\theta _{nj}^{1'\left( c_0 \right)}\left( {{\Omega _n}} \right)}  &= \theta _{nk}^{1'\left( c_0 \right)}\left( {{\Omega _n}} \right) + \sum\limits_{j:j \in {\Omega _n},j \ne k} {\theta _{nj}^{1'\left( c_0 \right)}\left( {{\Omega _n}} \right)}  \nonumber\\
&= \sum\limits_{{\Psi _n}:{\Psi _n} \subseteq {\Omega _n},{\Psi _n} \ne \emptyset } {\frac{{q_{n,{\Psi _n},{\Omega _n}}^{{\rm{rep}},{\rm{rmia}}}}}{{q_{n,{\Omega _n}}^{{\rm{rmia}}}}}\sum\limits_{j:j \in {\Psi _n}} {\theta _{nj}^{**\left( c_0 \right)}\left( {{\Psi _n}} \right)} }  + \frac{{q_{n,\emptyset ,{\Omega _n}}^{{\rm{rep}},{\rm{rmia}}}}}{{q_{n,{\Omega _n}}^{{\rm{rmia}}}}}\nonumber\\
&\le \sum\limits_{{\Psi _n}:{\Psi _n} \subseteq {\Omega _n},{\Psi _n} \ne \emptyset } {\frac{{q_{n,{\Psi _n},{\Omega _n}}^{{\rm{rep}},{\rm{rmia}}}}}{{q_{n,{\Omega _n}}^{{\rm{rmia}}}}}}  + \frac{{q_{n,\emptyset ,{\Omega _n}}^{{\rm{rep}},{\rm{rmia}}}}}{{q_{n,{\Omega _n}}^{{\rm{rmia}}}}}\nonumber\\
&= \frac{{\sum\limits_{{\Psi _n}:{\Psi _n} \subseteq {\Omega _n},{\Psi _n} \ne \emptyset } {q_{n,{\Psi _n},{\Omega _n}}^{{\rm{rep}},{\rm{rmia}}}}  + q_{n,\emptyset ,{\Omega _n}}^{{\rm{rep}},{\rm{rmia}}}}}{{\sum\limits_{{\Psi _n}:{\Psi _n} \subseteq {\Omega _n}} {q_{n,{\Psi _n},{\Omega _n}}^{{\rm{rep}},{\rm{rmia}}}} }} = 1.
\label{eq_forwarding_prob_1prime3}
\end{align}

With the flow increase on a single link, we further construct the increased input rate matrix $\left({\lambda'}_n^{\left(c\right)}\right)$. First let path $l$ be an arbitrary simple path leading from node $n_0$ to node $c_0$. According to (\ref{eq_flow_increase5}), consider the smallest possible flow rate increase among all the links along path $l$ for commodity $c_0$:
\begin{equation}
\Delta _{l,\min }^{\left( c_0 \right)} = \mathop {\min }\limits_{\left( {n,k} \right)} \left\{ {\Delta _{nk}^{\left( c_0 \right)},\left( {n,k} \right) \in l} \right\},
\label{eq_mini_increase_link}
\end{equation}
and denote the corresponding link as the \emph{minimum-flow-increase link} of path $l$. As is computed by (\ref{eq_flow_increase5}), each link $\left(n,k\right)$ along route $l$ can support a flow rate increase as much as $\Delta_{nk}^{\left(c_0\right)}$, which is no less than $\Delta_{l,\rm{min}}$. Therefore, each link $\left(n,k\right)$ along path $l$ can support a flow rate increase of $\Delta_{l,\rm{min}}$ just by assigning a new forwarding probability $\theta _{nk}^{2\left( c \right)}\left( {{\Omega _n}} \right)$ that is smaller than or equal to $\theta _{nk}^{1'\left( c \right)}\left( {{\Omega _n}} \right)$, i.e., there exists a value $\xi_{nk,\emptyset,\Omega_n}^{\left(c\right)}$, where $0<\xi_{nk,\emptyset,\Omega_n}^{\left(c_0\right)}\le 1$, such that
\begin{equation}
\Delta _{l,\min }^{\left( {{c_0}} \right)} = \Delta _{nk}^{\left( {{c_0}} \right)}\xi _{nk,\emptyset ,{\Omega _n}}^{\left( {{c_0}} \right)} = \alpha _n^{**\left( {{c_0}} \right)}\beta _n^{{\rm{rmia}}}\xi _{nk,\emptyset ,{\Omega _n}}^{\left( {{c_0}} \right)}\sum\limits_{{\Omega _n}:k \in {\Omega _n}} {q_{n,\emptyset ,{\Omega _n}}^{{\rm{rep}},{\rm{rmia}}}}.
\label{eq_policy2_flowing_increase}
\end{equation}

With (\ref{eq_policy2_flowing_increase}), we develop another policy $Policy^2$, which chooses commodity $c$ to transmit with probability $\alpha_n^{**\left(c\right)}$ in each timeslot and chooses the successful receiver $k$ as the forwarding node with probability $\theta_{nk}^{2\left(c\right)}\left(\Omega_n\right)$, which satisfies the following relations:
\begin{equation}
\theta _{nk}^{2\left( c_0 \right)}\left( {{\Omega _n}} \right) = \theta _{nk}^{1\left( c_0 \right)}\left( {{\Omega _n}} \right) + \frac{{q_{n,\emptyset ,{\Omega _n}}^{{\rm{rep}},{\rm{rmia}}}}}{{q_{n,{\Omega _n}}^{{\rm{rmia}}}}}\xi _{nk,\emptyset ,{\Omega _n}}^{\left( {{c_0}} \right)},\ 0<\xi _{nk,\emptyset ,{\Omega _n}}^{\left( {{c_0}} \right)}\le 1,{\ \rm{for}}\ \forall \left(n,k\right) \in l;
\end{equation}
\begin{equation}
\theta _{nk}^{2\left( c_0 \right)}\left( {{\Omega _n}} \right) = \theta _{nk}^{1\left( c_0 \right)}\left( {{\Omega _n}} \right), \ {\rm{for\ }}\forall \left(n,k\right) \notin l;
\end{equation}
\begin{equation}
\theta _{nk}^{2\left( c \right)}\left( {{\Omega _n}} \right) = \theta _{nk}^{1\left( c \right)}\left( {{\Omega _n}} \right), \ {\rm{for\ }} \forall c \neq c_0.
\end{equation}
Correspondingly, under $Policy^2$, the time average flow rate over each link can be expressed as follows:
\begin{equation}
b_{nk}^{2\left( c_0 \right)} = b_{nk}^{**\left( c_0 \right)} + \Delta _{l,{\rm{min}}}^{\left( c_0 \right)},{\rm{\ }}{\ \rm{for}}\ \forall \left(n,k\right) \in l;
\label{eq_flowing_rate_policy2_1}
\end{equation}
\begin{equation}
b_{nk}^{2\left( c_0 \right)} = b_{nk}^{**\left( c_0 \right)},\ {\rm{for\ }}\forall \left(n,k\right) \notin l;
\label{eq_flowing_rate_policy2_2}
\end{equation}
\begin{equation}
b_{nk}^{2\left( c \right)} = b_{nk}^{**\left( c \right)},\ {\rm{for\ }} \forall c \neq c_0.
\label{eq_flowing_rate_policy2_3}
\end{equation}
With (\ref{eq_flowing_rate_policy2_1})-(\ref{eq_flowing_rate_policy2_3}), the increased input rate matrix $\left({\lambda'}_n^{\left(c\right)}\right)$ can be supported under $Policy^2$. The explanation is: when the input rate matrix changes from $\left({\lambda}_n^{\left(c\right)}\right)$ to $\left({\lambda'}_n^{\left(c\right)}\right)$, where the $\left(n_0,c_0\right)$th entry increases by $\Delta_{l,\rm{min}}^{\left(c_0\right)}$ while other entries don't change, the adoption of $Policy^2$ can guarantee that the flow rate along the simple path $l$ for commodity $c_0$ can increase by $\Delta_{l,\rm{min}}^{\left(c_0\right)}$. Thus, $Policy^2$ with RMIA can stably support the exogenous input rate matrix $\left({\lambda'}_n^{\left(c\right)}\right)$.\\

\subsubsection{\textbf{Lower bounding the flow rate increase}}
The next step is to choose a proper simple path $l$ among all the simple paths leading from node $n_0$ to node $c_0$ for $Policy^2$, such that $\Delta_{l,\min}^{\left(c_0\right)}$ is lower bounded and $\left({\lambda'}_n^{\left(c\right)}\right)$ exceeds the boundary of $\Lambda_{\rm{REP}}$ in its $\left(n_0,c_0\right)$th dimension when $\left({\lambda}_n^{\left(c\right)}\right)$ sufficiently approaches the boundary.

Before starting the later proof, some definitions and explanations are necessary:
\begin{itemize}
\item Let $\cal L$ represent the set of the simple paths leading from node $n_0$ to node $c_0$.
\item Let $L$ represent the total number of paths in $\cal L$. Note that $L$ is a finite fixed integer because the network is stationary with finite number of nodes and links.
\item Define the set of flow paths from node $n_0$ to node $c_0$, which contain a common simple path but possibly different loops, as a \emph{path cluster}. After eliminating the loop part, each flowing path from node $n_0$ to node $c_0$ is reduced into a simple path. Here the path cluster is the set of paths which can be reduced into the same simple path. An important property for a path cluster is that each packet flowing through a path in one path cluster must flow through the corresponding common simple path, which is defined as the \emph{main trunk} of the path cluster.
\item For an arbitrary simple path $l$ leading from node $n_0$ to node $c_0$, define the link as the \emph{bottleneck link}, through which the time average flow rate of commodity $c_0$ under $Policy^{**}$ is the smallest among all the links of the path $l$, and denote the corresponding \emph{bottleneck flow rate} of commodity $c_0$ as $b_l^{**\left(c_0\right)}$. To be specific, for a simple path $l$ with $J$ hops, whose nodes' index sequence is $n_0,n_{l,1},\cdots,n_{l,J-1}, c_0$, the bottleneck flow rate is defined as
    \begin{equation}
    b_l^{**\left( {{c_0}} \right)} = \min \left\{ {b_{{n_0}{n_{l,1}}}^{**\left( {{c_0}} \right)},b_{{n_{l,1}}{n_{l,2}}}^{**\left( {{c_0}} \right)}, \cdots ,b_{{n_{l,J - 2}}{n_{l,J - 1}}}^{**\left( {{c_0}} \right)},b_{{n_{l,J - 1}}{c_0}}^{**\left( {{c_0}} \right)}} \right\}.
    \end{equation}
\end{itemize}

Define $Y_{n_0}^{**\left(c_0\right)}\left(t\right)$ as the number of packets delivered to destination $c_0$ with source node $n_0$ during the first $t$ timeslots under $Policy^{**}$; define $G_l^{**\left(c_0\right)}\left(t\right)$ as the number of units (copies of packets) passing through the bottleneck link of simple path $l$ during the first $t$ timeslots under $Policy^{**}$ with REP. Then we have the following relationship
\begin{equation}
Y_{{n_0}}^{**\left( {{c_0}} \right)}\left( t \right) \le \sum\limits_{l = 1}^L {G_l^{**\left( c_0 \right)}\left( t \right)}.
\label{eq_bottleneck_flow1}
\end{equation}
The inequality in (\ref{eq_bottleneck_flow1}) is due to the following reasons:
\begin{itemize}
\item The same packet may flow through the same bottleneck link more than once during the first $t$ timeslots, and therefore are counted multiple times in $G_l^{**\left(c_0\right)}\left(t\right)$. In contrast, $Y_{n_0}^{**\left(c_0\right)}\left(t\right)$ only counts each delivered packet once.
\item Some units with commodity $c_0$ but with different source nodes may pass through the bottleneck links, which are also counted by $G_l^{**\left(c_0\right)}\left(t\right)$.
\item Different simple paths possibly have a common bottleneck link. Thus, some duplicated terms of the bottleneck rate might be added on the right hand side of (\ref{eq_bottleneck_flow1}).
\end{itemize}

Under $Policy^{**}$ and with input rate matrix $\left(\lambda_{n}^{\left(c\right)}\right)$, the network is rate stable, and therefore we have
\begin{equation}
\lambda _{{n_0}}^{\left(c_0\right)} = \mathop {\lim }\limits_{t \to \infty } \frac{{Y_{{n_0}}^{**\left(c_0\right)}\left( t \right)}}{t} \le \mathop {\lim }\limits_{t \to \infty } \frac{{\sum\limits_{l = 1}^L {G_l^{**\left( c_0 \right)}\left( t \right)} }}{t} = \sum\limits_{l = 1}^L {b_l^{**\left( {{c_0}} \right)}}.
\label{eq_flow_rate_lower_bound1}
\end{equation}
Among the $L$ single trunk routes, if denoting the simple path with the largest bottleneck rate under $Policy^{**}$ as $l_{\rm{max}}$ and the corresponding bottleneck rate as
\begin{equation}
b_{{l_{\max }}}^{**\left( {{c_0}} \right)} = \mathop {\max }\limits_{l \in \cal L} \left\{ {b_l^{**\left( {{c_0}} \right)}} \right\},
\end{equation}
we can get the following inequality based on (\ref{eq_flow_rate_lower_bound1}):
\begin{equation}
b_{{l_{\max }}}^{**\left( {{c_0}} \right)} \ge \frac{1}{L}\sum\limits_{l = 1}^L {b_l^{**\left( {{c_0}} \right)}}  \ge \frac{{\lambda _{{n_0}}^{{\left(c_0\right)}}}}{L}.
\label{eq_flow_rate_lower_bound2}
\end{equation}

If we choose $l_{\rm{max}}$ as the simple path sustaining the corresponding flow rate increase $\Delta_{l_{\rm{max}},\rm{min}}^{\left(c_0\right)}$ under $Policy^2$, according to (\ref{eq_flow_increase5}) and (\ref{eq_mini_increase_link}) in Part 1), and letting $\left(\tilde n, \tilde k\right)$ represent the minimum-flow-increase link on the path $l_{\max}$, we have
\begin{align}
\Delta _{{l_{\max }},\min }^{\left( {{c_0}} \right)} &= \frac{{\sum\limits_{{\Omega _{\tilde n}}:\tilde k \in {\Omega _{\tilde n}}} {q_{\tilde n,\emptyset ,{\Omega _{\tilde n}}}^{{\rm{rep}},{\rm{rmia}}}} }}{{\sum\limits_{{\Omega _{\tilde n}}:\tilde k \in {\Omega _{\tilde n}}} {q_{\tilde n,{\Omega _{\tilde n}}}^{{\rm{rmia}}}} }}\alpha _{\tilde n}^{**\left( {{c_0}} \right)}\beta _{\tilde n}^{{\rm{rmia}}}\sum\limits_{{\Omega _{\tilde n}}:\tilde k \in {\Omega _{\tilde n}}} {q_{\tilde n,{\Omega _{\tilde n}}}^{{\rm{rmia}}}} \nonumber\\
&\ge \frac{{\sum\limits_{{\Omega _{\tilde n}}:\tilde k \in {\Omega _{\tilde n}}} {q_{\tilde n,\emptyset ,{\Omega _{\tilde n}}}^{{\rm{rep}},{\rm{rmia}}}} }}{{\sum\limits_{{\Omega _{\tilde n}}:\tilde k \in {\Omega _{\tilde n}}} {q_{\tilde n,{\Omega _{\tilde n}}}^{{\rm{rmia}}}} }}\alpha _{\tilde n}^{**\left( {{c_0}} \right)}\beta _{\tilde n}^{{\rm{rmia}}}\sum\limits_{{\Omega _{\tilde n}}:\tilde k \in {\Omega _{\tilde n}}} {q_{\tilde n,{\Omega _n}}^{{\rm{rmia}}}\theta _{\tilde n\tilde k}^{1\left( {{c_0}} \right)}\left( {{\Omega _{\tilde n}}} \right)}.
\label{eq_flow_increase6}
\end{align}
Note that the term $\alpha _{\tilde n}^{**\left( {{c_0}} \right)}\beta _{\tilde n}^{{\rm{rmia}}}\sum\limits_{{\Omega _{\tilde n}}:\tilde k \in {\Omega _{\tilde n}}} {q_{\tilde n,{\Omega _n}}^{{\rm{rmia}}}\theta _{\tilde n\tilde k}^{1\left( {{c_0}} \right)}\left( {{\Omega _{\tilde n}}} \right)}$ is the flow rate expression of link $\left(\tilde n,\tilde k\right)$ under $Policy^1$, which is the same as the flow rate of link $\left(\tilde n,\tilde k\right)$ under $Policy^{**}$ and is no less than the bottleneck flow rate of path $l_{\rm{max}}$. Then combined with (\ref{eq_flow_rate_lower_bound2}), it follows that
\begin{equation}
\Delta _{{l_{\max }},\min }^{\left( {{c_0}} \right)} \ge \frac{{\sum\limits_{{\Omega _{\tilde n}}:\tilde k \in {\Omega _{\tilde n}}} {q_{\tilde n,\emptyset ,{\Omega _{\tilde n}}}^{{\rm{rep}},{\rm{rmia}}}} }}{{\sum\limits_{{\Omega _{\tilde n}}:\tilde k \in {\Omega _{\tilde n}}} {q_{\tilde n,{\Omega _{\tilde n}}}^{{\rm{rmia}}}} }}b_{\tilde n\tilde k}^{**\left( {{c_0}} \right)} \ge \frac{{\sum\limits_{{\Omega _{\tilde n}}:\tilde k \in {\Omega _{\tilde n}}} {q_{\tilde n,\emptyset ,{\Omega _{\tilde n}}}^{{\rm{rep}},{\rm{rmia}}}} }}{{\sum\limits_{{\Omega _{\tilde n}}:\tilde k \in {\Omega _{\tilde n}}} {q_{\tilde n,{\Omega _{\tilde n}}}^{{\rm{rmia}}}} }}b_{{l_{\max }}}^{**\left( {{c_0}} \right)}\ge \frac{{\sum\limits_{{\Omega _{\tilde n}}:\tilde k \in {\Omega _{\tilde n}}} {q_{\tilde n,\emptyset ,{\Omega _{\tilde n}}}^{{\rm{rep}},{\rm{rmia}}}} }}{{\sum\limits_{{\Omega _{\tilde n}}:\tilde k \in {\Omega _{\tilde n}}} {q_{\tilde n,{\Omega _{\tilde n}}}^{{\rm{rmia}}}} }}\frac{{\lambda _{{n_0}}^{\left( {{c_0}} \right)}}}{L} \buildrel \Delta \over  = \eta _{{l_{\max }}}^{{\rm{rep}},{\rm{rmia}}}\lambda _{{n_0}}^{\left( {{c_0}} \right)},
\label{eq_flowing_increaes_lower_bound}
\end{equation}
where
\begin{equation}
\eta _{{l_{\max }}}^{{\rm{rep}},{\rm{rmia}}} = \frac{{\sum\limits_{{\Omega _{\tilde n}}:\tilde k \in {\Omega _{\tilde n}}} {q_{\tilde n,\emptyset ,{\Omega _{\tilde n}}}^{{\rm{rep}},{\rm{rmia}}}} }}{{L\sum\limits_{{\Omega _{\tilde n}}:\tilde k \in {\Omega _{\tilde n}}} {q_{\tilde n ,{\Omega _{\tilde n}}}^{{\rm{rmia}}}} }}.
\label{eq_topology_coff_eta}
\end{equation}
Here note that $\eta _{{l_{\max }}}^{{\rm{rep}},{\rm{rmia}}}$ must be positive, which can be checked with the help of (\ref{eq_cdf_information_perslot1}) in Lemma \ref{lemma: flowing_rate_property} as follows:
\begin{align}
\sum\limits_{{\Omega _n}:k \in {\Omega _n}} {q_{n,\emptyset ,{\Omega _n}}^{{\rm{rep,rmia}}}} &=\sum\limits_{{\Omega _n}:k \in {\Omega _n}} {\Pr \left( {{\cal A}_{n,\emptyset ,{\Omega _n}}^{{\rm{rep,rmia}}}} \right)} =\sum\limits_{m = 1}^\infty  {\sum\limits_{{\Omega _n}:k \in {\Omega _n}} {\Pr \left( {{T_n} = m,{\cal A}_{n,\emptyset ,{\Omega _n}}^{{\rm{rep,rmia}}}} \right)} } \nonumber\\
&= \sum\limits_{m = 1}^\infty  {\Pr \left( {\sum\limits_{i = 1}^{m - 1} {{R_{nk}}\left( i \right)}  < {H_0},{R_{nk}}\left( m \right) < {H_0},\sum\limits_{i = 1}^m {{R_{nk}}\left( i \right)}  \ge {H_0}} \right)}\nonumber\\
&\ \ \ \ \ \ \ \ \cdot\prod\limits_{j:j \ne k,j \in {{\cal K}_n}} {\Pr \left( {\sum\limits_{i = 1}^{m - 1} {{R_{nj}}\left( i \right)}  < {H_0},{R_{nj}}\left( m \right) < {H_0}} \right)} \nonumber\\
& = \sum\limits_{m = 1}^\infty  {\left[ {F_{{R_{nk}}}^{\left( {m - 1} \right)}\left( {{H_0}} \right){F_{{R_{nk}}}}\left( {{H_0}} \right) - F_{{R_{nk}}}^{\left( m \right)}\left( {{H_0}} \right)} \right]\prod\limits_{j:j \ne k,j \in {{\cal K}_n}} {F_{{R_{nj}}}^{\left( {m - 1} \right)}\left( {{H_0}} \right){F_{{R_{nj}}}}\left( {{H_0}} \right)} }>0.
\label{eq_sum_q_Psi_Omega1}
\end{align}
Then combining (\ref{eq_sum_q_Psi_Omega1}) with the result from (\ref{eq_q_omega_RMIA}) that $q_{n,{\Omega _n}}^{{\rm{rmia}}}>0$, it follows that for $\tilde n$, $\eta _{{l_{\max }}}^{{\rm{rep}},{\rm{rmia}}}>0$

Moreover, note that all the components of $\eta _{{l_{\max }}}^{{\rm{rep}},{\rm{rmia}}}$ ------ $q_{\tilde n, \emptyset, \Omega_{\tilde n}}^{\rm{rep, rmia}}$, $q_{\tilde n, \Omega_{\tilde n}}^{\rm{rmia}}$, and $L$------ depend on the network topology and $Policy^{**}$, under which the choice of $l_{\max}$ and the minimum-flow-increase link $\left(\tilde n, \tilde k\right)$ on the path $l_{\max}$ are fixed due to the stationary network topology. Thus $\eta _{{l_{\max }}}^{{\rm{rep}},{\rm{rmia}}}$ remains as a fixed positive constant when the input rate matrix $\left(\lambda_n^{\left(c\right)}\right)$ changes.

Then let $\left(\lambda_n^{\left(c\right)}\right)$ approach the boundary of the REP network capacity region $\Lambda_{\rm{REP}}$, which results in that $\lambda_{n_0}^{\left(c_0\right)}$ approaches to the $\left(n_0,c_0\right)$th entry of the boundary of $\Lambda_{\rm{REP}}$, denoted as $\Lambda_{n_0,\rm{REP}}^{\left(c_0\right)}$. To be specific, we can let $\lambda_{n_0}^{\left(c_0\right)}$ be so close to $\Lambda_{n_0,\rm{REP}}^{\left(c_0\right)}$ that $\frac{1}{{1 + \eta _{{l_{\max }}}^{{\rm{rep}},{\rm{rmia}}}}}\Lambda _{{n_0},{\rm{REP}}}^{\left( {{c_0}} \right)} < \lambda _{{n_0}}^{\left( {{c_0}} \right)} < \Lambda _{{n_0},{\rm{REP}}}^{\left( {{c_0}} \right)}$, then it follows based on (\ref{eq_flowing_increaes_lower_bound}) that
\begin{equation}
\lambda _{{n_0}}^{\left( {{c_0}} \right)} + \Delta _{{l_{\max }},\min }^{\left( {{c_0}} \right)} > \left( {1 + \eta _{{l_{\max }}}^{{\rm{rep}},{\rm{rmia}}}} \right)\frac{{\Lambda _{{n_0},{\rm{REP}}}^{\left( {{c_0}} \right)}}}{{1 + \eta _{{l_{\max }}}^{{\rm{rep}},{\rm{rmia}}}}} = \Lambda _{{n_0},{\rm{REP}}}^{\left( {{c_0}} \right)},
\label{eq_exceeds_boundary_Lambda}
\end{equation}
which demonstrates that $\left({\lambda'}_n^{\left(c\right)}\right)$ exceeds the boundary of $\Lambda_{\rm{REP}}$ in its $\left(n_0,c_0\right)$th dimension. Thus, there exists an input rate matrix that can be supported under RMIA but is outside of $\Lambda_{\rm{REP}}$.

The above analysis is derived in a conservative way: the increase from the original input rate matrix $\left(\lambda_n^{\left(c\right)}\right)$ in $\Lambda_{\rm{REP}}$ to the new one $\left({\lambda'}_n^{\left(c\right)}\right)$ in $\Lambda_{\rm{RMIA}}$ is only on a single entry $\left(n_0,c_0\right)$. Note that the exogenous input packet streams of commodity $c_0$ may enter the network through different nodes, and correspondingly, for commodity $c_0$, we can find multiple input rate matrices outside of $\Lambda_{\rm{REP}}$ that are increased from $\left(\lambda_n^{\left(c\right)}\right)$ at different single entries when switching from REP to RMIA. Moreover, it is even possible to simultaneously increase multiple entries of $\left(\lambda_n^{\left(c\right)}\right)$ with a common commodity component $c_0$, as long as the flow paths from different source nodes to node $c_0$ sustaining the flow increase are strongly disjoint (having no common node) except at node $c_0$. Furthermore, it is obvious that the positive supportable input rate increases can be achieved simultaneously for all the commodities without affecting each other, which forms further increased input rate matrices. Even more to follow, these increased matrices form a convex hull, in which additional input rate matrices can be found satisfying the condition of being outside of $\Lambda_{\rm{REP}}$ but in $\Lambda_{\rm{RMIA}}$.\\

Summarizing Subsection \ref{subsec_capacity_comparison_sub1} and Subsection \ref{subsec_capacity_comparison_sub2}, we conclude that any input rate matrix in $\Lambda_{\rm {REP}}$ is included by $\Lambda_{\rm {RMIA}}$, while there exists at least one input rate matrix $\left({\lambda'}_n^{\left(c\right)}\right)$ that is in $\Lambda_{\rm {RMIA}}$ but is outside of $\Lambda_{\rm {REP}}$. Therefore, $\Lambda_{\rm{REP}} \subset \Lambda_{\rm{RMIA}}$.

\section{Proof of Lemma \ref{lemma: characterize_metric_DIVBAR-RMIA}}
\label{appendix: single_epoch_opt_DIVBAR-RMIA}
Firstly, the expectation of the metric shown as (\ref{eq_single_epoch_metric_P}) given the backlog state ${\bf{\hat Q}}\left(u_{n,i}\right)$ can be upper bounded as follows:
\begin{align}
\mathbb{E}\left\{ {\left. {{Z_n}\left( {i,{\bf{\hat Q}}\left( {{u_{n,i}}} \right)} \right)} \right|{\bf{\hat{Q}}}\left( {{u_{n,i}}} \right)} \right\} &= \sum\limits_c {\mathbb{E}\left\{ {\left. {\sum\limits_{\tau  = {u_{n,i}}}^{{u_{n,i + 1}} - 1} {\sum\limits_{k \in {{\cal K}_n}} {b_{nk}^{\left( c \right)}\left( \tau  \right)\left[ {\hat Q_n^{\left( c \right)}\left( {{u_{n,i}}} \right) - \hat Q_k^{\left( c \right)}\left( {{u_{n,i}}} \right)} \right]} } } \right|{\bf{\hat Q}}\left( {{u_{n,i}}} \right)} \right\}}\nonumber\\
&\buildrel (a) \over \le \sum\limits_c {\mathbb{E}\left\{ {\left. {\sum\limits_{\tau  = {u_{n,i}}}^{{u_{n,i + 1}} - 1} {\sum\limits_{k \in {{\cal K}_n}} {b_{nk}^{\left( c \right)}\left( \tau  \right)\hat W_{nk}^{\left( c \right)}\left( {{u_{n,i}}} \right)} } } \right|{\bf{\hat Q}}\left( {{u_{n,i}}} \right)} \right\}}\nonumber\\
&\buildrel (b) \over = \sum\limits_c {\mathbb{E}\left\{ {\left. {\sum\limits_{k \in {{\cal K}_n}} {b_{nk}^{\left( c \right)}\left( {{u_{n,i + 1}} - 1} \right)\hat W_{nk}^{\left( c \right)}\left( {{u_{n,i}}} \right)} } \right|{\bf{\hat Q}}\left( {{u_{n,i}}} \right)} \right\}},
\label{eq_key_metric_over_single_epoch_RMIA_main1}
\end{align}
In (\ref{eq_key_metric_over_single_epoch_RMIA_main1}), the upper bound (a) is achieved in the case: $b_{nk}^{\left(c\right)}\left(\tau\right)=0$ when $W_{nk}^{\left(c\right)}\left(u_{n,i}\right)=0$, i.e., node $n$ never forwards a packet of commodity $c$ to node $k \in {\cal K}_n$ if node $k$ has non-positive differential backlog (zero differential backlog coefficient) of commodity $c$, which is consistent with the description in step \ref{step: forwarding_packet}) of the algorithm summary of DIVBAR-RMIA. The equality (b) in (\ref{eq_key_metric_over_single_epoch_RMIA_main1}) holds true because of the fact that, for any policy in set $\cal P$, $b_{nk}^{\left(c\right)}\left(\tau\right)=0$ when $\tau$ is not the first decoding timeslot of epoch $i$ for node $n$, i.e., $b_{nk}^{\left(c\right)}\left(\tau\right)=0$ when $u_{n,i}\le \tau<u_{n,i+1}-1$.

Additionally, we need to define several variables. Let $\mu_n\left(i\right)$ represent the variable that takes value 1 if node $n$ decides to transmit a commodity packet in its $i$th epoch, and takes value 0 if node $n$ decides to transmit a null packet; let $\mu_n^{\left(c\right)}\left(i\right)$ represent the variable that takes value 1 if node $n$ decides to transmit a packet of commodity $c$ in its $i$th epoch, and takes value 0 otherwise. Then we have
\begin{equation}
\sum\limits_c {\mu _n^{\left( c \right)}\left( i \right)}  = {\mu _n}\left( i \right) \le 1.
\label{eq_transmitting_decision}
\end{equation}
Define $X_{nk}^{{\cal P},\rm{RMIA}}\left(i\right)$ as the random variable that takes value 1 if node $k \in {\cal K}_n$ firstly decodes the packet transmitted by node $n$ under a policy within $\cal P$ and ends epoch $i$, and takes value $0$ otherwise. Here we use superscripts $\cal P$ and $\rm{RMIA}$ on $X_{nk}^{{\cal P},\rm{RMIA}}\left(i\right)$ to indicate that this variable is the same under all the policies with RMIA in $\cal P$, because all the policies in $\cal P$ have coherent epochs. The randomness of $X_{nk}^{{\cal P},\rm{RMIA}}\left(i\right)$ lies in the channel realizations of all the outgoing links of node $n$ in epoch $i$. Note that $X_{nk}^{{{\cal P},\rm{RMIA}}}\left( i \right)\mu _n^{\left( c \right)}\left( i \right)\in \left\{0,1\right\}$. Considering the fact that in the ending timeslot (first decoding timeslot) of epoch $i$, $b_{nk}^{\left(c\right)}\left(u_{n,i+1}-1\right)$ can be 1 only if $X_{nk}^{{\cal P},\rm{RMIA}}\left( i \right)\mu _n^{\left( c \right)}\left( i \right)=1$, we have the following relation:
\begin{equation}
b_{nk}^{\left( c \right)}\left( u_{n,i+1}-1  \right) = b_{nk}^{\left( c \right)}\left( u_{n,i+1}-1  \right)X_{nk}^{{\cal P},\rm{RMIA}}\left( i  \right)\mu _n^{\left( c \right)}\left( i \right).
\end{equation}

Then it follows from (\ref{eq_key_metric_over_single_epoch_RMIA_main1}) that
\begin{align}
&\ \ \ \ \mathbb{E}\left\{ {\left. {{Z_n}\left( {i,{\bf{\hat Q}}\left( {{u_{n,i}}} \right)} \right)} \right|{\bf{\hat Q}}\left( {{u_{n,i}}} \right)} \right\} \nonumber\\
&\le \sum\limits_c {\mathbb{E}\left\{ {\left. {\mu _n^{\left( c \right)}\left( i \right)\sum\limits_{k \in {{\cal K}_n}} {b_{nk}^{\left( c \right)}\left( {{u_{n,i + 1}} - 1} \right)X_{nk}^{{\cal P},{\rm{RMIA}}}\left( i \right)\hat W_{nk}^{\left( c \right)}\left( {{u_{n,i}}} \right)} } \right|{\bf{\hat Q}}\left( {{u_{n,i}}} \right)} \right\}}\nonumber\\
&= \sum\limits_c {\mathbb{E}\left\{ {\left. {\sum\limits_{k \in {{\cal K}_n}} {b_{nk}^{\left( c \right)}\left( {{u_{n,i + 1}} - 1} \right)X_{nk}^{{\cal P},{\rm{RMIA}}}\left( i \right)\hat W_{nk}^{\left( c \right)}\left( {{u_{n,i}}} \right)} } \right|{\bf{\hat Q}}\left( {{u_{n,i}}} \right),\mu _n^{\left( c \right)}\left( i \right) = 1} \right\}\mathbb{E}{\left\{ {\left. {\mu _n^{\left( c \right)}\left( i \right)} \right|{\bf{\hat Q}}\left( {{u_{n,i}}} \right)} \right\}}}\nonumber\\
&\buildrel (a) \over \le \sum\limits_c {\mathbb{E}\left\{ {\left. {\mathop {\max }\limits_{k \in {{\cal K}_n}} \left\{ {X_{nk}^{{{\cal P},\rm{RMIA}}}\left( i \right)\hat W_{nk}^{\left( c \right)}\left( {{u_{n,i}}} \right)} \right\}} \right|{\bf{\hat Q}}\left( {{u_{n,i}}} \right),\mu _n^{\left( c \right)}\left( i \right) = 1} \right\}\mathbb{E}{\left\{ {\left. {\mu _n^{\left( c \right)}\left( i \right)} \right|{\bf{\hat Q}}\left( {{u_{n,i}}} \right)} \right\}}}.\ \ \ \ \ \ \ \ \
\label{eq_key_metric_over_single_epoch_RMIA_main2}
\end{align}
The inequality (a) in (\ref{eq_key_metric_over_single_epoch_RMIA_main2}) holds true because $b_{nk}^{\left(c\right)}\left(u_{n,i+1}-1\right)\ge 0$ and $\sum\nolimits_{k \in {{\cal K}_n}} {b_{nk}^{\left( c \right)}\left( {{u_{n,i + 1}} - 1} \right) \le 1}$; the inequality (a) becomes an equality in the case: $b_{nk}^{\left(c\right)}\left(u_{n,i+1}-1\right)=1$ only if node $k$ has the largest positive term $X_{nk}^{{\cal P},\rm{RMIA}}\left(i\right)W_{nk}^{\left(c\right)}\left(i\right)$, i.e., node $n$ forwards a packet to node $k$ only if node $k$ is the successful receiver with the largest differential backlog of commodity $c$, which is the same as step \ref{step: forwarding_packet}) of the algorithm summary of DIVBAR-RMIA.

In order to further compute the metric in the upper bound of (\ref{eq_key_metric_over_single_epoch_RMIA_main2}), firstly let $\hat 1_{nk,i}^{\left(c\right)}$ represent the indicator function that takes value 1 if and only if $X_{nk}^{{\cal P},\rm{RMIA}}=1$, while $X_{nj}^{{\cal P},\rm{RMIA}}=0$ for all $j \in \hat {\cal R}_{nk}^{{\rm{high}},\left(c\right)}\left(u_{n,i}\right)$, where $\hat {\cal R}_{nk}^{{\rm{high}},\left(c\right)}\left(u_{n,i}\right)$ is determined by the backlog state ${\bf{\hat Q}}\left(u_{n,i}\right)$ and is defined in step \ref{step: define_hiararchy_set_RMIA_main}) of the algorithm summary of DIVBAR-RMIA. In other words, $\hat 1_{nk,i}^{\left(c\right)}$ indicates the event that, when node $n$ is transmitting, node $k$ has the largest differential backlog coefficient $\hat W_{nk}^{\left(c\right)}\left(u_{n,i}\right)$ among the receivers in the first successful receiver set. Then the largest differential backlog coefficient of the first successful receiver set can be expressed as follows:
\begin{equation}
\mathop {\max }\limits_{k \in {{\cal K}_n}} \left\{ {X_{nk}^{{\cal P},{\rm{RMIA}}}\left( i\right)\hat W_{nk}^{\left( c \right)}\left( {{u_{n,i}}} \right)} \right\} = \sum\limits_{k \in {{\cal K}_n}} {\hat W_{nk}^{\left( c \right)}\left( {{u_{n,i}}} \right)\hat 1_{nk,i}^{\left( c \right)}},
\end{equation}
and its conditional expectation given the backlog state observation ${\bf{\hat{Q}}}\left(u_{n,i}\right)$ and that a packet of a commodity $c$ is transmitted by node $n$ in epoch $i$ can be computed as follows:
\begin{equation}
\mathbb{E}\left\{ {\left. {\mathop {\max }\limits_{k \in {{\cal K}_n}} \left\{ {X_{nk}^{{\cal P},{\rm{RMIA}}}\left( i \right)\hat W_{nk}^{\left( c \right)}\left( {{u_{n,i}}} \right)} \right\}} \right|{\bf{\hat Q}}\left( {{u_{n,i}}} \right),\mu _n^{\left( c \right)}\left( i \right) = 1} \right\} = \sum\limits_{k \in {{\cal K}_n}} {\hat W_{nk}^{\left( c \right)}\left( {{u_{n,i}}} \right)\hat \varphi _{nk}^{\left( c \right)}\left( i \right)},
\label{eq_metric_single_epoch_choose_commodity_2}
\end{equation}
where $\hat \varphi _{nk}^{\left( c \right)}\left( i \right)$ is defined in (\ref{eq_backpressure_metric_main}) in step \ref{step: compute_metric_DIVBAR-RMIA_main}) of the algorithm description of DIVBAR-RMIA and represents the probability that $\hat 1_{nk,i}^{\left(c\right)}=1$, given the backlog state observation ${\bf{\hat{Q}}}\left(u_{n,i}\right)$ and the decision that a packet of commodity $c$ is transmitted by node $n$ in epoch $i$. This computation result is the same as the backpressure metric expression for each commodity $c$ over a single epoch under DIVBAR-RMIA shown as (\ref{eq_backpressure_metric_main}) in step \ref{step: compute_metric_DIVBAR-RMIA_main}) of its algorithm summary.

Then letting $\hat c_n\left(i\right)$ represent the commodity that maximizes the metric in (\ref{eq_metric_single_epoch_choose_commodity_2}), and denoting the corresponding metric value as $\hat \Xi_n\left(i\right)$, we plug (\ref{eq_metric_single_epoch_choose_commodity_2}) into (\ref{eq_key_metric_over_single_epoch_RMIA_main2}) to get the following upper bound shown as (\ref{eq_metric_single_epoch_comparison_DIVBAR-RMIA_main}):
\begin{align}
\mathbb{E}\left\{ {\left. {{Z_n}\left( {i,{\bf{\hat Q}}\left( {{u_{n,i}}} \right)} \right)} \right|{\bf{\hat Q}}\left( {{u_{n,i}}} \right)} \right\} &\buildrel (a) \over \le \sum\limits_{k \in {{\cal K}_n}} {\hat W_{nk}^{\left( {{{\hat c}_n}\left( i \right)} \right)}\left( {{u_{n,i}}} \right)\hat \varphi _{nk}^{\left( {{{\hat c}_n}\left( i \right)} \right)}\left( i \right)} \sum\limits_{c } {\mathbb{E}\left\{ {\left. {\mu _n^{\left( c \right)}\left( i \right)} \right|{\bf{\hat Q}}\left( {{u_{n,i}}} \right)} \right\}}\nonumber\\
&= {{\hat \Xi }_n}\left( i \right)\mathbb{E}\left\{ {\left. {{\mu _n}\left( i \right)} \right|{\bf{\hat Q}}\left( {{u_{n,i}}} \right)} \right\}\nonumber\\
&\buildrel (b) \over \le {{\hat \Xi }_n}\left( i \right).
\label{eq_metric_single_epoch_choose_commodity_3}
\end{align}
The inequality (a) in (\ref{eq_metric_single_epoch_choose_commodity_3}) becomes an equality in the case that $\mu_n^{\left(\hat c_n\right)}\left(i\right)=1$ if and only if $\mu_n\left(i\right)=1$, i.e., node $n$ only transmit the packet whose commodity maximizes the metric of (\ref{eq_metric_single_epoch_choose_commodity_2}) if it decides to transmit a commodity packet. The upper bound condition of (b) in (\ref{eq_metric_single_epoch_choose_commodity_3}) can be achieved by setting $\mu_n\left(i\right)=1$ if $\hat \Xi_n\left(i\right)>0$ and setting $\mu_n\left(i\right)=0$ if $\hat \Xi_n\left(i\right)=0$. These two cases that achieve the upper bounds (a) and (b) can be realized by implementing step \ref{step: choose_commodity_to_transmit_RMIA_main}) in the algorithm description of DIVBAR-RMIA.

In summary, from the upper bound achieving conditions of (\ref{eq_key_metric_over_single_epoch_RMIA_main1}), (\ref{eq_key_metric_over_single_epoch_RMIA_main2}), and (\ref{eq_metric_single_epoch_choose_commodity_3}), we can see that the backpressure metric $\mathbb{E}\left\{ {\left. {{Z_n}\left( {i,{\bf{\hat Q}}\left( {{u_{n,i}}} \right)} \right)} \right|{\bf{\hat Q}}\left( {{u_{n,i}}} \right)} \right\}$ over a single epoch under a policy within $\cal P$ can be maximized if the policy is chosen as DIVBAR-RMIA.

\section{Proof of Theorem \ref{thm: DIVBAR_MIA_throughput_optimal}}
\label{appendix: DIVBAR-RMIA_optimallity}
Denote $\hat{Policy}$ as the DIVBAR-RMIA policy and any variable uniquely specified by $\hat{Policy}$ is denoted in the form $\hat{x}$. To achieve strong stability, Ref. \cite{Neely_Rahul_DIVBAR_2009} analyzes the one-timeslot Lyapunov drift on both DIVBAR policy and stationary randomized policy with REP. Then the strong stability can be shown through the comparison between the upper bounds of Lyapunov drifts under the two policies. The Lyapunov drift analysis with RMIA assumption has some similarity with that in the REP case. However, as explained in Section \ref{sec: network_capacity_region_RMIA}, with RMIA, $d$-timeslot average Lyapunov drift needs to be analyzed to compare the two policies, where $d$ is sufficiently large. Additionally, the stationary randomized policy $Policy^*$ that can support all input rate matrices within $\Lambda_{\rm{RMIA}}$ is not directly used in our proof of strong stability with RMIA, instead, the comparison of the upper bound of the $d$-slot Lyapunov drift is between $\hat{Policy}$ and a modified version of the stationary randomized policy: ${Policy'}^*$, which satisfies the following properties: \emph{it is the same as DIVBAR-RMIA in the interval from timeslot $0$ to timeslot $t_0-1$; starting from timeslot $t_0$ without using the pre-accumulated partial information, it is the same as $Policy^*$ starting from timeslot $0$.}

Moreover, in order to facilitate the comparison between $\hat{Policy}$ and ${Policy'}^*$, an intermediate policy $\tilde{Policy}$ is introduced, which is the intermediate policy with the following properties: \emph{it is the same as DIVBAR-RMIA in the interval from timeslot $0$ to timeslot $t_0-1$; starting from timeslot $t_0$ without using the pre-accumulated partial information, each node chooses the commodity to transmit according to the maximization of a backpressure metric, and keeps transmitting the packets of the chosen commodity in later timeslots with RMIA, and forwards each decoded packet to the successful receiver with the largest positive differential backlog observed in timeslot $t_0$.} The epochs for each node $n$ under $\tilde{Policy}$ are shown by Fig. \ref{fig_epoch relations among several policies} in Appendix \ref{appendix: policy_list}.

We use four steps to compare the upper bounds of Lyapunov drift under DIVBAR-RMIA and under the stationary randomized policy over a $D$ timeslots interval.
\begin{enumerate}[1)]
\item Transform the comparison on the upper bounds of Lyapunov drift under $\hat{Policy}$ and ${Policy'}^*$ to the comparison on the key backpressure metrics under the two policies, which is shown in Subsection \ref{subsec_througput_optimal_sub1}.
\item Compare the key metrics under $\hat{Policy}$ and $\tilde{Policy}$, which is shown in Subsection \ref{subsec_througput_optimal_sub2}.
\item Compare the key metrics under $\tilde{Policy}$ and ${Policy'}^*$, which is shown in Subsection \ref{subsec_througput_optimal_sub3}.
\item Combine the results in 2) and 3) and get the conclusion of strong stability, which is shown in Subsection \ref{subsec_througput_optimal_sub4}.
\end{enumerate}

\subsection{Transforming the comparison on the upper bound of Lyapunov drift to the comparison on the key backpressure metric}
\label{subsec_througput_optimal_sub1}
Similar to the derivations from (\ref{eq_new_D_step_queuing_dynamics})-(\ref{eq_corrl_lyapunov_drift1}) in Section \ref{sec: network_capacity_region_RMIA}, the $t$-timeslot Lyapunov drift under $\hat{Policy}$ is
\begin{align}
&\ \ \ \ \frac{1}{t}\sum\limits_{n,c} \mathbb{E} \left\{ {\left. {{{\left( {\hat Q_n^{\left( c \right)}\left( {{t_0} + t} \right)} \right)}^2} - {{\left( {\hat Q_n^{\left( c \right)}\left( {{t_0}} \right)} \right)}^2}} \right|{\bf{\hat Q}}\left( {{t_0}} \right)} \right\}\nonumber\\
& = B\left(t\right) + \frac{2}{t}\sum\limits_{n,c} {\hat Q_n^{\left( c \right)}\left( {{t_0}} \right)\sum\limits_{\tau  = {t_0}}^{{t_0} + t - 1} {\mathbb{E}\left\{ {a_n^{\left( c \right)}\left( \tau  \right)} \right\}} } -\frac{2}{t}\sum\limits_{n,c} {\hat Q_n^{\left( c \right)}\left( {{t_0}} \right)\mathbb{E}\left\{ {\left. {\sum\limits_{\tau  = {t_0}}^{{t_0} + t - 1} {\left[ {\sum\limits_{k \in {{\cal K}_n}} {\hat b_{nk}^{\left( c \right)}\left( \tau  \right)}  - \sum\limits_{k \in {{\cal K}_n}} {\hat b_{kn}^{\left( c \right)}\left( \tau  \right)} } \right]} } \right|{\bf{\hat Q}}\left( {{t_0}} \right)} \right\}},
\label{eq_DIVBAR_RMIA_Lyapunov_drift}
\end{align}
where $B\left(t\right)=N^2t\left[1+\left(N+A_{\rm{max}}\right)^2\right]$ as is defined by (\ref{eq_B_value}) in the proof of Corollary \ref{cap_region_corr}. Since the following equation holds true for networks:
\begin{equation}
\sum\limits_{n,c} {\hat Q_n^{\left( c \right)}\left( {{t_0}} \right)\sum\limits_{\tau  = {t_0}}^{{t_0} + t - 1} {\left[ {\sum\limits_{k \in {{\cal K}_n}} {\hat b_{nk}^{\left( c \right)}\left( \tau  \right)}  - \sum\limits_{k \in {{\cal K}_n}} {\hat b_{kn}^{\left( c \right)}\left( \tau  \right)} } \right]} }  = \sum\limits_{n,c} {\sum\limits_{\tau  = {t_0}}^{{t_0} + t - 1} {\sum\limits_{k \in {{\cal K}_n}} {\hat b_{nk}^{\left( c \right)}\left( \tau  \right)\left[ {\hat Q_n^{\left( c \right)}\left( {{t_0}} \right) - \hat Q_k^{\left( c \right)}\left( {{t_0}} \right)} \right]} } },
\end{equation}
then (\ref{eq_DIVBAR_RMIA_Lyapunov_drift}) becomes
\begin{align}
&\ \ \ \ \frac{1}{t}\sum\limits_{n,c} \mathbb{E} \left\{ {\left. {{{\left( {\hat Q_n^{\left( c \right)}\left( {{t_0} + t} \right)} \right)}^2} - {{\left( {\hat Q_n^{\left( c \right)}\left( {{t_0}} \right)} \right)}^2}} \right|{\bf{\hat Q}}\left( {{t_0}} \right)} \right\}\nonumber\\
& \le B\left(t\right) + \frac{2}{t}\sum\limits_{n,c} {\hat Q_n^{\left( c \right)}\left( {{t_0}} \right)\sum\limits_{\tau  = {t_0}}^{{t_0} + t - 1} {\mathbb{E}\left\{ {a_n^{\left( c \right)}\left( \tau  \right)} \right\}} } - \frac{2}{t}\sum\limits_{n,c} {\mathbb{E}\left\{ {\left. {\sum\limits_{\tau  = {t_0}}^{{t_0} + t - 1} {\sum\limits_{k \in {{\cal K}_n}} {\hat b_{nk}^{\left( c \right)}\left( \tau  \right)\left[ {\hat Q_n^{\left( c \right)}\left( {{t_0}} \right) - \hat Q_k^{\left( c \right)}\left( {{t_0}} \right)} \right]} } } \right|{\bf{\hat Q}}\left( {{t_0}} \right)} \right\}}\nonumber\\
&\buildrel \Delta \over = B\left(t\right) + \frac{2}{t}\sum\limits_{n,c} {\hat Q_n^{\left( c \right)}\left( {{t_0}} \right)\sum\limits_{\tau  = {t_0}}^{{t_0} + t - 1} {\mathbb{E}\left\{ {a_n^{\left( c \right)}\left( \tau  \right)} \right\}} } - 2\sum\limits_n {\mathbb{E}\left\{ {\left. {\left. {{{\hat Z}_n}\left( {{\bf{\hat Q}}\left( {{t_0}} \right)} \right)} \right|_{{t_0}}^{{t_0} + t - 1}} \right|{\bf{\hat Q}}\left( {{t_0}} \right)} \right\}},
\label{eq_DIVBAR_RMIA_Lyapunov_drift1}
\end{align}
where, under an arbitrary policy, the summation metric $\left. {{{ Z}_n}\left( {{\bf{ \hat Q}}\left( {{t_0}} \right)} \right)} \right|_{{t_0}}^{{t_0} + t - 1}$ represents the following expression:
\begin{equation}
\left. {{Z_n}\left( {{\bf{ \hat Q}}\left( {{t_0}} \right)} \right)} \right|_{{t_0}}^{{t_0} + t - 1} = \frac{1}{t}\sum\limits_{\tau  = {t_0}}^{{t_0} + t - 1} {\sum\limits_c {\sum\limits_{k \in {{\cal K}_n}} {b_{nk}^{\left( c \right)}\left( \tau  \right)\left[ { \hat Q_n^{\left( c \right)}\left( {{t_0}} \right) -  \hat Q_k^{\left( c \right)}\left( {{t_0}} \right)} \right]} } }.
\label{eq_general_backpressure_metric}
\end{equation}

In order to facilitate the comparison of the values of the key metric $\sum\limits_n {\mathbb{E}\left\{ {\left. {\left. {{Z_n}\left( {{\bf{\hat Q}}\left( {{t_0}} \right)} \right)} \right|_{{t_0}}^{{t_0} + t - 1}} \right|{\bf{\hat Q}}\left( {{t_0}} \right)} \right\}}$ respectively under $\hat{Policy}$ and ${Policy'}^*$, we introduce the intermediate policy: $\tilde{Policy}$, as is described just before this subsection and in Appendix \ref{appendix: policy_list}. $\tilde{Policy}$ can be shown to maximize the key metric $\sum\limits_n {\mathbb{E}\left\{ {\left. {\left. {{Z_n}\left( {{\bf{\hat Q}}\left( {{t_0}} \right)} \right)} \right|_{{t_0}}^{{t_0} + t - 1}} \right|{\bf{\hat Q}}\left( {{t_0}} \right)} \right\}}$ and serves as a "bridge" connecting $\hat{Policy}$ and ${Policy'}^*$. Therefore, the proof proceeds into the following two steps respectively shown as Subsection \ref{subsec_througput_optimal_sub2} and Subsection \ref{subsec_througput_optimal_sub3}.

\subsection{Comparison on the key backpressure metric between $\tilde{Policy}$ and ${Policy'}^*$}
\label{subsec_througput_optimal_sub2}
In this part of proof, both policies $\tilde{Policy}$ and ${Policy'}^*$ are analyzed on the interval starting from timeslot $t_0$ to timeslot $t_0+t-1$ (denoted as $\left[t_0,t_0+t-1\right]$) because they are the same before timeslot $t_0$. As we know, in the implementation of DIVBAR-RMIA, the partial information accumulated at the receivers by the beginning of timeslot $t_0$ may not be zero, but this issue will be dealt with in Subsection \ref{subsec_througput_optimal_sub3}.

The proof in this section consists of two parts: first to show that $\tilde{Policy}$ has the best forwarding strategy, according to which the optimal forwarding node are confirmed; second to compare the transmitting strategies of $\tilde{Policy}$ and $Policy^*$, which concerns the strategy of choosing commodities to transmit.

\subsubsection{\textbf{Confirming the best forwarding strategy}}
Define $\alpha _n^{\left( c \right)}\left( {{t_0},t} \right)$ as the number of timeslots used to transmit commodity $c$ packets by node $n$ within the interval from timeslot $t_0$ to $t_0+t-1$; let $\left\{\tau_j^{\left(c\right)}\right\}$ represent the subsequence of the timeslots when node $n$ is transmitting commodity $c$ packets. Then based on (\ref{eq_general_backpressure_metric}), it follows that
\begin{align}
&\ \ \ \ \sum\limits_n {\mathbb{E}\left\{ {\left. {\left. {{{ Z}_n}\left( {{\bf{\hat Q}}\left( {{t_0}} \right)} \right)} \right|_{{t_0}}^{{t_0} + t - 1}} \right|{\bf{\hat Q}}\left( {{t_0}} \right)} \right\}}\nonumber\\
&= \sum\limits_n {\sum\limits_c {\mathbb{E}\left\{ {\left. {\frac{{\alpha _n^{\left( c \right)}\left( {{t_0},t} \right)}}{t}\frac{1}{{\alpha _n^{\left( c \right)}\left( {{t_0},t} \right)}}\sum\limits_{j = 0}^{\alpha _n^{\left( c \right)}\left( {{t_0},t} \right) - 1} {\sum\limits_{k \in {{\cal K}_n}} {b_{nk}^{\left( c \right)}\left( {\tau _j^{\left( c \right)}} \right)\left[ {\hat Q_n^{\left( c \right)}\left( {{t_0}} \right) - \hat Q_k^{\left( c \right)}\left( {{t_0}} \right)} \right]} } } \right|{\bf{\hat Q}}\left( {{t_0}} \right)} \right\}} }.
\label{eq_backpressure_inequality2}
\end{align}

Additionally, based on the definition of the epoch shown as Fig. \ref{fig_slots_allocation} in Section \ref{sec: network_model}, define the number of timeslots in the $i$th epoch of commodity $c$ being counted from timeslot $t_0$ as $T_n^{\left(c\right)}\left(i\right)$. For transmitting node $n$ and with the assumption of not using the pre-accumulated partial information by timeslot $t_0$, define $\overline M_n^{\left(c\right)}\left(t_0,t\right)$ as the number of epochs of commodity $c$ for node $n$ that are entirely located within the interval $\left[t_0,t_0+t-1\right]$, i.e.,
\begin{equation}
\overline M_n^{\left( c \right)}\left( {{t_0},t} \right) = \max \left\{ {m:\sum\limits_{i = 1}^m {{T_n^{\left(c\right)}}\left( i \right)}  \le \alpha _n^{\left( c \right)}\left( {{t_0},t} \right)} \right\}.
\label{eq_definition_M}
\end{equation}

In the timeslots of transmitting commodity $c$ packets, each epoch for node $n$ consists of timeslots contiguous in the subsequence $\left\{\tau_j^{\left(c\right)}\right\}$. Based on this, define $X_{nk}^{{\rm{RMIA}}, \left(c\right)}\left(i\right)$ as the random variable that takes value $1$ if the receiving node $k \in {\cal K}_n$ decodes the packet at the end of the $i$th epoch of commodity $c$ under a policy with RMIA; and takes value $0$ otherwise. Then with the similar derivation strategy as in the proof of Lemma \ref{lemma: characterize_metric_DIVBAR-RMIA}, (\ref{eq_backpressure_inequality2}) can be upper bounded as follows:
\begin{align}
&\ \ \ \ \sum\limits_n {\mathbb{E}\left\{ {\left. {\left. {{Z_n}\left( {{\bf{\hat Q}}\left( {{t_0}} \right)} \right)} \right|_{{t_0}}^{{t_0} + t - 1}} \right|{\bf{\hat Q}}\left( {{t_0}} \right)} \right\}}\nonumber\\
&=\sum\limits_n {\sum\limits_c {\mathbb{E}\left\{ {\left. {\frac{{\alpha _n^{\left( c \right)}\left( {{t_0},t} \right)}}{t}\frac{1}{{\alpha _n^{\left( c \right)}\left( {{t_0},t} \right)}}\sum\limits_{i = 1}^{\overline M_n^{\left( c \right)}\left( {{t_0},t} \right)} {\sum\limits_{j = u_{n,i}^{\left( c \right)}}^{u_{n,i + 1}^{\left( c \right)} - 1} {\sum\limits_{k \in {{\cal K}_n}} {b_{nk}^{\left( c \right)}\left( {\tau _j^{\left( c \right)}} \right)\left[ {\hat Q_n^{\left( c \right)}\left( {{t_0}} \right) - \hat Q_k^{\left( c \right)}\left( {{t_0}} \right)} \right]} } } } \right|{\bf{\hat Q}}\left( {{t_0}} \right)} \right\}} }\nonumber\\
&\le \sum\limits_n {\sum\limits_c {\mathbb{E}\left\{ {\left. {\frac{{\alpha _n^{\left( c \right)}\left( {{t_0},t} \right)}}{t}\frac{1}{{\alpha _n^{\left( c \right)}\left( {{t_0},t} \right)}}\sum\limits_{i = 1}^{\overline M_n^{\left( c \right)}\left( {{t_0},t} \right)} {\mathop {\max }\limits_{k \in {{\cal K}_n}} \left\{ {X_{nk}^{{\rm{RMIA,}}\left( c \right)}\left( i \right)\hat W_{nk}^{\left( c \right)}\left( {{t_0}} \right)} \right\}} } \right|{\bf{\hat Q}}\left( {{t_0}} \right)} \right\}} },\ \ \ \ \ \ \ \ \ \ \
\label{eq_backpressure_inequality3}
\end{align}
where $\hat W_{nk}^{\left( c \right)}\left( {{t_0}} \right) = \max \left\{ {\hat Q_n^{\left( c \right)}\left( {{t_0}} \right) - \hat Q_k^{\left( c \right)}\left( {{t_0}} \right),0} \right\}$; $u_{n,i}^{\left(c\right)}$ is the starting timeslot of the $i$th epoch of commodity $c$ for node $n$ being counted from timeslot $t_0$; $\overline M_n^{\left(c\right)}\left(t_0,t\right)$ can be $0$ and correspondingly, the whole summation metric can be $0$. In (\ref{eq_backpressure_inequality3}), the upper bound can be achieved if node $n$ only forwards the decoded packet to the successful receiver $k$ with the largest positive differential backlog $\hat W_{nk}^{\left( c \right)}\left( {{t_0}} \right)$, which is consistent with the forwarding strategy of $\tilde{Policy}$. Thus, $\tilde{Policy}$ has the optimal forwarding strategy in the sense of maximizing the key backpressure metric of each commodity.

\subsubsection{\textbf{Comparing the strategies of choosing commodities to transmit under $\tilde{Policy}$ and ${Policy'}^*$}}
Now the remaining part is to confirm the strategy of choosing commodities to transmit, which affects the value of $\alpha_n^{\left(c\right)}\left(t_0,t\right)$. We introduce another intermediate policy, denoted as \emph{$\tilde{Policy^*}$, with the following properties: it is the same as DIVBAR-RMIA in the interval from timeslot $0$ to timeslot $t_0-1$; starting from timeslot $t_0$ without using the pre-accumulated partial information, each node uses the same probabilities as $Policy^{*}$ to choose commodities to transmit but uses the same strategy as $\tilde{Policy}$ to forward the decoded packets.} The later proof logic is to firstly compare $\tilde{Policy^*}$ and ${Policy'}^*$, and then compare $\tilde{Policy}$ and $\tilde{Policy^*}$.

According to (\ref{eq_backpressure_inequality3}), the comparison of the backpressure metric values under $\tilde{Policy^*}$ and ${Policy'}^*$ is shown as follows:
\begin{align}
&\ \ \ \ \sum\limits_n {\mathbb{E}\left\{ {\left. {\left. {{{ \tilde {Z}^*}_n}\left( {{\bf{\hat Q}}\left( {{t_0}} \right)} \right)} \right|_{{t_0}}^{{t_0} + t - 1}} \right|{\bf{\hat Q}}\left( {{t_0}} \right)} \right\}}\nonumber\\
&=\sum\limits_n {\sum\limits_c {\mathbb{E}\left\{ {\left. {\frac{{\alpha _n^{*\left( c \right)}\left( {{t_0},t} \right)}}{t}\frac{1}{{\alpha _n^{*\left( c \right)}\left( {{t_0},t} \right)}}\sum\limits_{i = 1}^{\overline M_n^{\left( c \right)}\left( {{t_0},t} \right) } {\mathop {\max }\limits_{k \in {{\cal K}_n}} \left\{ {{X^{\rm{RMIA}}_{nk}}\left( {\tau _{{u^{\left( c \right)}_{n,i + 1}} - 1}^{\left( c \right)}} \right)\hat W_{nk}^{\left( c \right)}\left( {{t_0}} \right)} \right\}} } \right|{\bf{\hat Q}}\left( {{t_0}} \right)} \right\}} }\nonumber\\
&\ge \sum\limits_n {\sum\limits_c {\mathbb{E}\left\{ {\left. {\frac{{\alpha _n^{*\left( c \right)}\left( {{t_0},t} \right)}}{t}\frac{1}{{\alpha _n^{*\left( c \right)}\left( {{t_0},t} \right)}}\sum\limits_{i = 1}^{\overline M_n^{\left( c \right)}\left( {{t_0},t} \right)} {\sum\limits_{k \in {{\cal K}_n}} {{b'}_{nk}^{*\left( c \right)}\left( {\tau _{u_{n,i + 1}^{\left( c \right)} - 1}^{\left( c \right)}} \right)\left[ {\hat Q_n^{\left( c \right)}\left( {{t_0}} \right) - \hat Q_k^{\left( c \right)}\left( {{t_0}} \right)} \right]} } } \right|{\bf{\hat Q}}\left( {{t_0}} \right)} \right\}} }\nonumber\\
&= \sum\limits_n {\mathbb{E}\left\{ {\left. {\left. {{{Z'}^*_n}\left( {{\bf{\hat Q}}\left( {{t_0}} \right)} \right)} \right|_{{t_0}}^{{t_0} + t - 1}} \right|{\bf{\hat Q}}\left( {{t_0}} \right)} \right\}}.
\label{eq_compare_inter_vs_stationary}
\end{align}

The result of (\ref{eq_compare_inter_vs_stationary}) demonstrates that the key metric value under $\tilde{Policy^*}$ is no less than the key metric value under ${Policy'}^*$ over the interval from timeslot $t_0$ to timeslot $t_0+t-1$, which finishes the comparison between the two polices.

The next step is to compare the key metrics under $\tilde{Policy}$ and $\tilde{Policy^*}$. Consider that both $\tilde{Policy}$ and $\tilde {Policy^*}$ use fixed probabilities to choose commodities to transmit in each timeslot ($\tilde{Policy}$ keeps transmitting a single commodity during the $t$ timeslots, which is equivalent to choosing this commodity to transmit with probability 1). To facilitate the later proof, we restrict the policy set to the one, denoted as $\cal Y$, \emph{which consists of the policies that are the same as DIVBAR-RMIA from timeslot $0$ to timeslot $t_0-1$ and, from timeslot $t_0$ without using the pre-accumulated partial information, use all possible fixed probabilities to choose commodities to transmit and use backpressure strategy to forward the decoded packet.} Note that both $\tilde{Policy}$ or $\tilde{Policy}^*$ belong to $\cal Y$. Let $Z_n^{\left(c\right)}\left(i,{{\bf{\hat Q}}\left( {{t_0}} \right)}\right)$ represent the backpressure metric value for node $n$ over a single epoch of commodity $c$ under a policy with backpressure forwarding strategy, i.e.,
\begin{equation}
Z_n^{\left( c \right)}\left( {i,{\bf{\hat Q}}\left( {{t_0}} \right)} \right) = \mathop {\max }\limits_{k \in {{\cal K}_n}} \left\{ {{X^{\rm{RMIA}}_{nk}}\left( {\tau _{{u_{n,i + 1}} - 1}^{\left( c \right)}} \right)\hat W_{nk}^{\left( c \right)}\left( {{t_0}} \right)} \right\},
\label{eq_metric_single_epoch_arbitrary_core}
\end{equation}
where $i$ is the epoch's index. Therefore, the metric expression shown as (\ref{eq_metric_single_epoch_arbitrary_core}) over a single epoch for each node is valid for each policy in $\cal Y$. With the renewal operation, the value of $Z_n^{\left( c \right)}\left( {i,{\bf{\hat Q}}\left( {{t_0}} \right)} \right)$ under each policy in $\cal Y$ only depends on the backlog state ${{\bf{\hat Q}}\left( {{t_0}} \right)}$ and the channel realizations in the $i$th epoch for node $n$, and therefore $\left\{Z_n^{\left( c \right)}\left( {i,{\bf{\hat Q}}\left( {{t_0}} \right)} \right): i\ge 1\right\}$ are i.i.d. across epochs. Thus, to simplify the notation, for an arbitrary policy within $\cal Y$, we can safely use the following notation:
\begin{equation}
\mathbb{E}\left\{ {\left. {Z_n^{\left( c \right)}\left( {i,{\bf{\hat Q}}\left( {{t_0}} \right)} \right)} \right|{\bf{\hat Q}}\left( {{t_0}} \right)} \right\} \buildrel \Delta \over = z_n^{\left( c \right)}\left( {{\bf{\hat Q}}\left( {{t_0}} \right)} \right),
\label{eq_metric_single_epoch_arbitrary}
\end{equation}
Moreover, for the same commodity $c$, since the policies in $\cal Y$ have the same forwarding strategy, even $Z_n^{\left( c \right)}\left( {i,{\bf{\hat Q}}\left( {{t_0}} \right)} \right)$ under different policies in $\cal Y$ have the same distribution. Therefore, it is unnecessary to use any notation on $z_n^{\left( c \right)}\left( {{\bf{\hat Q}}\left( {{t_0}} \right)} \right)$ to specify the policy being used, as long as the policy is in $\cal Y$.

Under a policy within $\cal Y$, (\ref{eq_backpressure_inequality3}) can be written as
\begin{equation}
\sum\limits_n {\mathbb{E}\left\{ {\left. {\left. {{{ Z}_n}\left( {{\bf{\hat Q}}\left( {{t_0}} \right)} \right)} \right|_{{t_0}}^{{t_0} + t - 1}} \right|{\bf{\hat Q}}\left( {{t_0}} \right)} \right\}}=\sum\limits_n {\sum\limits_c {\mathbb{E}\left\{ {\left. {\frac{{\alpha _n^{\left( c \right)}\left( {{t_0},t} \right)}}{t}\frac{{\frac{1}{{\overline M_n^{\left( c \right)}\left( {{t_0},t} \right) }}\sum\limits_{i = 1}^{\overline M_n^{\left( c \right)}\left( {{t_0},t} \right)} {Z_n^{\left( c \right)}\left( {i,{\bf{\hat Q}}\left( {{t_0}} \right)} \right)} }}{{\frac{{\alpha _n^{\left( c \right)}\left( {{t_0},t} \right)}}{{\overline M_n^{\left( c \right)}\left( {{t_0},t} \right) }}}}} \right|{\bf{\hat Q}}\left( {{t_0}} \right)} \right\}} }.
\label{eq_backpressure_inequality4}
\end{equation}
Since each node $n$ has a fixed probability $\alpha_n^{\left(c\right)}$ to choose commodity $c$ to transmit, according to the strong law of large numbers, we have
\begin{equation}
\mathop {\lim }\limits_{t \to \infty } \frac{{\alpha _n^{\left( c \right)}\left( {{t_0},t} \right)}}{t} = \alpha _n^{\left( c \right)}{\rm{\ with\ prob }}{\rm{.\ }}1.
\label{eq_convergence_alpha}
\end{equation}
Additionally, since ${\lim _{t \to \infty }}\overline M_n^{\left( c \right)}\left( {{t_0},t} \right) = \infty {\rm{ }}$ with prob. 1, and ${Z_n^{\left( c \right)}\left( {i,{\bf{\hat Q}}\left( {{t_0}} \right)} \right)}$ are i.i.d. over different epochs, we can also use the strong law of large numbers to get
\begin{align}
\mathop {\lim }\limits_{t \to \infty } \frac{1}{{\overline M_n^{\left( c \right)}\left( {{t_0},t} \right)}}\sum\limits_{i = 1}^{\overline M_n^{\left( c \right)}\left( {{t_0},t} \right) } {Z_n^{\left( c \right)}\left( {i,{\bf{\hat Q}}\left( {{t_0}} \right)} \right)} & = z_n^{\left( c \right)}\left( {{\bf{\hat Q}}\left( {{t_0}} \right)} \right){\rm{\ with\ prob}}{\rm{.\ }}1
\label{eq_convergence_Z}
\end{align}
Moreover, since $\sum\limits_{i = 1}^{\overline M_n^{\left( c \right)}\left( {{t_0},t} \right) } {{T_n^{\left(c\right)}}\left( i \right)}  \le \alpha _n^{\left( c \right)}\left( {{t_0},t} \right) < \sum\limits_{i = 1}^{\overline M_n^{\left( c \right)}\left( {{t_0},t} \right)+1} {{T_n^{\left(c\right)}}\left( i \right)}$, we can get
\begin{equation}
\frac{1}{{\overline M_n^{\left( c \right)}\left( {{t_0},t} \right)}}\sum\limits_{i = 1}^{\overline M_n^{\left( c \right)}\left( {{t_0},t} \right)} {{T_n^{\left(c\right)}}\left( i \right)}  \le \frac{{\alpha _n^{\left( c \right)}\left( {{t_0},t} \right)}}{{\overline M_n^{\left( c \right)}\left( {{t_0},t} \right)}} < \frac{{\overline M_n^{\left( c \right)}\left( {{t_0},t} \right)+1}}{{\overline M_n^{\left( c \right)}\left( {{t_0},t} \right) }}\cdot \frac{1}{{\overline M_n^{\left( c \right)}\left( {{t_0},t} \right)+1}}\sum\limits_{i = 1}^{\overline M_n^{\left( c \right)}\left( {{t_0},t} \right)+1} {{T_n^{\left(c\right)}}\left( i \right)}.
\label{eq_convergence_T1}
\end{equation}
According to the strong law of large numbers, we have
\begin{equation}
\mathop {\lim }\limits_{t \to \infty } \frac{1}{{\overline M_n^{\left( c \right)}\left( {{t_0},t} \right)}}\sum\limits_{i = 1}^{\overline M_n^{\left( c \right)}\left( {{t_0},t} \right)} {{T_n^{\left(c\right)}}\left( i \right)}  =\mathbb{E}\left\{T_n^{\left(c\right)}\left(i\right)\right\} {\rm{\ with\ prob\ }}{\rm{. }}1
\end{equation}
\begin{equation}
\mathop {\lim }\limits_{t \to \infty } \frac{{\overline M_n^{\left( c \right)}\left( {{t_0},t} \right)+1}}{{\overline M_n^{\left( c \right)}\left( {{t_0},t} \right) }}\cdot \frac{1}{{\overline M_n^{\left( c \right)}\left( {{t_0},t} \right)+1}}\sum\limits_{i = 1}^{\overline M_n^{\left( c \right)}\left( {{t_0},t} \right)+1} {{T_n^{\left(c\right)}}\left( i \right)} =\mathbb{E}\left\{T_n^{\left(c\right)}\left(i\right)\right\}{\rm{\ with\ prob}}{\rm{.\ }}1
\label{eq_convergence_T2}
\end{equation}
combining (\ref{eq_convergence_T1})-(\ref{eq_convergence_T2}) and denoting $\mathbb{E}\left\{T_n\left(i\right)\right\} \buildrel \Delta \over = \mathbb{E}\left\{T_n\right\}$ ($\left\{T_n^{\left(c\right)}\left(i\right): i\ge 1, c \in {\cal N}\right\}$ are i.i.d.), it follows that
\begin{equation}
\mathop {\lim }\limits_{t \to \infty } \frac{{\alpha _n^{\left( c \right)}\left( {{t_0},t} \right)}}{{\overline M_n^{\left( c \right)}\left( {{t_0},t} \right)}} = \mathbb{E}\left\{ {{T_n}} \right\},\ {\rm{with\ prob.}}\ 1
\label{eq_convergence_T3}
\end{equation}

Given the backlog state ${{\bf{\hat Q}}\left( {{t_0}} \right)}$, plugging (\ref{eq_convergence_alpha}), (\ref{eq_convergence_Z}) and (\ref{eq_convergence_T3}) back into (\ref{eq_backpressure_inequality4}) yields
\begin{equation}
\mathop {\lim }\limits_{t \to \infty } \sum\limits_n {\mathbb{E}\left\{ {\left. {\left. {{Z_n}\left( {{\bf{\hat Q}}\left( {{t_0}} \right)} \right)} \right|_{{t_0}}^{{t_0} + t - 1}} \right|{\bf{\hat Q}}\left( {{t_0}} \right)} \right\}}  = {\sum\limits_n {\sum\limits_c {\alpha _n^{\left( c \right)}\frac{{z_n^{\left( c \right)}\left( {{\bf{\hat Q}}\left( {{t_0}} \right)} \right)}}{{\mathbb{E}\left\{ {{T_n}} \right\}}}} } }.
\label{eq_backpressure_convergence}
\end{equation}
Since $\sum\limits_c {\alpha _n^{\left( c \right)}}  \le 1$, it follows that
\begin{equation}
\sum\limits_n {\sum\limits_c {\alpha _n^{\left( c \right)}\frac{{z_n^{\left( c \right)}\left( {{\bf{\hat Q}}\left( {{t_0}} \right)} \right)}}{{\mathbb{E}\left\{ {{T_n}} \right\}}}} }  \le \sum\limits_n {\frac{{\mathop {\max }\limits_c \left\{ {z_n^{\left( c \right)}\left( {{\bf{\hat Q}}\left( {{t_0}} \right)} \right)} \right\}}}{{\mathbb{E}\left\{ {{T_n}} \right\}}}} .
\label{eq_compare_convergence_limit}
\end{equation}
If defining $\tilde c = \mathop {\arg \max }\limits_c \left\{ {z_n^{\left( c \right)}\left( {{\bf{\tilde Q}}\left( {{t_0}} \right)} \right)} \right\}$, (\ref{eq_compare_convergence_limit}) becomes an equality when
\begin{equation}
\alpha _n^{\left( c \right)} = \left\{ {\begin{array}{*{20}{c}}
{1,{\rm{\ if\ }}c = \tilde c}\\
{0,{\rm{\ if\ }}c \ne \tilde c}
\end{array}} \right.,
\end{equation}
i.e., the equality holds true when node $n$ chooses commodity $\tilde c$ to transmit through the whole $t$ timeslots interval. As we know, $\tilde{Policy}$'s strategy of choosing commodity to transmit satisfies the equality condition in (\ref{eq_compare_convergence_limit}).

Now we start comparing $\tilde{Policy}$ and $\tilde{Policy^*}$. First, under $\tilde {Policy^*}$, due to (\ref{eq_backpressure_convergence}), for $\forall \varepsilon >0$, there exists an integer $\tilde{D}^*$ such that whenever $t\geq \tilde{D}^*$, we have
\begin{equation}
\left| {\sum\limits_n {\mathbb{E}\left\{ {\left. {\left. {{\tilde Z^*_n}\left( {{\bf{\hat Q}}\left( {{t_0}} \right)} \right)} \right|_{{t_0}}^{{t_0} + t - 1}} \right|{\bf{\hat Q}}\left( {{t_0}} \right)} \right\}}  - \sum\limits_n {\sum\limits_c {\alpha _n^{*\left( c \right)}\frac{{z_n^{\left( c \right)}\left( {{\bf{\hat Q}}\left( {{t_0}} \right)} \right)}}{{\mathbb{E}\left\{ {{T_n}} \right\}}}} } } \right|\le \frac{\varepsilon }{16}\sum\limits_{n,c} {\hat Q_n^{\left( c \right)}\left( {{t_0}} \right)}.
\label{eq_backpressure_deviation2}
\end{equation}
In (\ref{eq_backpressure_deviation2}), we choose ${{\varepsilon \sum\limits_{n,c} {\hat Q_n^{\left( c \right)}\left( {{t_0}} \right)} } \mathord{\left/
 {\vphantom {{\varepsilon \sum\limits_{n,c} {\hat Q_n^{\left( c \right)}\left( {{t_0}} \right)} } {16}}} \right.
 \kern-\nulldelimiterspace} {16}}$ as the deviation bound so as to guarantee that the value of $\tilde{D}^*$ does not depend on $t_0$. Similarly for $\tilde {Policy}$, there exists an integer $\tilde D_1$ such that, $\forall t \geq \tilde D_1$, we have
\begin{equation}
\left| {\sum\limits_n {\mathbb{E}\left\{ {\left. {\left. {{{\tilde Z}_n}\left( {{\bf{\hat Q}}\left( {{t_0}} \right)} \right)} \right|_{{t_0}}^{{t_0} + t - 1}} \right|{\bf{\hat Q}}\left( {{t_0}} \right)} \right\}}  - \sum\limits_n {\frac{{\mathop {\max }\limits_c \left\{ {z_n^{\left( c \right)}\left( {{\bf{\hat Q}}\left( {{t_0}} \right)} \right)} \right\}}}{{\mathbb{E}\left\{ {{T_n}} \right\}}}} } \right| \le \frac{\varepsilon }{16}\sum\limits_{n,c} {\hat Q_n^{\left( c \right)}\left( {{t_0}} \right)}.
\label{eq_backpressure_deviation3}
\end{equation}

Now choose ${{\tilde D}_2} = \max \left\{ {{{\tilde D}^*},{{\tilde D}_1}} \right\}$, based on (\ref{eq_compare_convergence_limit}), (\ref{eq_backpressure_deviation2}) and (\ref{eq_backpressure_deviation3}), we can get, $\forall t\geq \tilde D_2$,
\begin{align}
\sum\limits_n {\mathbb{E}\left\{ {\left. {\left. {{{\tilde Z}_n}\left( {{\bf{\hat Q}}\left( {{t_0}} \right)} \right)} \right|_{{t_0}}^{{t_0} + t - 1}} \right|{\bf{\hat Q}}\left( {{t_0}} \right)} \right\}}  &\ge \sum\limits_n {\frac{{\mathop {\max }\limits_c \left\{ {z_n^{\left( c \right)}\left( {{\bf{\hat Q}}\left( {{t_0}} \right)} \right)} \right\}}}{{\mathbb{E}\left\{ {{T_n}} \right\}}}}  - \frac{\varepsilon }{16}\sum\limits_{n,c} {\hat Q_n^{\left( c \right)}\left( {{t_0}} \right)}\nonumber\\
&\ge \sum\limits_n {\sum\limits_c {\alpha _n^{*\left( c \right)}\frac{{z_n^{\left( c \right)}\left( {{\bf{\hat Q}}\left( {{t_0}} \right)} \right)}}{{\mathbb{E}\left\{ {{T_n}} \right\}}} - \frac{\varepsilon }{16}\sum\limits_{n,c} {\hat Q_n^{\left( c \right)}\left( {{t_0}} \right)} } } \nonumber\\
&= \sum\limits_n {\mathbb{E}\left\{ {\left. {\left. {\tilde Z_n^*\left( {{\bf{\hat Q}}\left( {{t_0}} \right)} \right)} \right|_{{t_0}}^{{t_0} + t - 1}} \right|{\bf{\hat Q}}\left( {{t_0}} \right)} \right\}}  - \frac{\varepsilon }{8}\sum\limits_{n,c} {\hat Q_n^{\left( c \right)}\left( {{t_0}} \right)}.
\label{eq_compare_inter_vs_tilde}
\end{align}

Summarizing the results in (\ref{eq_compare_inter_vs_tilde}) and (\ref{eq_compare_inter_vs_stationary}), we finished the comparison of the key metrics  under $\tilde{Policy}$ and $Policy^*$ as follows: for $\forall t \ge \tilde D^*$,
\begin{equation}
\sum\limits_n {\mathbb{E}\left\{ {\left. {\left. {{{\tilde Z}_n}\left( {{\bf{\hat Q}}\left( {{t_0}} \right)} \right)} \right|_{{t_0}}^{{t_0} + t - 1}} \right|{\bf{\hat Q}}\left( {{t_0}} \right)} \right\}}  \ge \sum\limits_n {\mathbb{E}\left\{ {\left. {\left. {{Z'}_n^*\left( {{\bf{\hat Q}}\left( {{t_0}} \right)} \right)} \right|_{{t_0}}^{{t_0} + t - 1}} \right|{\bf{\hat Q}}\left( {{t_0}} \right)} \right\}}  - \frac{\varepsilon }{8}\sum\limits_{n,c} {\hat Q_n^{\left( c \right)}\left( {{t_0}} \right)}.
\label{eq_compare_stationary_vs_tilde}
\end{equation}

\subsection{Comparison on the key backpressure metric between $\hat{Policy}$ and $\tilde{Policy}$}
\label{subsec_througput_optimal_sub3}
The goal of this part of proof is to compare values of the key metric $\sum\limits_n {\mathbb{E}\left\{ {\left. {\left. {{ Z_n}\left( {{\bf{\hat Q}}\left( {{t_0}} \right)} \right)} \right|_{{t_0}}^{{t_0} + t - 1}} \right|{\bf{\hat Q}}\left( {{t_0}} \right)} \right\}}$ under $\hat {Policy}$ and $\tilde{Policy}$.
If implementing $\hat {Policy}$, the arbitrary timeslot $t_0$ may not be the starting timeslot of one epoch, which means that $\hat{Policy}$ may not start transmitting a new packet from timeslot $t_0$. In contrast, $\tilde {Policy}$ starts transmitting a packet from $t_0$ without using the pre-accumulated partial information, i.e., timeslot $t_0$ is the starting timeslot of a new epoch. Thus, as is shown in Fig \ref{fig_epoch relations among several policies} in Appendix \ref{appendix: policy_list}, within the interval $[t_0,t_0+t-1]$, $\hat{Policy}$ and $\tilde{Policy}$ may not have the synchronized epochs, which makes the direct comparison between them difficult.

In order to deal with the non-synchronized epoches, this part of the proof introduces another intermediate policy: $\tilde{\tilde{Policy}}$, \emph{the intermediate and non-causal policy with the following properties: the epochs for any node have contiguous timeslots; for each node, it is the same as DIVBAR-RMIA in the interval from timeslot $0$ to $u_{n,1}-1$ ($u_{n,1}$ is the starting timeslot of the epoch for the transmitting node $n$ that includes timeslot $t_0$); starting from timeslot $u_{n,1}$ without using the pre-accumulated partial information, each node chooses the same commodity to transmit as that chosen by $\tilde{Policy}$ in timeslot $t_0$ ($u_{n,1}\le t_0$), and keeps transmitting the packets of the chosen commodity during the later timeslots with RMIA, and forwards each decoded packet to the receiver with the largest differential backlog formed under DIVBAR-RMIA in timeslot $t_0$.} Here note that $\tilde{c}$ is actually decided based on the value ${{\bf{\hat Q}}\left( {{t_0}} \right)}$, which might be the backlog state in a later timeslot with respect to timeslot $u_{n,1}$ and results from $\hat{Policy}$. With this non-causality, $\tilde{\tilde{Policy}}$ is a non-realizable policy but is used to facilitate the theoretical analysis.

On the one hand, based on the above definition, we can guarantee that $\tilde{\tilde{Policy}}$ has synchronized epochs with $\hat{Policy}$; on the other hand, each node under $\tilde{\tilde{Policy}}$ transmits the packets of the same commodity as $\tilde{Policy}$ from timeslot $u_{n,i}$, $\tilde{\tilde{Policy}}$ serves as a "bridge" connecting $\hat{Policy}$ and $\tilde{Policy}$, and the proof logic naturally becomes the following two steps: first to compare the metric values under $\hat{Policy}$ and $\tilde{\tilde{Policy}}$; second to compare the metric values under $\tilde{\tilde{Policy}}$ and $\tilde{Policy}$.

\subsubsection{\textbf{Comparison between $\hat{Policy}$ and $\tilde{\tilde{Policy}}$}}
To begin with, define $M_n\left(t_0,t\right)$ as the minimum number of epochs for node $n$ that cover the time interval $\left[t_0,t_0+t-1\right]$ under $\hat{Policy}$ or $\tilde{\tilde{Policy}}$, i.e.
\begin{equation}
{M_n}\left( {{t_0},t} \right) = \min \left\{ {m:{u_{n,1}} + \sum\limits_{i = 1}^m {{T_n}\left( i \right)}  - 1 \ge {t_0} + t - 1} \right\},
\label{eq_def_of_M}
\end{equation}
where $T_n\left(i\right)$ is the number of timeslots in the $i$th epoch of an arbitrary commodity for node $n$. The covering of the interval $\left[t_0,t_0+t-1\right]$ is demonstrated by Fig. \ref{fig_epoch relations among several policies}. Additionally, because $\left\{{\tilde {\tilde Z}_n^{\left( c \right)}\left( {i,{\bf{\hat Q}}\left( {{t_0}} \right)} \right)}, i\ge 1\right\}$ are i.i.d., we can use the following notation under $\tilde{\tilde{Policy}}$:
\begin{equation}
\mathbb{E}\left\{ {\left. {\tilde {\tilde Z}_n^{\left( c \right)}\left( {i,{\bf{\hat Q}}\left( {{t_0}} \right)} \right)} \right|{\bf{\hat Q}}\left( {{t_0}} \right)} \right\}\buildrel \Delta \over =\tilde {\tilde z}_n^{\left( c \right)}\left( {{\bf{\hat Q}}\left( {{t_0}} \right)} \right) .
\end{equation}

With the definitions of $M_n\left(t_0,t\right)$ and $\tilde {\tilde z}_n^{\left( c \right)}\left( {{\bf{\hat Q}}\left( {{t_0}} \right)} \right)$, we propose the following lemma to compare $\hat{Policy}$ and $\tilde{\tilde{Policy}}$:
\begin{lemma}
\label{lemma: hat_vs_tildetilde}
There exists a positive integer $\hat D_2$ such that, for $\forall t \ge \hat D_2$, $\hat{Policy}$ and $\tilde{\tilde{Policy}}$ satisfy the following relationship given the backlog state ${\bf{\hat Q}}\left(t_0\right)$:
\begin{equation}
\sum\limits_n {\mathbb{E}\left\{ {\left. {\left. {{{\hat Z}_n}\left( {{\bf{\hat Q}}\left( {{t_0}} \right)} \right)} \right|_{{t_0}}^{{t_0} + t - 1}} \right|{\bf{\hat Q}}\left( {{t_0}} \right)} \right\}} \geq \frac{1}{t}\sum\limits_n {\tilde {\tilde z}_n\left( {{\bf{\hat Q}}\left( {{t_0}} \right)} \right)\mathbb{E}\left\{ {{M_n}\left( {{t_0},t} \right)} \right\}}  - \left[ {N{C_2}\left( t \right) + {C_1}\left( t \right) + \frac{\varepsilon }{8}\sum\limits_{n,c} {\hat Q_n^{\left( c \right)}\left( {{t_0}} \right)} } \right],
\label{eq_key_metric_comparison1}
\end{equation}
where $C_1\left(t\right)=Nt\left(N+A_{\max}+1\right)$; $C_2\left(t\right)=t\left(N+A_{\max}+1\right)$.
\end{lemma}
The detailed proof of Lemma \ref{lemma: hat_vs_tildetilde} is shown in Appendix \ref{appendix: hat_vs_tildetilde}. Lemma \ref{lemma: hat_vs_tildetilde} demonstrates the comparison between $\hat{Policy}$ and $\tilde{\tilde{Policy}}$ on the key expectation metrics. The term $\sum\limits_n {\mathbb{E}\left\{ {\left. {\left. {{{\hat Z}_n}\left( {{\bf{\hat Q}}\left( {{t_0}} \right)} \right)} \right|_{{t_0}}^{{t_0} + t - 1}} \right|{\bf{\hat Q}}\left( {{t_0}} \right)} \right\}}$ is the key metric over the interval $\left[t_0,t_0+t-1\right]$ under $\hat{Policy}$, while the term $\sum\limits_n {\tilde {\tilde z}_n\left( {{\bf{\hat Q}}\left( {{t_0}} \right)} \right)\mathbb{E}\left\{ {{M_n}\left( {{t_0},t} \right)} \right\}}$ is the key metric over the shortest interval of multiple epochs that covers $\left[t_0,t_0+t-1\right]$ under $\tilde{\tilde{Policy}}$. Note that the two intervals may have a marginal difference at the two ends, as is shown by Fig \ref{fig_epoch relations among several policies} in Appendix \ref{appendix: policy_list}. The derivation result demonstrates that the effect of the marginal part vanishes as $t$ grows. The constant term $N{C_2}\left( t \right) + {C_1}\left( t \right)$ arises from the switching backlog coefficients during the derivations.

\subsubsection{\textbf{Comparison between $\tilde{\tilde{Policy}}$ and $\tilde{Policy}$}}
The next step is to compare $\tilde{\tilde{Policy}}$ and $\tilde{Policy}$. A challenge in the comparison between the two policies is that they don't have synchronized epochs. However, dealing with this issue becomes easier after focusing the analysis on the new metric $\frac{1}{t}\sum\limits_n {{{\tilde {\tilde z}}_n}\left( {{\bf{\hat Q}}\left( {{t_0}} \right)} \right)\mathbb{E}\left\{ {{M_n}\left( {{t_0},t} \right)} \right\}}$.

Similar to (\ref{eq_def_of_M}), first define $\tilde M_n\left(t_0,t\right)$ as the minimum number of epochs for node $n$ that covers the interval $\left[t_0,t_0+t-1\right]$ under $\tilde{Policy}$. Correspondingly, define $\tilde u_{n,i}$ as the starting timeslot of the $i$th epoch under $\tilde{Policy}$, where $\tilde u_{n,1}=t_0$. Note that ${{{\tilde {\tilde Z}}_n}\left( {i,{\bf{\hat Q}}\left( {{t_0}} \right)} \right)}$ and ${{{\tilde Z}_n}\left( {j,{\bf{\hat Q}}\left( {{t_0}} \right)} \right)}$, where $1\le i \le M_n\left(t_0,t\right)$ and $1\le j \le \tilde M_n\left(t_0,t\right)$, are identically distributed since $\tilde{\tilde{Policy}}$ and $\tilde{Policy}$ choose the same commodity to transmit and use the same forwarding strategy based on the same backlog state ${\bf{\hat{Q}}}\left(t_0\right)$ in these epochs. Therefore, we get
\begin{equation}
\tilde {\tilde z}_n\left( {{\bf{\hat Q}}\left( {{t_0}} \right)} \right)=\tilde z_n\left( {{\bf{\hat Q}}\left( {{t_0}} \right)} \right).
\label{eq_key_metric_comparison_single_epoch}
\end{equation}
On the other hand, under the two policies, the number of epochs for node $n$ covering $\left[t_0,t_0+t-1\right]$ might be different. Specifically, since $u_{n,1}\le t_0 = \tilde{u}_{n,1}$, we can guarantee that, with any possible channel realization (common for both policies), $u_{n,i} \le \tilde{u}_{n,i}$, and correspondingly, we have
\begin{equation}
M_n\left(t_0,t\right) \ge \tilde{M}_n\left(t_0,t\right).
\label{eq_comparison_number_of_epochs}
\end{equation}
Based on the results shown in (\ref{eq_key_metric_comparison_single_epoch}) and (\ref{eq_comparison_number_of_epochs}), we have
\begin{equation}
\tilde {\tilde z}_n\left( {{\bf{\hat Q}}\left( {{t_0}} \right)} \right)\mathbb{E}\left\{ {{M_n}\left( {{t_0},t} \right)} \right\} \ge \tilde z_n\left( {{\bf{\hat Q}}\left( {{t_0}} \right)} \right)\mathbb{E}\left\{ {{{\tilde M}_n}\left( {{t_0},t} \right)} \right\}.
\label{eq_key_metric_comparison2}
\end{equation}

Furthermore, under $\tilde{Policy}$, considering that $M_n\left(t_0,t\right)\le t$ because $\tilde T_n\left(i\right)\ge 1$, define the following indicator function of integer $i = 1,2,\cdots,t$:
\begin{equation}
{\tilde 1_n}\left( i \right) = \left\{ {\begin{array}{*{20}{c}}
{1,\;1 \le i \le {{\tilde M}_n}\left( {{t_0},t} \right) \le t}\\
{0,{\rm{\ \ \ \ }}\;{{\tilde M}_n}\left( {{t_0},t} \right) < i \le t,}
\end{array}} \right.
\label{eq_indicator_function}
\end{equation}
and denote $\tilde T_n\left(i\right)$ as the length of the $i$th epoch of commodity $\tilde c$ under $\tilde{Policy}$. Then, similar to the proof of Lemma \ref{lemma: hat_vs_tildetilde} in Appendix \ref{appendix: hat_vs_tildetilde}, $\tilde Z_n\left(i,{\bf \hat Q}\left(t_0\right)\right)$ and $\tilde 1_n\left(i\right)$ are independent because $\left\{\tilde Z_n\left(i,{\bf \hat Q}\left(t_0\right)\right)\right\}$ are i.i.d. and $\tilde 1_n\left(i\right)$ only depends on $\tilde T_n\left(1\right), \cdots, \tilde T_n\left(i-1\right)$. Therefore it follows that
\begin{align}
\mathbb{E}\left\{ {\left. {\sum\limits_{i = 1}^{{{\tilde M}_n}\left( {{t_0},t} \right)} {{{\tilde Z}_n}\left( {i,{\bf{\hat Q}}\left( {{t_0}} \right)} \right)} } \right|{\bf{\hat Q}}\left( {{t_0}} \right)} \right\} &= \mathbb{E}\left\{ {\left. {\sum\limits_{i = 1}^t {{{\tilde Z}_n}\left( {i,{\bf{\hat Q}}\left( {{t_0}} \right)} \right){{\tilde 1}_n}\left( i \right)} } \right|{\bf{\hat Q}}\left( {{t_0}} \right)} \right\}\nonumber\\
&=\sum\limits_{i = 1}^t {\mathbb{E}\left\{ {\left. {{{\tilde Z}_n}\left( {i,{\bf{\hat Q}}\left( {{t_0}} \right)} \right)} \right|{\bf{\hat Q}}\left( {{t_0}} \right)} \right\}} \mathbb{E}\left\{ {{{\tilde 1}_n}\left( i \right)} \right\}\nonumber\\
&= \tilde z_n\left( {{\bf{\hat Q}}\left( {{t_0}} \right)} \right)\sum\limits_{i = 1}^t {\Pr \left\{ {{{\tilde M}_n}\left( {{t_0},t} \right) \ge i} \right\}} \nonumber\\
&=\tilde z_n\left( {{\bf{\hat Q}}\left( {{t_0}} \right)} \right)\mathbb{E}\left\{ {{{\tilde M}_n}\left( {{t_0},t} \right)} \right\}.
\label{eq_key_metric_comparison3}
\end{align}
Plugging (\ref{eq_key_metric_comparison3}) into (\ref{eq_key_metric_comparison2}) yields
\begin{equation}
\tilde {\tilde z}_n\left( {{\bf{\hat Q}}\left( {{t_0}} \right)} \right)\mathbb{E}\left\{ {{M_n}\left( {{t_0},t} \right)} \right\} \ge \mathbb{E}\left\{ {\left. {\sum\limits_{i = 1}^{{{\tilde M}_n}\left( {{t_0},t} \right)} {{{\tilde Z}_n}\left( {i,{\bf{\hat Q}}\left( {{t_0}} \right)} \right)} } \right|{\bf{\hat Q}}\left( {{t_0}} \right)} \right\},
\label{eq_key_metric_comparison3.5}
\end{equation}
which completes the comparison between $\tilde{\tilde{Policy}}$ and $\tilde{Policy}$.

\subsubsection{\textbf{Comparison between $\hat {Policy}$ and $\tilde {Policy}$}}
Plug (\ref{eq_key_metric_comparison3.5}) back into (\ref{eq_key_metric_comparison1}) in Lemma \ref{lemma: hat_vs_tildetilde}, which forms the following inequality: $\forall t\ge \hat D_2$,
\begin{align}
&\ \ \ \sum\limits_n {\mathbb{E}\left\{ {\left. {\left. {{{\hat Z}_n}\left( {{\bf{\hat Q}}\left( {{t_0}} \right)} \right)} \right|_{{t_0}}^{{t_0} + t - 1}} \right|{\bf{\hat Q}}\left( {{t_0}} \right)} \right\}} \nonumber\\
&\geq \sum\limits_n {\mathbb{E}\left\{ {\left. {\frac{1}{t}\sum\limits_{i = 1}^{{{\tilde M}_n}\left( {{t_0},t} \right)} {{{\tilde Z}_n}\left( {i,{\bf{\hat Q}}\left( {{t_0}} \right)} \right)} } \right|{\bf{\hat Q}}\left( {{t_0}} \right)} \right\}}-\left[ {N{C_2}\left( t \right) + {C_1}\left( t \right) + \frac{\varepsilon }{8}\sum\limits_{n,c} {\hat Q_n^{\left( c \right)}\left( {{t_0}} \right)} } \right].
\label{eq_key_metric_comparison4}
\end{align}

Moreover, under $\tilde{Policy}$, we have
\begin{align}
&\ \ \ \ \sum\limits_n {\mathbb{E}\left\{ {\left. {\frac{1}{t}\sum\limits_{i = 1}^{{{\tilde M}_n}\left( {{t_0},t} \right)} {{{\tilde Z}_n}\left( {i,{\bf{\hat Q}}\left( {{t_0}} \right)} \right)} } \right|{\bf{\hat Q}}\left( {{t_0}} \right)} \right\}}\ \ \ \ \ \ \ \ \ \ \ \ \ \ \ \ \ \ \ \ \ \ \ \ \ \ \ \ \ \ \ \ \ \nonumber
\end{align}
\begin{align}
&= \sum\limits_n {\mathbb{E}\left\{ {\frac{1}{t}\sum\limits_{\tau  = {t_0}}^{{t_0} + t - 1} {\sum\limits_c {\sum\limits_{k \in {K_n}} {\tilde b_{nk}^{\left( c \right)}\left( \tau  \right)\left[ {\hat Q_n^{\left( c \right)}\left( {{t_0}} \right) - \hat Q_k^{\left( c \right)}\left( {{t_0}} \right)} \right]} } } } \right.} \nonumber\\
&\ \ \ \ \ \ \left. {\left. { + \frac{1}{t}\sum\limits_{\tau  = {t_0} + t}^{{{\tilde u}_{n,{{\tilde M}_n}\left( {{t_0},t} \right) + 1}} - 1} {\sum\limits_c {\sum\limits_{k \in {K_n}} {\tilde b_{nk}^{\left( c \right)}\left( \tau  \right)\left[ {\hat Q_n^{\left( c \right)}\left( {{t_0}} \right) - \hat Q_k^{\left( c \right)}\left( {{t_0}} \right)} \right]} } } } \right|{\bf{\hat Q}}\left( {{t_0}} \right)} \right\}\nonumber\\
&\ge \sum\limits_n {\mathbb{E}\left\{ {\left. {\left. {{{\tilde Z}_n}\left( {{\bf{\hat Q}}\left( {{t_0}} \right)} \right)} \right|_{{t_0}}^{{t_0} + t - 1}} \right|{\bf{\hat Q}}\left( {{t_0}} \right)} \right\}}.
\label{eq_key_metric_comparison5}
\end{align}
Plug the result of (\ref{eq_key_metric_comparison5}) into (\ref{eq_key_metric_comparison4}) to finally obtain the comparison between the key metric values under $\hat{Policy}$ and $\tilde{Policy}$ over the timeslots interval $\left[t_0, t_0+t-1\right]$, and form the following inequality: $\forall t \ge \hat D_2$,
\begin{align}
&\ \ \ \sum\limits_n {\mathbb{E}\left\{ {\left. {\left. {{{\hat Z}_n}\left( {{\bf{\hat Q}}\left( {{t_0}} \right)} \right)} \right|_{{t_0}}^{{t_0} + t - 1}} \right|{\bf{\hat Q}}\left( {{t_0}} \right)} \right\}} \nonumber\\
&\ge\sum\limits_n {\mathbb{E}\left\{ {\left. {\left. {{{\tilde Z}_n}\left( {{\bf{\hat Q}}\left( {{t_0}} \right)} \right)} \right|_{{t_0}}^{{t_0} + t - 1}} \right|{\bf{\hat Q}}\left( {{t_0}} \right)} \right\}} - \left[ {N{C_2}\left( t \right) + {C_1}\left( t \right) + \frac{\varepsilon }{8}\sum\limits_{n,c} {\hat Q_n^{\left( c \right)}\left( {{t_0}} \right)} } \right].
\label{eq_key_metric_comparison6}
\end{align}

\subsection{Strong stability achieved under $\hat{Policy}$}
\label{subsec_througput_optimal_sub4}
Combining the result of the comparison on the key metrics between $\tilde{Policy}$ and ${Policy'}^*$ shown as (\ref{eq_compare_stationary_vs_tilde}) in Subsection \ref{subsec_througput_optimal_sub2} and the result of the comparison on the key metrics between $\hat{Policy}$ and $\tilde{Policy}$ shown as (\ref{eq_key_metric_comparison6}) in Subsection \ref{subsec_througput_optimal_sub3}, the comparison on the key backpressure metric between $\hat{Policy}$ and ${Policy'}^*$ is shown as follows: $\forall t \ge \max\left\{\hat D_2, \tilde D_2\right\}$,
\begin{align}
&\ \ \ \sum\limits_n {\mathbb{E}\left\{ {\left. {\left. {{{\hat Z}_n}\left( {{\bf{\hat Q}}\left( {{t_0}} \right)} \right)} \right|_{{t_0}}^{{t_0} + t - 1}} \right|{\bf{\hat Q}}\left( {{t_0}} \right)} \right\}}\nonumber\\
&\ge \sum\limits_n {\mathbb{E}\left\{ {\left. {\left. {{Z'}_n^*\left( {{\bf{\hat Q}}\left( {{t_0}} \right)} \right)} \right|_{{t_0}}^{{t_0} + t - 1}} \right|{\bf{\hat Q}}\left( {{t_0}} \right)} \right\}}  - \left[N{C_2}\left( t \right) + {C_1}\left( t \right) + \frac{\varepsilon }{4}\sum\limits_{n,c} {\hat Q_n^{\left( c \right)}\left( {{t_0}}, \right)}\right].
\label{eq_key_metric_comparison7}
\end{align}

Going back to (\ref{eq_DIVBAR_RMIA_Lyapunov_drift1}) in Subsection \ref{subsec_througput_optimal_sub1}, after plugging (\ref{eq_key_metric_comparison7}) into (\ref{eq_DIVBAR_RMIA_Lyapunov_drift1}), the $t$-slot average Lyapunov drift of $\hat{Policy}$ at arbitrary timeslot $t_0$ can be further upper bounded as follows:
\begin{align}
&\ \ \ \ \frac{1}{t}\sum\limits_{n,c} \mathbb{E} \left\{ {\left. {{{\left( {\hat Q_n^{\left( c \right)}\left( {{t_0} + t} \right)} \right)}^2} - {{\left( {\hat Q_n^{\left( c \right)}\left( {{t_0}} \right)} \right)}^2}} \right|{\bf{\hat Q}}\left( {{t_0}} \right)} \right\}\nonumber\\
&\le B\left(t\right) + 2\left[ {{C_1}\left( t \right) + N{C_2}\left( t \right)} \right] + \frac{\varepsilon}{2} \sum\limits_{n,c} {\hat Q_n^{\left( c \right)}\left( {{t_0}} \right)}\nonumber\\
&\ \ \ - 2\left[ {\sum\limits_n {\mathbb{E}\left\{ {\left. {\left. {{Z'}_n^*\left( {{\bf{\hat Q}}\left( {{t_0}} \right)} \right)} \right|_{{t_0}}^{{t_0} + t - 1}} \right|{\bf{\hat Q}}\left( {{t_0}} \right)} \right\}}  - \frac{1}{t}\sum\limits_{n,c} {\hat Q_n^{\left( c \right)}\left( {{t_0}} \right)\sum\limits_{\tau  = {t_0}}^{{t_0} + t - 1} {\mathbb{E}\left\{ {a_n^{\left( c \right)}\left( \tau  \right)} \right\}} } } \right]\nonumber\\
&\buildrel\Delta\over = B\left(t\right) + 2\left[ {{C_1}\left( t \right) + N{C_2}\left( t \right)} \right] + \frac{\varepsilon}{2} \sum\limits_{n,c} {\hat Q_n^{\left( c \right)}\left( {{t_0}} \right)} -2\Upsilon \left( {{\bf{\hat Q}}\left( {{t_0}} \right)} \right),
\label{eq_lyapunov_drift_bound1}
\end{align}
where $\Upsilon \left( {{\bf{\hat Q}}\left( {{t_0}} \right)} \right)$ in (\ref{eq_lyapunov_drift_bound1}) is as follows:
\begin{align}
\Upsilon\left( {{\bf{\hat Q}}\left( {{t_0}} \right)} \right)&= \sum\limits_n {\mathbb{E}\left\{ {\left. {\left. {{Z'}_n^*\left( {{\bf{\hat Q}}\left( {{t_0}} \right)} \right)} \right|_{{t_0}}^{{t_0} + t - 1}} \right|{\bf{\hat Q}}\left( {{t_0}} \right)} \right\}}  - \frac{1}{t}\sum\limits_{n,c} {\hat Q_n^{\left( c \right)}\left( {{t_0}} \right)\sum\limits_{\tau  = {t_0}}^{{t_0} + t - 1} {\mathbb{E}\left\{ {a_n^{\left( c \right)}\left( \tau  \right)} \right\}} }\nonumber\\
&=\sum\limits_{n,c} {\hat Q_n^{\left( c \right)}\left( {{t_0}} \right)\frac{1}{t}\sum\limits_{\tau  = {t_0}}^{{t_0} + t - 1} {\mathbb{E}\left\{ {\sum\limits_{k \in {{\cal K}_n}} {{b'}_{nk}^{*\left( c \right)}\left( \tau  \right)}  - \sum\limits_{k \in {{\cal K}_n}} {{b'}_{kn}^{*\left( c \right)}\left( \tau  \right) - } a_n^{\left( c \right)}\left( \tau  \right)} \right\}} },
\label{eq_key_metric_star}
\end{align}
where the given backlog state observation ${{\bf{\hat Q}}\left( {{t_0}} \right)}$ in the expectation is dropped because, from timeslot $t_0$, ${Policy'}^*$ is the same as the stationary randomized policy $Policy^*$ starting from timeslot $0$ and makes all decisions independent of the backlog observations. According to (\ref{eq_bound_average_rate_term}) in appendix \ref{appendix: proof_sufficiency_corollary1}, there exists a positive integer $D^*$ such that, for $\forall t\ge D^*$, we have
\begin{equation}
\frac{1}{t}\sum\limits_{\tau  = {t_0}}^{{t_0} + t - 1} {\mathbb{E}\left\{ {\sum\limits_{k \in {{\cal K}_n}} {{b'}_{nk}^{*\left( c \right)}\left( \tau  \right)}  - \sum\limits_{k \in {{\cal K}_n}} {{b'}_{kn}^{*\left( c \right)}\left( \tau  \right) - } a_n^{\left( c \right)}\left( \tau  \right)} \right\}}  = \frac{1}{t}\sum\limits_{\tau  = 0}^{t - 1} {\mathbb{E}\left\{ {\sum\limits_{k \in {{\cal K}_n}} {b_{nk}^{*\left( c \right)}\left( \tau  \right)}  - \sum\limits_{k \in {{\cal K}_n}} {b_{kn}^{*\left( c \right)}\left( \tau  \right) - } a_n^{\left( c \right)}\left( \tau  \right)} \right\}}  \ge \frac{\varepsilon }{2}.
\label{eq_convergence_star2}
\end{equation}

Plug (\ref{eq_convergence_star2}) into right hand side of (\ref{eq_key_metric_star}), and it follows that, $\forall t\ge D^*$,
\begin{equation}
\Upsilon \left( {{\bf{\hat Q}}\left( {{t_0}} \right)} \right) \ge \frac{\varepsilon}{2} \sum\limits_{n,c} {\hat Q_n^{\left( c \right)}\left( {{t_0}} \right)}.
\label{eq_key_metric_star2}
\end{equation}
Then further plug (\ref{eq_key_metric_star2}) back into (\ref{eq_lyapunov_drift_bound1}), and let $t\ge {\rm{max}}\left\{\hat D_2, \tilde D_2, D^*\right\}$ to get
\begin{align}
\frac{1}{t}\sum\limits_{n,c} \mathbb{E} \left\{ {\left. {{{\left( {\hat Q_n^{\left( c \right)}\left( {{t_0} + t} \right)} \right)}^2} - {{\left( {\hat Q_n^{\left( c \right)}\left( {{t_0}} \right)} \right)}^2}} \right|{\bf{\hat Q}}\left( {{t_0}} \right)} \right\} & \le B\left(t\right) + C\left(t\right) - \frac{\varepsilon}{2} \sum\limits_{n,c} {\hat Q_n^{\left( c \right)}\left( {{t_0}} \right)}.
\label{eq_lyapunov_drift_bound2}
\end{align}
where $C\left(t\right)=2\left[C_1\left(t\right)+NC_2\left(t\right)\right]=4Nt\left(N+A_{\rm{max}}+1\right)$. Now we let $t=D = {\rm{max}}\left\{\hat D_2, \tilde D_2, D^*\right\}$, and (\ref{eq_lyapunov_drift_bound2}) becomes
\begin{equation}
\frac{1}{D}\sum\limits_{n,c} \mathbb{E} \left\{ {\left. {{{\left( {\hat Q_n^{\left( c \right)}\left( {{t_0} + D} \right)} \right)}^2} - {{\left( {\hat Q_n^{\left( c \right)}\left( {{t_0}} \right)} \right)}^2}} \right|{\bf{\hat Q}}\left( {{t_0}} \right)} \right\}\le B\left(D\right) + C\left(D\right) - \frac{\varepsilon}{2} \sum\limits_{n,c} {\hat Q_n^{\left( c \right)}\left( {{t_0}} \right)},
\label{eq_lyapunov_drift_bound3_0}
\end{equation}
where $B\left(D\right)=N^2D\left[1+\left(N+A_{\rm{max}}\right)^2\right]$; $C\left(D\right)=4ND\left(N+A_{\rm{max}}+1\right)$.

Now a proper upper bound for the $D$-timeslot average Lyapunov drift of $\hat{Policy}$ is formed by (\ref{eq_lyapunov_drift_bound3_0}), which satisfies the condition of Lemma \ref{lemma: strong_stability} in Appendix \ref{appendix: proof_sufficiency_corollary1}. Then it follows that, with the arbitrary input rate matrix $\left(\lambda_n^{\left(c\right)}\right)$ within $\Lambda_{\rm{RMIA}}$ such that $\exists \varepsilon >0$, $\left(\lambda_n^{\left(c\right)}+\varepsilon\right)\in \Lambda_{\rm{RMIA}}$, we get
\begin{equation}
\mathop {\lim \sup }\limits_{t \to \infty } \frac{1}{t}\sum\limits_{\tau  = 0}^{t - 1} {\sum\limits_{n,c} {\mathbb{E}\left\{ {\hat Q_n^{\left( c \right)}\left( \tau  \right)} \right\}} }  \le \frac{2\left[B\left(D\right)+C\left(D\right)\right]}{\varepsilon },
\label{eq_strong_stability_hat}
\end{equation}
which demonstrates the strong stability in the network. Thus, DIVBAR-RMIA is throughput optimal among all the policies with RMIA. Here the values of $B\left(D\right)$ and $C\left(D\right)$ are linear functions of $D$, which increases as $\varepsilon$ decreases, and therefore, the smaller $\varepsilon$ is, the larger the upper bound of the mean time average total backlog is.

\section{Proof of Theorem \ref{thm: DIVBAR-RMIA_vs_DIVBAR-MIA}}
\label{appendix: DIVBAR-MIA_vs_DIVBAR-RMIA}
Similar to the proof of Theorem \ref{thm: DIVBAR_MIA_throughput_optimal}, the ultimate goal of the proof in this theorem is to show the strong stability under DIVBAR-MIA with any input rate matrix $\left(\lambda_n^{\left(c\right)}\right)$ within $\Lambda_{\rm{RMIA}}$.

Define $\hat{\hat{Policy}}$ as the DIVBAR-RMIA policy and any variable uniquely specified by $\hat{\hat{Policy}}$ is denoted in the form $\hat {\hat x}$. The proof strategy in this theorem is the same as that of Theorem \ref{thm: DIVBAR_MIA_throughput_optimal}. We aim to compare the upper bounds of the $d$-slot average Lyapunov drift under $\hat{\hat{Policy}}$ and ${Policy'}^*$, where ${Policy'}^*$ is the modified version of the stationary randomized policy $Policy^*$: \emph{being the same as $\hat{\hat{Policy}}$ in the interval from timeslot $0$ to timeslot $t_0-1$ and, from timeslot $t_0$ without using the pre-accumulated partial information, is the same as $Policy^*$ starting from timeslot $0$}. Note that the only difference between the ${Policy'}^*$ introduced here and the one introduced in Theorem \ref{thm: DIVBAR_MIA_throughput_optimal} lies in the interval $\left[0,t_0-1\right]$, where the one introduced here is the same as DIVBAR-MIA, while the one introduced in Theorem \ref{thm: DIVBAR_MIA_throughput_optimal} is the same as DIVBAR-RMIA.

In the proof, we also introduce the intermediate policy $\tilde{Policy}$ for the arbitrary timeslot $t_0$, whose definition is the same as the one introduced in Theorem \ref{thm: DIVBAR_MIA_throughput_optimal} except that the $\tilde{Policy}$ here is the same as DIVBAR-MIA instead of DIVVAR-RMIA in the interval $\left[0,t_0-1\right]$, and the decisions made during the time starting from timeslot $t_0$ under $\tilde{Policy}$ here is based on the backlog observations ${\bf{\hat{\hat{Q}}}}\left(t_0\right)$.

The whole proof of this theorem is also divided into four steps:
\begin{enumerate}[1.]
\item Transform the comparison on the upper bounds of Lyapunov drift under $\hat{\hat {Policy}}$ and ${Policy'}^*$ to the comparison on key backpressure metrics under the two policies, which is shown in Subsection \ref{subsec_DIVBAR-RMIA_vs_DIVBAR-MIA_sub1}.
\item Compare the key metrics under $\hat{\hat{Policy}}$ and $\tilde{Policy}$, which is shown in Subsection \ref{subsec_DIVBAR-RMIA_vs_DIVBAR-MIA_sub2}.
\item Compare the key metrics under $\tilde{Policy}$ and ${Policy'}^*$, which is shown in Subsection \ref{subsec_DIVBAR-RMIA_vs_DIVBAR-MIA_sub3}.
\item Combine the results in 2 and 3 and get the conclusion of the strong stability, which is shown in Subsection \ref{subsec_DIVBAR-RMIA_vs_DIVBAR-MIA_sub4}.
\end{enumerate}

\subsection{Transforming the comparison on the upper bounds of Lyapunov drift to the comparison on the key backpressure metric}
\label{subsec_DIVBAR-RMIA_vs_DIVBAR-MIA_sub1}
Similar to the proof of Theorem \ref{thm: DIVBAR_MIA_throughput_optimal} in Appendix \ref{subsec_througput_optimal_sub1}, we start from the queueing dynamics and derive the upper bound of the $t$-slot average Lyapunov drift under $\hat{\hat{Policy}}$ as follows:
\begin{align}
&\ \ \ \ \frac{1}{t}\sum\limits_{n,c} \mathbb{E} \left\{ {\left. {{{\left( {\hat {\hat Q}_n^{\left( c \right)}\left( {{t_0} + t} \right)} \right)}^2} - {{\left( {\hat {\hat Q}_n^{\left( c \right)}\left( {{t_0}} \right)} \right)}^2}} \right|{\bf{\hat {\hat Q}}}\left( {{t_0}} \right)} \right\}\nonumber\\
&\le B\left( t \right) + \frac{2}{t}\sum\limits_{n,c} {\hat {\hat Q}_n^{\left( c \right)}\left( {{t_0}} \right)\sum\limits_{\tau  = {t_0}}^{{t_0} + t - 1} {\mathbb{E}\left\{ {a_n^{\left( c \right)}\left( \tau  \right)} \right\}} } - 2\sum\limits_n {\mathbb{E}\left\{ {\left. {\left. {{{\hat {\hat Z}}_n}\left( {{\bf{\hat {\hat Q}}}\left( {{t_0}} \right)} \right)} \right|_{{t_0}}^{{t_0} + t - 1}} \right|{\bf{\hat {\hat Q}}}\left( {{t_0}} \right)} \right\}},
\label{eq_upper_bound_lyapunov_drift_hathat}
\end{align}
where $B\left(t\right)=N^2t\left[1+\left(N+A_{\rm{max}}\right)^2\right]$; under an arbitrary policy, the metric $\left. {{{ { Z}}_n}\left( {{\bf{\hat {\hat Q}}}\left( {{t_0}} \right)} \right)} \right|_{{t_0}}^{{t_0} + t - 1}$ is as follows:
\begin{equation}
\left. {{{ { Z}}_n}\left( {{\bf{\hat {\hat Q}}}\left( {{t_0}} \right)} \right)} \right|_{{t_0}}^{{t_0} + t - 1} = \frac{1}{t}\sum\limits_{\tau  = {t_0}}^{{t_0} + t - 1} {\sum\limits_c {\sum\limits_{k \in {{\cal K}_n}} { { b}_{nk}^{\left( c \right)}\left( \tau  \right)\left[ {\hat {\hat Q}_n^{\left( c \right)}\left( {{t_0}} \right) - \hat {\hat Q}_k^{\left( c \right)}\left( {{t_0}} \right)} \right]} } }.
\end{equation}

Then the proof reduces to the comparison of the key metric $\mathbb{E}\left\{ {\left. {\left. {{Z_n}\left( {{\bf{\hat {\hat Q}}}\left( {{t_0}} \right)} \right)} \right|_{{t_0}}^{{t_0} + t - 1}} \right|{\bf{\hat {\hat Q}}}\left( {{t_0}} \right)} \right\}$ under $\hat{\hat{Policy}}$ and ${Policy'}^*$. In order to facilitate the comparison on the key metrics, we introduce the intermediate policy $\tilde{Policy}$. Then the later proof naturally consists of two parts: first to compare $\tilde{Policy}$ and $Policy^*$; second to compare $\hat{\hat{Policy}}$ and $\tilde{Policy}$, which are shown as the following two subsections: Subsection \ref{subsec_DIVBAR-RMIA_vs_DIVBAR-MIA_sub2} and Subsection \ref{subsec_DIVBAR-RMIA_vs_DIVBAR-MIA_sub3}.

\subsection{Comparison on the key backpressure metric between $\tilde{Policy}$ and ${Policy'}^*$}
\label{subsec_DIVBAR-RMIA_vs_DIVBAR-MIA_sub2}
The comparison on the key metric $\mathbb{E}\left\{ {\left. {\left. {{Z_n}\left( {{\bf{\hat {\hat Q}}}\left( {{t_0}} \right)} \right)} \right|_{{t_0}}^{{t_0} + t - 1}} \right|{\bf{\hat {\hat Q}}}\left( {{t_0}} \right)} \right\}$ is the same as that in the proof of Theorem \ref{thm: DIVBAR_MIA_throughput_optimal} shown in Appendix \ref{subsec_througput_optimal_sub2}, except that the backlog coefficient here is ${{\bf{\hat {\hat Q}}}\left( {{t_0}} \right)}$. The final comparison results should be, for arbitrary starting timeslot $t_0$, there exists an integer $\tilde{D}_2>0$, such that, for $\forall t\ge \tilde{D}_2$,
\begin{equation}
\sum\limits_n {\mathbb{E}\left\{ {\left. {\left. {{{\tilde Z}_n}\left( {{\bf{\hat {\hat Q}}}\left( {{t_0}} \right)} \right)} \right|_{{t_0}}^{{t_0} + t - 1}} \right|{\bf{\hat {\hat Q}}}\left( {{t_0}} \right)} \right\}}  \ge \sum\limits_n {\mathbb{E}\left\{ {\left. {\left. {{Z'}_n^*\left( {{\bf{\hat {\hat Q}}}\left( {{t_0}} \right)} \right)} \right|_{{t_0}}^{{t_0} + t - 1}} \right|{\bf{\hat {\hat Q}}}\left( {{t_0}} \right)} \right\}}  - \frac{\varepsilon }{8}\sum\limits_{n,c} {\hat {\hat Q}_n^{\left( c \right)}\left( {{t_0}} \right)}.
\label{eq_comparison_result_tilde_vs_star}
\end{equation}

\subsection{Comparison on the key backpressure metric between $\hat{\hat{Policy}}$ and $\tilde{Policy}$}
\label{subsec_DIVBAR-RMIA_vs_DIVBAR-MIA_sub3}
A direct comparison between $\hat{\hat{Policy}}$ and $\tilde{Policy}$ is difficult to analyze for two reasons: firstly, the policies have asynchronized epochs during the interval $\left[t_0,t_0+t-1\right]$; secondly, the two policies have different transmission schemes, i.e, $\hat{\hat{Policy}}$ does not clear partial information at the end of each epoch, while $\tilde{Policy}$ does. To respectively deal with the two difficulties, we introduce two intermediate policies: $\tilde{\tilde{Policy}}$ and $\hat{Policy'}$.

The intermediate policy $\tilde{\tilde{Policy}}$ with RMIA is also non-causal and is the same as the one introduced in the proof of Theorem \ref{thm: DIVBAR_MIA_throughput_optimal} except that the $\tilde{\tilde{Policy}}$ here is the same as DIVBAR-MIA instead of DIVBAR-RMIA in the interval $\left[0,u_{n,1}-1\right]$, where $u_{n,1}$ is the starting timeslot of the epoch for node $n$ that includes timeslot $t_0$ under $\hat{\hat{Policy}}$. As a summary, $\tilde{\tilde{Policy}}$ chooses the same commodity $\tilde{c}$ as $\tilde{Policy}$ to transmit, while it has exactly coherent epochs as $\hat{\hat{Policy}}$. The intermediate policy $\hat{Policy'}$ is used to compare the metric values under $\hat{\hat{Policy}}$ and $\tilde{\tilde{Policy}}$ over a single epoch for each node $n$ with starting timeslot $u_{n,i}$. To be specific, \emph{for each node $n$, $\hat{Policy'}$ is the same as DIVBAR-MIA from timeslot $0$ to timeslot $u_{n,i}-1$, while starting from timeslot $u_{n,i}$ without using the pre-accumulated partial information, it is the same as DIVBAR-RMIA.} Note that the only difference between $\hat{Policy'}$ and $\hat{Policy}$ over epoch $i$ is that they make decisions based on different backlog states: ${{\bf{\hat {\hat Q}}}\left( {{u_{n,i}}} \right)}$ and ${{\bf{{\hat Q}}}\left( {{u_{n,i}}} \right)}$, respectively.

With similar logic as Theorem \ref{thm: DIVBAR_MIA_throughput_optimal}, the later proof in this subsection consists of two steps: 1) compare $\hat{\hat{Policy}}$ with $\tilde{\tilde{Policy}}$, in which $\hat{Policy'}$ is used as the intermediate policy in the comparison on the metric a single epoch; 2) compare $\tilde{\tilde{Policy}}$ and $\tilde{Policy}$.

\subsubsection{\textbf{Comparison between $\hat{\hat{Policy}}$ and $\tilde{\tilde{Policy}}$}}
Since $\hat{\hat{Policy}}$ and $\tilde{\tilde{Policy}}$ have coherent epochs with $\hat{Policy}$, the variables $M_n\left(t_0,t\right)$ and  $\left\{u_{n,i}:1\le i \le M_n\left(t_0,t\right)\right\}$ used in the proof of Theorem \ref{thm: DIVBAR_MIA_throughput_optimal} are still valid for implementing $\hat{\hat{Policy}}$  and $\tilde{\tilde{Policy}}$. Similar to the proof of Lemma \ref{lemma: hat_vs_tildetilde} shown in Appendix \ref{subsec_hat_vs_tildetilde_1} except switching the backlog coefficients to ${{\bf{\hat {\hat Q}}}\left( {{u_{n,i}}} \right)}$ and ${{\bf{{\hat {\hat Q}}}}\left( {{t_0}} \right)}$, we get, for $\forall t\ge \max\{\hat {\hat D}_1, T_{\max}\}\buildrel \Delta \over = \hat{\hat D}_2$,
\begin{equation}
\mathbb{E}\left\{ {\left. {\left. {{{\hat {\hat Z}}_n}\left( {{\bf{\hat Q}}\left( {{t_0}} \right)} \right)} \right|_{{t_0}}^{{t_0} + t - 1}} \right|{\bf{\hat {\hat Q}}}\left( {{t_0}} \right)} \right\} \ge \sum\limits_n {\mathbb{E}\left\{ {\left. {\frac{1}{t}\sum\limits_{i = 1}^{{M_n}\left( {{t_0},t} \right)} {{{\hat {\hat Z}}_n}\left( {i,{\bf{\hat {\hat Q}}}\left( {{u_{n,i}}} \right)} \right)} } \right|{\bf{\hat {\hat Q}}}\left( {{t_0}} \right)} \right\}}  - {C_1}\left( t \right) - \frac{\varepsilon }{8}\sum\limits_{n,c} {\hat {\hat Q}_n^{\left( c \right)}\left( {{t_0}} \right)},
\label{eq_key_metric_hathat3}
\end{equation}
where ${C_1}\left( t \right) = Nt\left( {N + {A_{\max }} + 1} \right)$, and
\begin{equation}
{{\hat {\hat Z}}_n}\left( {i,{\bf{\hat {\hat Q}}}\left( {{t_0}} \right)} \right) = \sum\limits_{\tau  = {u_{n,i}}}^{{u_{n,i + 1}} - 1} {\sum\limits_c {\sum\limits_{k \in {{\cal K}_n}} {\hat {\hat b}_{nk}^{\left( c \right)}\left( \tau  \right)\left[ {\hat {\hat Q}_n^{\left( c \right)}\left( {{t_0}} \right) - \hat {\hat Q}_k^{\left( c \right)}\left( {{t_0}} \right)} \right]} } } ,{\rm{\ }}1 \le i \le {M_n}\left( {{t_0},t} \right).
\end{equation}

Furthermore, we compare the key metric values for node $n$ over a single epoch: $\mathbb{E}\left\{ {\left. {{Z_n}\left( {i,{\bf{\hat {\hat Q}}}\left( {{u_{n,i}}} \right)} \right)} \right|{\bf{\hat {\hat Q}}}\left( {{t_0}} \right)} \right\}$ under $\hat{\hat{Policy}}$ and $\tilde{\tilde{Policy}}$.
Considering that $\hat{\hat{Policy}}$ and $\tilde{\tilde{Policy}}$ have synchronized epochs, the indicator function $1_n\left(i\right)$ defined in (\ref{eq_indicator_function_hat}) in Appendix \ref{subsec_hat_vs_tildetilde_2} can still be used, and correspondingly, we have
\begin{equation}
\mathbb{E}\left\{ {\left. {{{\hat {\hat Z}}_n}\left( {i,{\bf{\hat {\hat Q}}}\left( {{u_{n,i}}} \right)} \right)} \right|{\bf{\hat {\hat Q}}}\left( {{t_0}} \right),1 \le i \le {M_n}\left( {{t_0},t} \right)} \right\} = \mathbb{E}\left\{ {\left. {{{\hat {\hat Z}}_n}\left( {i,{\bf{\hat {\hat Q}}}\left( {{u_{n,i}}} \right)} \right)} \right|{\bf{\hat {\hat Q}}}\left( {{t_0}} \right),{1_n}\left( i \right) = 1} \right\}.
\end{equation}

As described before, for a particular epoch $i$ for node $n$, we introduce the intermediate policy $\hat{Policy'}$, which implements the same strategy as DIVBAR-RMIA over epoch $i$ based on the backlog state ${{\bf{\hat {\hat Q}}}\left( {{u_{n,i}}} \right)}$. Under either $\hat{\hat{Policy}}$ or $\hat{Policy'}$, node $n$ observes the backlog state ${{\bf{\hat {\hat Q}}}\left( {{u_{n,i}}} \right)}$ and follows the backpressure strategy to make routing decisions. However, under $\hat{\hat{Policy}}$, each receiver may also have some pre-accumulated partial information by the beginning of timeslot $u_{n,i}$ and can take advantage of it to decode the packet being transmitted, while under $\hat{Policy'}$, each receiver can not use the pre-accumulated partial information. Thus, by the ending timeslot of epoch $i$, the successful receiver set under $\hat{\hat{Policy}}$ should include the first successful receiver set under $\hat{Policy'}$, which leads to the intuition that $\hat{\hat{Policy}}$ should perform at least as well as $\hat{Policy'}$. Following this intuition, we summarize the comparison between $\hat{\hat{Policy}}$ and $\hat{Policy'}$ over a single epoch as the following lemma:

\begin{lemma}
\label{lemma: hathat_vs_hatprime}
The key metrics $\hat{\hat{Policy}}$ and $\hat{Policy'}$ over a single epoch have the following relationship: for $\forall t \ge \hat{\hat{D}}_2$,
\begin{equation}
\mathbb{E}\left\{ {\left. {{{\hat {\hat Z}}_n}\left( {i,{\bf{\hat {\hat Q}}}\left( {{u_{n,i}}} \right)} \right)} \right|{\bf{\hat {\hat Q}}}\left( {{t_0}} \right),{1_n}\left( i \right) = 1} \right\} \ge \mathbb{E}\left\{ {\left. {\hat Z{'_n}\left( {i,{\bf{\hat {\hat Q}}}\left( {{u_{n,i}}} \right)} \right)} \right|{\bf{\hat {\hat Q}}}\left( {{t_0}} \right),{1_n}\left( i \right) = 1} \right\}.
\label{eq_key_metric_comparison_hathat_vs_hatprime}
\end{equation}
\end{lemma}

The detailed proof of Lemma \ref{lemma: hathat_vs_hatprime} is shown in Appendix \ref{appendix: hathat_vs_hatprime}.

The next step is to compare $\tilde{\tilde{Policy}}$ and $\hat{Policy'}$ over a single epoch, and the mathematical manipulations are similar to the proof of Lemma \ref{lemma: hat_vs_tildetilde} shown in Appendix \ref{subsec_hat_vs_tildetilde_2}, except that the backlog coefficients are switched to ${{\bf{\hat {\hat Q}}}\left( {{u_{n,i}}} \right)}$ and ${{\bf{{\hat {\hat Q}}}}\left( {{t_0}} \right)}$. The comparison result is shown as follows: for $\forall t \ge \hat{\hat{D}}_2$,
\begin{equation}
\mathbb{E}\left\{ {\left. {\hat Z{'_n}\left( {i,{\bf{\hat {\hat Q}}}\left( {{u_{n,i}}} \right)} \right)} \right|{\bf{\hat {\hat Q}}}\left( {{t_0}} \right),{1_n}\left( i \right) = 1} \right\} \ge \mathbb{E}\left\{ {\left. {{{\tilde {\tilde Z}}_n}\left( {i,{\bf{\hat {\hat Q}}}\left( {{t_0}} \right)} \right)} \right|{\bf{\hat {\hat Q}}}\left( {{t_0}} \right),{1_n}\left( i \right) = 1} \right\}-C_2\left(t\right).
\label{eq_key_metric_comparison_hatprime_vs_tildetilde3}
\end{equation}
where $C_2\left(t\right)=t\left(N+A_{\max}+1\right)$.

Combining the results shown as (\ref{eq_key_metric_comparison_hathat_vs_hatprime}) and (\ref{eq_key_metric_comparison_hatprime_vs_tildetilde3}), we finally get the comparison result between $\hat{\hat{Policy}}$ and $\tilde{\tilde{Policy}}$ over a single epoch shown as follows: for $\forall t \ge \hat{\hat{D}}_2$,
\begin{equation}
\mathbb{E}\left\{ {\left. {{{\hat {\hat Z}}_n}\left( {i,{\bf{\hat {\hat Q}}}\left( {{u_{n,i}}} \right)} \right)} \right|{\bf{\hat {\hat Q}}}\left( {{t_0}} \right),{1_n}\left( i \right) = 1} \right\} \ge \mathbb{E}\left\{ {\left. {{{\tilde {\tilde Z}}_n}\left( {i,{\bf{\hat {\hat Q}}}\left( {{t_0}} \right)} \right)} \right|{\bf{\hat {\hat Q}}}\left( {{t_0}} \right),{1_n}\left( i \right) = 1} \right\} - {C_2}\left( t \right).
\label{eq_key_metric_comparison_hathat_vs_tildetidle}
\end{equation}

Based on the comparison result shown as (\ref{eq_key_metric_comparison_hathat_vs_tildetidle}) over a single epoch between $\hat{\hat{Policy}}$ and $\tilde{\tilde{Policy}}$, we extend the comparison to $M_n\left(t_0,t\right)$ epochs. Following the proof of Lemma \ref{lemma: hat_vs_tildetilde} shown in Appendix \ref{subsec_hat_vs_tildetilde_3} except changing backlog coefficients to ${\bf{\hat{\hat Q}}}\left(u_{n,i}\right)$ and ${\bf{\hat{\hat Q}}}\left(t_0\right)$, we can get the comparison result as follows: for $\forall t \ge \hat{\hat D}_2$,
\begin{equation}
\sum\limits_n {\mathbb{E}\left\{ {\left. {\left. {{{\hat {\hat Z}}_n}\left( {{\bf{\hat {\hat Q}}}\left( {{t_0}} \right)} \right)} \right|_{{t_0}}^{{t_0} + t - 1}} \right|{\bf{\hat {\hat Q}}}\left( {{t_0}} \right)} \right\}} \geq \frac{1}{t}\sum\limits_n {\tilde {\tilde z}\left( {{\bf{\hat {\hat Q}}}\left( {{t_0}} \right)} \right)\mathbb{E}\left\{ {{M_n}\left( {{t_0},t} \right)} \right\}}  - \left[ {N{C_2}\left( t \right) + {C_1}\left( t \right) + \frac{\varepsilon }{8}\sum\limits_{n,c} {\hat {\hat Q}_n^{\left( c \right)}\left( {{t_0}} \right)} } \right].
\label{eq_key_metric_comparison_hathat_vs_tildetilde6}
\end{equation}

\subsubsection{\textbf{Comparison between $\tilde{\tilde{Policy}}$ and $\tilde{Policy}$}}
The comparison of $\tilde{\tilde{Policy}}$ and $\tilde{Policy}$ is the same as the proof in Theorem \ref{thm: DIVBAR_MIA_throughput_optimal} shown as the step 2) of Appendix \ref{subsec_througput_optimal_sub3}, except that the backlog state observations coefficients are switched to ${{\bf{\hat {\hat Q}}}\left( {{t_0}} \right)}$. Then the comparison result should be, for $\forall n \in \cal N $,
\begin{equation}
{{\tilde {\tilde z}}_n}\left( {{\bf{\hat {\hat Q}}}\left( {{t_0}} \right)} \right)\mathbb{E}\left\{ {{M_n}\left( {{t_0},t} \right)} \right\} \ge {{\tilde z}_n}\left( {{\bf{\hat {\hat Q}}}\left( {{t_0}} \right)} \right)\mathbb{E}\left\{ {{{\tilde M}_n}\left( {{t_0},t} \right)} \right\} = \mathbb{E}\left\{ {\left. {\sum\limits_{i = 1}^{{{\tilde M}_n}\left( {{t_0},t} \right)} {{{\tilde Z}_n}\left( {i,{\bf{\hat {\hat Q}}}\left( {{t_0}} \right)} \right)} } \right|{\bf{\hat {\hat Q}}}\left( {{t_0}} \right)} \right\}.
\label{eq_key_metric_comparison_tildetilde_vs_tilde}
\end{equation}

Combining the results of the above two steps shown as (\ref{eq_key_metric_comparison_tildetilde_vs_tilde}) and (\ref{eq_key_metric_comparison_hathat_vs_tildetilde6}), we get, $\forall t\ge \hat D_2$,
\begin{align}
&\ \ \ \ \sum\limits_n {\mathbb{E}\left\{ {\left. {\left. {{{\hat {\hat Z}}_n}\left( {{\bf{\hat {\hat Q}}}\left( {{t_0}} \right)} \right)} \right|_{{t_0}}^{{t_0} + t - 1}} \right|{\bf{\hat {\hat Q}}}\left( {{t_0}} \right)} \right\}} \nonumber\\
& \ge \sum\limits_n {\mathbb{E}\left\{ {\left. {\frac{1}{t}\sum\limits_{i = 1}^{{{\tilde M}_n}\left( {{t_0},t} \right)} {{{\tilde Z}_n}\left( {i,{\bf{\hat {\hat Q}}}\left( {{t_0}} \right)} \right)} } \right|{\bf{\hat {\hat Q}}}\left( {{t_0}} \right)} \right\}} - \left[ {N{C_2}\left( t \right) + {C_1}\left( t \right) + \frac{\varepsilon }{8}\sum\limits_{n,c} {\hat {\hat Q}_n^{\left( c \right)}\left( {{t_0}} \right)} } \right]\nonumber\\
&\ge \sum\limits_n {\mathbb{E}\mathbb{}\left\{ {\left. {\left. {{{\tilde Z}_n}\left( {{\bf{\hat {\hat Q}}}\left( {{t_0}} \right)} \right)} \right|_{{t_0}}^{{t_0} + t - 1}} \right|{\bf{\hat {\hat Q}}}\left( {{t_0}} \right)} \right\}}  - \left[ {N{C_2}\left( t \right) + {C_1}\left( t \right) + \frac{\varepsilon }{8}\sum\limits_{n,c} {\hat {\hat Q}_n^{\left( c \right)}\left( {{t_0}} \right)} } \right].
\label{eq_key_metric_comparison_hathat_tilde2_1}
\end{align}

\subsection{Strong stability achieved under $\hat{\hat{Policy}}$}
\label{subsec_DIVBAR-RMIA_vs_DIVBAR-MIA_sub4}
Combining (\ref{eq_comparison_result_tilde_vs_star}) in Appendix \ref{subsec_DIVBAR-RMIA_vs_DIVBAR-MIA_sub2} and (\ref{eq_key_metric_comparison_hathat_tilde2_1}) in Appendix \ref{subsec_DIVBAR-RMIA_vs_DIVBAR-MIA_sub3}: if we let $t\ge \max\left\{\hat {\hat D}_2,\tilde D_2\right\}$, we have
\begin{align}
&\ \ \ \ \sum\limits_n {\mathbb{E}\left\{ {\left. {\left. {{{\hat {\hat Z}}_n}\left( {{\bf{\hat {\hat Q}}}\left( {{t_0}} \right)} \right)} \right|_{{t_0}}^{{t_0} + t - 1}} \right|{\bf{\hat {\hat Q}}}\left( {{t_0}} \right)} \right\}}\nonumber\\
&\ge \sum\limits_n {\mathbb{E}\left\{ {\left. {\left. {{Z'}_n^*\left( {{\bf{\hat {\hat Q}}}\left( {{t_0}} \right)} \right)} \right|_{{t_0}}^{{t_0} + t - 1}} \right|{\bf{\hat {\hat Q}}}\left( {{t_0}} \right)} \right\}}  - \left[N{C_2}\left( t \right) + {C_1}\left( t \right) + \frac{\varepsilon }{4}\sum\limits_{n,c} {\hat {\hat Q}_n^{\left( c \right)}\left( {{t_0}}, \right)}\right].
\label{eq_key_metric_comparison_hathat_vs_star}
\end{align}

Going back to (\ref{eq_upper_bound_lyapunov_drift_hathat}) and plugging (\ref{eq_key_metric_comparison_hathat_vs_star}) into it, and following the steps similar to the proof of Theorem \ref{thm: DIVBAR_MIA_throughput_optimal} shown in Appendix \ref{subsec_througput_optimal_sub4}, we get
\begin{equation}
\frac{1}{t}\sum\limits_{n,c} \mathbb{E} \left\{ {\left. {{{\left( {\hat {\hat Q}_n^{\left( c \right)}\left( {{t_0} + t} \right)} \right)}^2} - {{\left( {\hat {\hat Q}_n^{\left( c \right)}\left( {{t_0}} \right)} \right)}^2}} \right|{\bf{\hat {\hat Q}}}\left( {{t_0}} \right)} \right\}\le B\left(t\right) + C\left(t\right) - \frac{\varepsilon}{2} \sum\limits_{n,c} {\hat {\hat Q}_n^{\left( c \right)}\left( {{t_0}} \right)},
\label{eq_lyapunov_drift_bound2.7}
\end{equation}
where  $C\left(t\right)=4Nt\left(N+A_{\rm{max}}+1\right)$; $t \ge {\rm{max}}\left\{\hat D_2, \tilde D_2, D^*\right\}\buildrel \Delta \over = D$; in retrospect, $D^*$ is defined for (\ref{eq_bound_average_rate_term}) in Appendix \ref{appendix: proof_sufficiency_corollary1}. Now we set $t = D$, and rewrite (\ref{eq_lyapunov_drift_bound2.7}) as
\begin{equation}
\frac{1}{D}\sum\limits_{n,c} \mathbb{E} \left\{ {\left. {{{\left( {\hat {\hat Q}_n^{\left( c \right)}\left( {{t_0} + D} \right)} \right)}^2} - {{\left( {\hat {\hat Q}_n^{\left( c \right)}\left( {{t_0}} \right)} \right)}^2}} \right|{\bf{\hat {\hat Q}}}\left( {{t_0}} \right)} \right\}\le B\left(D\right) + C\left(D\right) - \frac{\varepsilon}{2} \sum\limits_{n,c} {\hat {\hat Q}_n^{\left( c \right)}\left( {{t_0}} \right)}.
\label{eq_lyapunov_drift_bound3}
\end{equation}
According to Lemma \ref{lemma: strong_stability} in Appendix \ref{appendix: proof_sufficiency_corollary1}, we finally get the result of the strong stability:
\begin{equation}
\mathop {\lim \sup }\limits_{t \to \infty } \frac{1}{t}\sum\limits_{\tau  = 0}^{t - 1} {\sum\limits_{n,c} {\mathbb{E}\left\{ {\hat {\hat Q}_n^{\left( c \right)}\left( \tau  \right)} \right\}} }  \le \frac{2\left[B\left(D\right)+C\left(D\right)\right]}{\varepsilon }.
\label{eq_strong_stability_hathat}
\end{equation}

As is shown in (\ref{eq_strong_stability_hathat}), DIVBAR-MIA ($\hat{\hat{Policy}}$) can support any exogenous input rate within the RMIA network capacity region $\Lambda_{\rm{RMIA}}$, i.e., $\exists \varepsilon >0$ such that $\left(\lambda_n^{\left(c\right)}+\varepsilon \right)\in \Lambda$, the mean time average backlog of the whole network under DIVBAR-MIA is upper bounded by a constant ${{2\left[ {B\left( D \right) + C\left( D \right)} \right]} \mathord{\left/
 {\vphantom {{\left[ {B\left( D \right) + C\left( D \right)} \right]} \varepsilon }} \right.
 \kern-\nulldelimiterspace} \varepsilon }$, which indicates the strong stability.

\section{Proof of Lemma \ref{lemma: uniformly convergence}}
\label{appendix: uniformly_convergence}
For later use, we first define the \emph{extended epoch of commodity $c$ on link $\left(n,k\right)$} as the time interval, whose starting timeslot is the first timeslot after a forwarding timeslot of a commodity $c$ packet from node $n$ to node $k$, and whose ending timeslot is the next forwarding timeslot of a commodity $c$ packet from node $n$ to node $k$. Here timeslot $0$ is the starting timeslot of the first extended epoch of all the commodities on link $\left(n,k\right)$, though there is no forwarding timeslot before timeslot $0$. With this definition, we can guarantee that all the receiving nodes of a transmitting node $n$ start accumulating the partial information of a commodity $c$ packet from zero at the beginning of each extended epoch on link $\left(n,k\right)$. Additionally, $b_{nk}^{\left(c\right)}\left(\tau\right)=1$ only when timeslot $\tau$ is the ending timeslot of an extended epoch of commodity $c$ on link $\left(n,k\right)$.

For a commodity $c$, considering the fact that the arbitrary timeslot $t_0$ must be located between the forwarding timeslots of two successive extended epochs of this commodity, the main strategy of the later proof is to find two suitable time averaging sequences, which respectively upper bound and lower bound the target time average sequence $\frac{1}{t}\sum\nolimits_{\tau  = {t_0}}^{{t_0} + t - 1} {\mathbb{E}\left\{ {b_{nk}^{\left( c \right)}\left( \tau  \right)} \right\}}$ but converge to the same value with a convergence speed independent of $t_0$.

Suppose up to timeslot $t_0$, $M_0$ units of commodity $c$ have been forwarded from node $n$ to node $k$, where $M_0 \ge 0$ (here $M_0=0$ means that the timeslot $t_0$ is within the first extended epoch of commodity $c$ on link $\left(n,k\right)$, whose starting timeslot is $0$). Therefore, timeslot $t_0$ must be located between the two forwarding timeslots of the $M_0$th and $\left(M_0+1\right)$th extended epochs of commodity $c$. To be specific, if ${\overline t}_{nk,i-1}^{\left(c\right)}$ represents the starting timeslot of the $i$th extended epoch of commodity $c$ on link $\left(n,k\right)$, and correspondingly, ${\overline t}_{nk,i}^{\left(c\right)}-1$ is the forwarding timeslot (ending timeslot) of $i$th extended epoch of commodity $c$ on link $\left(n,k\right)$, then we have
\begin{equation}
{\overline t}_{nk,M_0}^{\left(c\right)}-1\le t_0 < {\overline t}_{nk,M_0+1}^{\left(c\right)}-1,\ M_0\ge 1,
\label{eq_t_zero_location}
\end{equation}
and if $t_0$ is in the first extended epoch, then $0\le t_0 \le \overline t_{nk,1}^{\left(c\right)}$.

\begin{figure}
        \centering
        \subfigure[For lower bounding $\frac{1}{t}\sum\nolimits_{\tau  = 0}^{t - 1} {\mathbb{E}\left\{ {b_{nk}^{\left( c \right)}\left( \tau  \right)} \right\}}$: if ${t_0} \ge {\overline t}_{nk,M_0}^{\left(c\right)} $]{
        \includegraphics[width = 12cm]{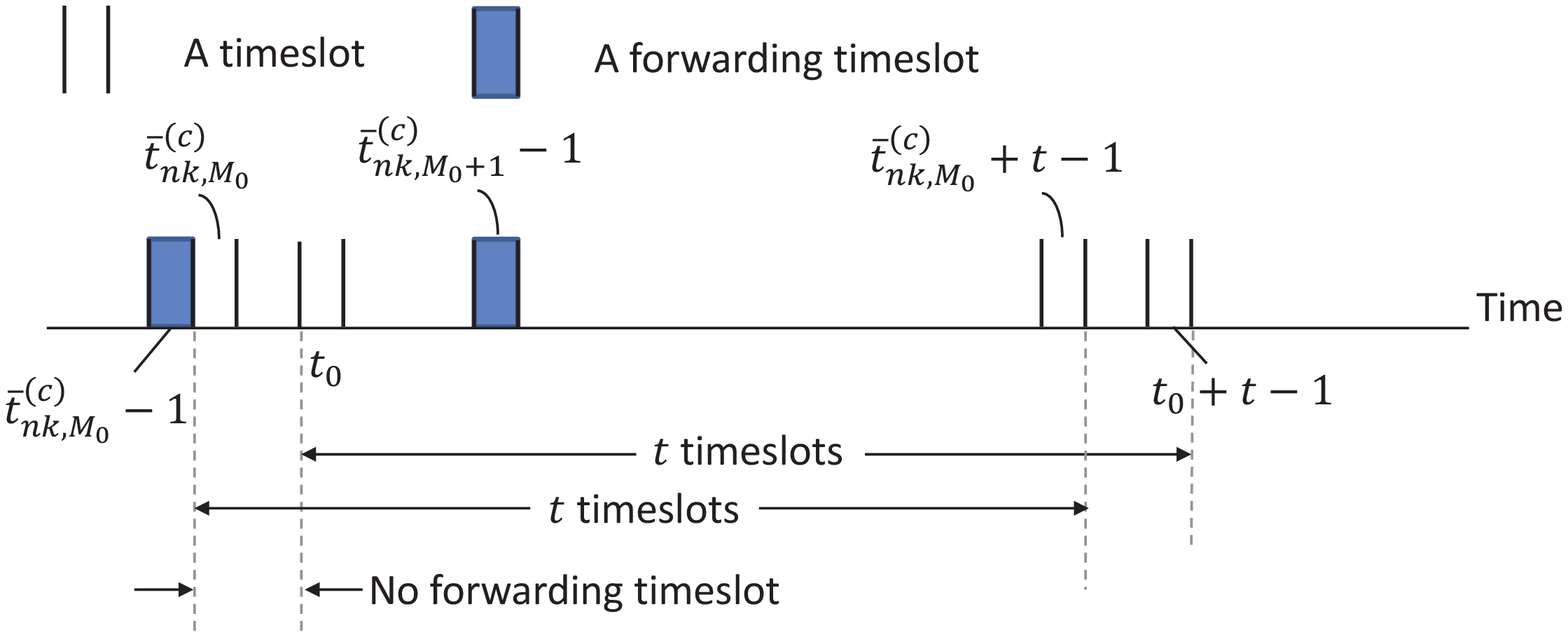}
        \label{fig_uniformly_convergence_a}
        }
        \subfigure[For lower bounding $\frac{1}{t}\sum\nolimits_{\tau  = 0}^{t - 1} {\mathbb{E}\left\{ {b_{nk}^{\left( c \right)}\left( \tau  \right)} \right\}}$: if ${t_0} = {\overline t}_{nk,M_0}^{\left(c\right)} - 1$]{
        \includegraphics[width = 12cm]{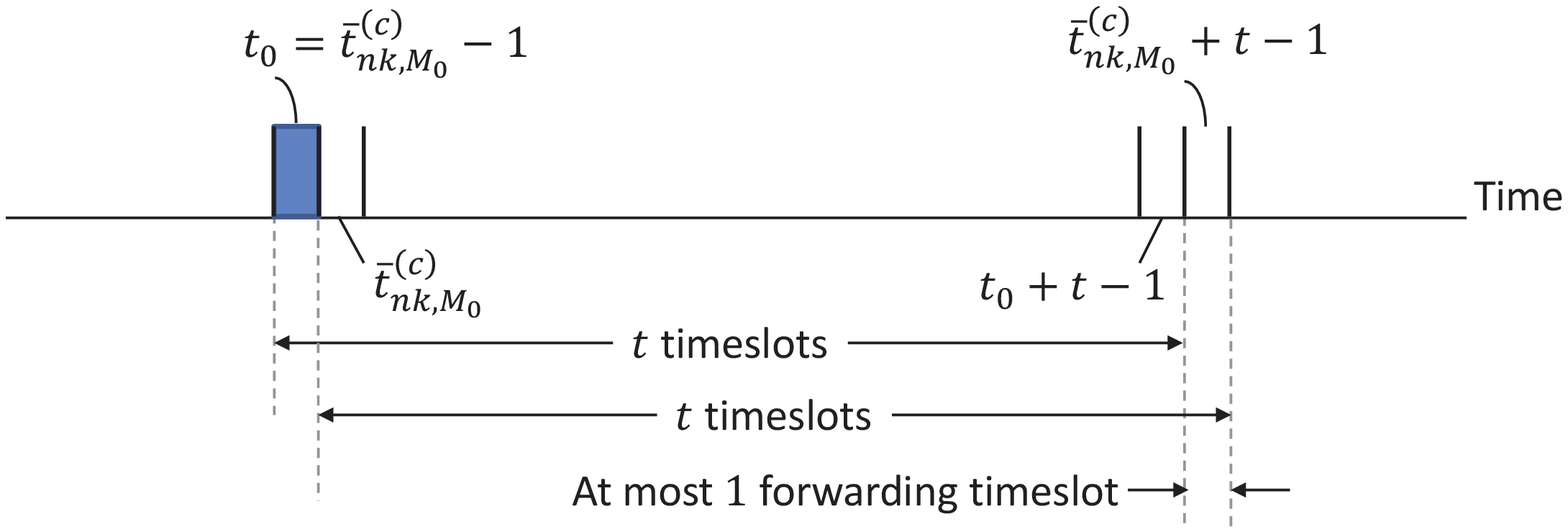}
        \label{fig_uniformly_convergence_b}
        }
        \subfigure[For upper bounding $\frac{1}{t}\sum\nolimits_{\tau  = 0}^{t - 1} {\mathbb{E}\left\{ {b_{nk}^{\left( c \right)}\left( \tau  \right)} \right\}}$]{
        \includegraphics[width = 12cm]{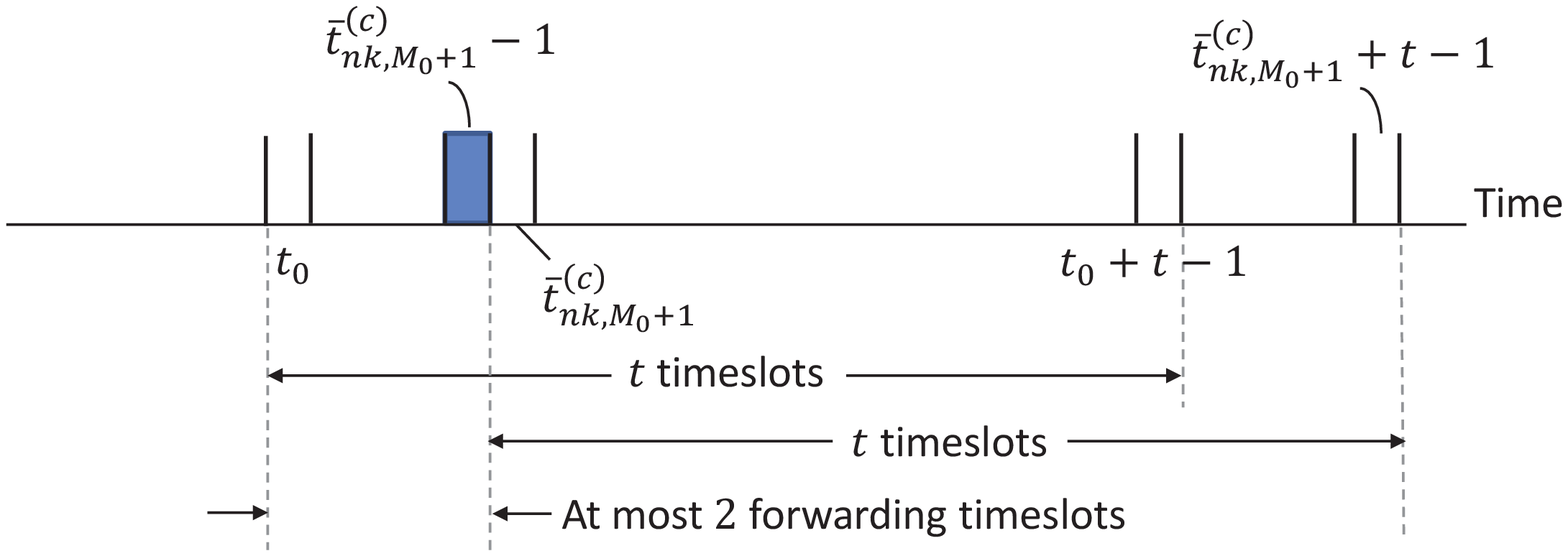}
        \label{fig_uniformly_convergence_c}
        }
        \caption{The timeslot axis illustrating the intervals relationships when lower bounding and upper bounding $\frac{1}{t}\sum\nolimits_{\tau  = 0}^{t - 1} {\mathbb{E}\left\{ {b_{nk}^{\left( c \right)}\left( \tau  \right)} \right\}}$}
\end{figure}

Furthermore, because of the renewal operation implemented at the end of each extended epoch, and each node stationarily uses the fixed probabilities to choose commodities to transmit and forward, any future transmitting and forwarding events respective to the beginning of timeslot ${\overline t}_{nk,i}^{\left(c\right)}$ has the identical distributions as those respective to the beginning of timeslot $0$. Because of this fact, $\frac{1}{t}\sum\nolimits_{\tau  = \overline t_{nk,i}^{\left( c \right)}}^{\overline t_{nk,i}^{\left( c \right)} + t - 1} {b_{nk}^{\left( c \right)}\left( \tau  \right)} $ must be identically distributed as $\frac{1}{t}\sum\nolimits_{\tau  = 0}^{t - 1} {b_{nk}^{\left( c \right)}\left( \tau  \right)}$, which can be expressed in the following form:
\begin{equation}
\frac{1}{t}\sum\limits_{\tau  = {0}}^{t - 1} {b_{nk}^{\left( c \right)}\left( \tau  \right)}  \sim \frac{1}{t}\sum\limits_{\tau  = \overline t_{nk,{i}}^{\left( c \right)}}^{\overline t_{nk,{i}}^{\left( c \right)} + t - 1} {b_{nk}^{\left( c \right)}\left( \tau  \right)},\ {\rm{for\ }}i\ge0,
\label{eq_uniform_convergence_identical_distribution}
\end{equation}
where $\sim$ means being identically distributed; when $i=0$, $t_{nk,0}^{\left(c\right)}=0$ and the two sides of (\ref{eq_uniform_convergence_identical_distribution}) are equal. The lower and upper bounding sequences shown in the later proof are closely related with the property shown in (\ref{eq_uniform_convergence_identical_distribution}), and are respectively described in Subsection \ref{subsec_uniform_convergence_lower_bound} and Subsection \ref{subsec_uniform_convergence_upper_bound}.

\subsection{Lower bounding $\frac{1}{t}\sum\nolimits_{\tau  = {t_0}}^{{t_0} + t - 1} {\mathbb{E}\left\{ {b_{nk}^{\left( c \right)}\left( \tau  \right)} \right\}}$}
\label{subsec_uniform_convergence_lower_bound}
We compare the values of the two summation terms: $\sum\nolimits_{\tau  = {t_0}}^{{t_0} + t - 1} {b_{nk}^{\left( c \right)}\left( \tau  \right)} $ and $\sum\nolimits_{\tau  = {\overline t}_{nk,M_0}^{\left( c \right)}}^{{\overline t}_{nk,M_0}^{\left( c \right)} + t - 1} {b_{nk}^{\left( c \right)}\left( \tau  \right)}$. The comparison should be discussed in two cases: $t_0 \ge {\overline t}_{nk,M_0}^{\left( c \right)}$ and $t_0={\overline t}_{nk,M_0}^{\left( c \right)}-1$, which are respectively shown in Fig. \ref{fig_uniformly_convergence_a} and Fig. \ref{fig_uniformly_convergence_b}.

\subsubsection{\textbf{Case} $t_0 \ge {\overline t}_{nk,M_0}^{\left( c \right)}$}
Based on the location of $t_0$ shown in Fig. \ref{fig_uniformly_convergence_a}, if $t_0 \ge {\overline t}_{nk,M_0}^{\left( c \right)}$, $\sum\nolimits_{\tau  = {t_0}}^{{t_0} + t - 1} {b_{nk}^{\left( c \right)}\left( \tau  \right)}$ can be lower bounded in the following form:
\begin{align}
\sum\limits_{\tau  = {t_0}}^{{t_0} + t - 1} {b_{nk}^{\left( c \right)}\left( \tau  \right)}  &= \sum\limits_{\tau  = \overline t_{nk,{M_0}}^{\left( c \right)}}^{\overline t_{nk,{M_0}}^{\left( c \right)} + t - 1} {b_{nk}^{\left( c \right)}\left( \tau  \right)}  - \sum\limits_{\tau  = \overline t_{nk,{M_0}}^{\left( c \right)}}^{{t_0} - 1} {b_{nk}^{\left( c \right)}\left( \tau  \right)}  + \sum\limits_{\tau  = \overline t_{nk,{M_0}}^{\left( c \right)} + t}^{{t_0} + t - 1} {b_{nk}^{\left( c \right)}\left( \tau  \right)}\nonumber\\
&\ge \sum\limits_{\tau  = \overline t_{nk,{M_0}}^{\left( c \right)}}^{\overline t_{nk,{M_0}}^{\left( c \right)} + t - 1} {b_{nk}^{\left( c \right)}\left( \tau  \right)} -0 + 0 = \sum\limits_{\tau  = \overline t_{nk,{M_0}}^{\left( c \right)}}^{\overline t_{nk,{M_0}}^{\left( c \right)} + t - 1} {b_{nk}^{\left( c \right)}\left( \tau  \right)},
\label{eq_uniform_convergence_lower_bound_1_4}
\end{align}
where neither of the summation terms $\sum\nolimits_{\tau  = \overline t_{nk,{M_0}}^{\left( c \right)}}^{{t_0} - 1} {b_{nk}^{\left( c \right)}\left( \tau  \right)}$ and $\sum\nolimits_{\tau  = \overline t_{nk,{M_0}}^{\left( c \right)} + t}^{{t_0} + t - 1} {b_{nk}^{\left( c \right)}\left( \tau  \right)} $ exists, if ${t_0} = \overline t_{nk,{M_0}}^{\left( c \right)}$. From now on, for any summation term $\sum\nolimits_{\tau  = x}^y {f\left(\tau\right)}$, if $y<x$, the summation value is zero.

\subsubsection{\textbf{Case} $t_0 = {\overline t}_{nk,M_0}^{\left( c \right)}-1$}
Based on the location of $t_0$ shown in Fig. \ref{fig_uniformly_convergence_b}, if $t_0 = {\overline t}_{nk,M_0}^{\left( c \right)}-1$, $\sum\nolimits_{\tau  = {t_0}}^{{t_0} + t - 1} {b_{nk}^{\left( c \right)}\left( \tau  \right)}$ can be written as follows:
\begin{align}
\sum\limits_{\tau  = {t_0}}^{{t_0} + t - 1} {b_{nk}^{\left( c \right)}\left( \tau  \right)}  &= \sum\limits_{\tau  = \overline t_{nk,{M_0}}^{\left( c \right)}}^{\overline t_{nk,{M_0}}^{\left( c \right)} + t - 1} {b_{nk}^{\left( c \right)}\left( \tau  \right)}  + b_{nk}^{\left( c \right)}\left( {\overline t_{nk,{M_0}}^{\left( c \right)} - 1} \right) - b_{nk}^{\left( c \right)}\left( {\overline t_{nk,{M_0}}^{\left( c \right)} + t - 1} \right)\nonumber\\
&\ge\sum\limits_{\tau  = \overline t_{nk,{M_0}}^{\left( c \right)}}^{\overline t_{nk,{M_0}}^{\left( c \right)} + t - 1} {b_{nk}^{\left( c \right)}\left( \tau  \right)}  + 1 - 1 = \sum\limits_{\tau  = \overline t_{nk,{M_0}}^{\left( c \right)}}^{\overline t_{nk,{M_0}}^{\left( c \right)} + t - 1} {b_{nk}^{\left( c \right)}\left( \tau  \right)}.
\label{eq_uniform_convergence_lower_bound_2_2}
\end{align}\\

Apart from the two cases 1) and 2), on the other hand, when $t_0=0$, the target time average sequence reduces to $\frac{1}{t}\sum\nolimits_{\tau  = 0}^{t - 1} {b_{nk}^{\left( c \right)}\left( \tau  \right)} $, which approaches to $b_{nk}^{\left(c\right)}$ with probability 1 as $t\rightarrow \infty$. According to the dominated convergence property (see Ref. \cite{domenate_convergence_Wong_Hajek_1985}, Chapter 1), for $\forall \varepsilon >0$, there exists an integer $D_{nk,0}^{\left(c\right)}$, such that $\forall t\ge D_{nk,0}^{\left(c\right)}$, we have
\begin{equation}
\left| {\frac{1}{t}\sum\limits_{\tau  = 0}^{t - 1} {\mathbb{E}\left\{ {b_{nk}^{\left( c \right)}\left( \tau  \right)} \right\}}  - b_{nk}^{\left( c \right)}} \right| \le \frac{\varepsilon}{2} ,
\label{eq_convergence_average_from_zero}
\end{equation}

Going back to case 1) and 2), both results shown as (\ref{eq_uniform_convergence_lower_bound_1_4}) and (\ref{eq_uniform_convergence_lower_bound_2_2}) demonstrate the same comparison relationship. Then combining (\ref{eq_uniform_convergence_lower_bound_1_4}) and (\ref{eq_uniform_convergence_lower_bound_2_2}) with (\ref{eq_uniform_convergence_identical_distribution}) and (\ref{eq_convergence_average_from_zero}), it follows that, for $\forall t\ge D_{nk,0}^{\left(c\right)}\buildrel \Delta \over = D_{nk,\rm{lower}}^{\left(c\right)}$,
\begin{equation}
\frac{1}{t}\sum\limits_{\tau  = {t_0}}^{{t_0} + t - 1} {\mathbb{E}\left\{ {b_{nk}^{\left( c \right)}\left( \tau  \right)} \right\}}\ge \frac{1}{t}\mathbb{E}\left\{ {\sum\limits_{\tau  = \overline t_{nk,{M_0}}^{\left( c \right)}}^{\overline t_{nk,{M_0}}^{\left( c \right)} + t - 1} {b_{nk}^{\left( c \right)}\left( \tau  \right)} } \right\} =\frac{1}{t}\sum\limits_{\tau  = 0}^{t - 1} {\mathbb{E}\left\{ {b_{nk}^{\left( c \right)}\left( \tau  \right)} \right\}} \ge b_{nk}^{\left( c \right)} - \frac{\varepsilon}{2}>b_{nk}^{\left( c \right)}-\varepsilon,
\label{eq_uniform_convergence_lower_bound_5}
\end{equation}
which lower bounds the target time average sequence.

\subsection{Upper bounding $\frac{1}{t}\sum\nolimits_{\tau  = {t_0}}^{{t_0} + t - 1} {\mathbb{E}\left\{ {b_{nk}^{\left( c \right)}\left( \tau  \right)} \right\}}$}
\label{subsec_uniform_convergence_upper_bound}
We compare the values of two summation terms: $\sum\nolimits_{\tau  = {t_0}}^{{t_0} + t - 1} {b_{nk}^{\left( c \right)}\left( \tau  \right)} $ and $\sum\nolimits_{\tau  = {\overline t}_{nk,M_0+1}^{\left( c \right)}}^{{\overline t}_{nk,M_0+1}^{\left( c \right)} + t - 1} {b_{nk}^{\left( c \right)}\left( \tau  \right)}$. The time interval relationship of the two comparison terms are shown in Fig. \ref{fig_uniformly_convergence_c}.

Based on the location of timeslot $t_0$ shown in Fig. \ref{fig_uniformly_convergence_c}, $\sum\nolimits_{\tau  = {t_0}}^{{t_0} + t - 1} {b_{nk}^{\left( c \right)}\left( \tau  \right)} $ can be upper bounded as follows:
\begin{align}
\sum\limits_{\tau  = {t_0}}^{{t_0} + t - 1} {b_{nk}^{\left( c \right)}\left( \tau  \right)}  &= \sum\limits_{\tau  = \overline t_{nk,{M_0} + 1}^{\left( c \right)}}^{\overline t_{nk,{M_0} + 1}^{\left( c \right)} + t - 1} {b_{nk}^{\left( c \right)}\left( \tau  \right)}  + \sum\limits_{\tau  = {t_0}}^{\overline t_{nk,{M_0} + 1}^{\left( c \right)} - 1} {b_{nk}^{\left( c \right)}\left( \tau  \right)}  - \sum\limits_{\tau  = {t_0} + t}^{\overline t_{nk,{M_0} + 1}^{\left( c \right)} + t - 1} {b_{nk}^{\left( c \right)}\left( \tau  \right)}\nonumber\\
&\le \sum\limits_{\tau  = \overline t_{nk,{M_0} + 1}^{\left( c \right)}}^{\overline t_{nk,{M_0} + 1}^{\left( c \right)} + t - 1} {b_{nk}^{\left( c \right)}\left( \tau  \right)}  + 2 - 0 = \sum\limits_{\tau  = \overline t_{nk,{M_0} + 1}^{\left( c \right)}}^{\overline t_{nk,{M_0} + 1}^{\left( c \right)} + t - 1} {b_{nk}^{\left( c \right)}\left( \tau  \right)}  + 2.
\label{eq_uniform_convergence_upper_bound4}
\end{align}

On the other hand, $\forall t\ge \left\lceil {{4 \mathord{\left/{\vphantom {4 \varepsilon }} \right.\kern-\nulldelimiterspace} \varepsilon }} \right\rceil $, we have ${2 \mathord{\left/{\vphantom {2 t}} \right.\kern-\nulldelimiterspace} t} \le {\varepsilon  \mathord{\left/{\vphantom {\varepsilon  2}} \right.\kern-\nulldelimiterspace} 2}$. Then combining this result and (\ref{eq_convergence_average_from_zero}) with the result shown in (\ref{eq_uniform_convergence_upper_bound4}) yields: $\forall t \ge \max\left\{D_{nk,0}^{\left(c\right)}, \left\lceil {{4 \mathord{\left/{\vphantom {4 \varepsilon }} \right.\kern-\nulldelimiterspace} \varepsilon }} \right\rceil \right\} \buildrel \Delta \over = D_{nk,\rm{upper}}^{\left(c\right)}$,
\begin{align}
\frac{1}{t}\sum\limits_{\tau  = {t_0}}^{{t_0} + t - 1} {\mathbb{E}\left\{ {b_{nk}^{\left( c \right)}\left( \tau  \right)} \right\}}  &\le \frac{1}{t}\mathbb{E}\left\{ {\sum\limits_{\tau  = \overline t_{nk,{M_0} + 1}^{\left( c \right)}}^{\overline t_{nk,{M_0} + 1}^{\left( c \right)} + t - 1} {b_{nk}^{\left( c \right)}\left( \tau  \right)} } \right\} + \frac{2}{t} \nonumber\\
&\le \left[b_{nk}^{\left( c \right)} + \frac{\varepsilon }{2}\right] + \frac{\varepsilon }{2} = b_{nk}^{\left( c \right)} + \varepsilon,
\label{eq_uniform_convergence_upper_bound8}
\end{align}
which upper bounds the target time average sequence.\\

Summarizing, combining the lower bound and upper bound of the target average sequence, shown as (\ref{eq_uniform_convergence_lower_bound_5}) and (\ref{eq_uniform_convergence_upper_bound8}) respectively in Subsection \ref{subsec_uniform_convergence_lower_bound} and \ref{subsec_uniform_convergence_upper_bound}, for $\forall t \ge \max \left\{ {D_{nk,{\rm{lower}}}^{\left( c \right)},D_{nk,{\rm{upper}}}^{\left( c \right)}} \right\} \buildrel \Delta \over = D_{nk}^{\left( c \right)}$, we have
\begin{equation}
\left| {\frac{1}{t}\sum\limits_{\tau  = {t_0}}^{{t_0} + t - 1} {\mathbb{E}\left\{ {b_{nk}^{\left( c \right)}\left( \tau  \right)} \right\}}  - b_{nk}^{\left( c \right)}} \right| \le \varepsilon.
\label{eq_uniform_convergence_upper_bound9}
\end{equation}
Here note that if $\varepsilon$ is fixed, $D_{nk}^{\left(c\right)}$ is fixed for arbitrary $t_0$ because both $D_{nk,{\rm{lower}}}^{\left( c \right)}$ and $D_{nk,{\rm{upper}}}^{\left( c \right)}$ depend only on $D_{nk,0}^{\left(c\right)}$ and $\varepsilon$, where $D_{nk,0}^{\left(c\right)}$ depends on the convergence speed of the sequence $\frac{1}{t}\sum\nolimits_{\tau  = 0}^{t - 1} {b_{nk}^{\left( c \right)}\left( \tau  \right)}$ and $\varepsilon$. Therefore, the value of the integer $D_{nk}^{\left(c\right)}$ is independent of $t_0$.

\section{Proof of Lemma \ref{lemma: strong_stability}}
\label{appendix: stong_stability}
Taking expectation over ${\bf{Q}}\left(t_0\right)$ on both sides of (\ref{eq_D_step_lyapunov_drift}), we get
\begin{equation}
\frac{1}{d}\sum\limits_{n,c} {\mathbb{E}\left\{ {{{\left( {Q_n^{\left( c \right)}\left( {{t_0} + d} \right)} \right)}^2}} \right\}}  - \mathbb{E}\left\{ {{{\left( {Q_n^{\left( c \right)}\left( {{t_0}} \right)} \right)}^2}} \right\} \le B\left( d \right) - \varepsilon \sum\limits_{n,c} {\mathbb{E}\left\{ {{{ {Q_n^{\left( c \right)}\left( {{t_0}} \right)} }}} \right\}}.
\label{eq_D_step_lyapunov_drift1}
\end{equation}
Writing (\ref{eq_D_step_lyapunov_drift1}) for all timeslots $0,1,2,\cdots,t-1$ and doing concatenated summations, it follows that
\begin{equation}
\frac{1}{{dt}}\sum\limits_{\tau  = t}^{d + t - 1} {\mathbb{E}\left\{ {{{\left( {Q_n^{\left( c \right)}\left( \tau  \right)} \right)}^2}} \right\}}  - \frac{1}{{dt}}\sum\limits_{\tau  = 0}^{d - 1} {\mathbb{E}\left\{ {{{\left( {Q_n^{\left( c \right)}\left( \tau  \right)} \right)}^2}} \right\}}  \le B\left( d \right) - \varepsilon \frac{1}{t}\sum\limits_{\tau  = 0}^{t - 1} {\sum\limits_{n,c} {\mathbb{E}\left\{ {Q_n^{\left( c \right)}\left( \tau  \right)} \right\}} }.
\end{equation}
Dropping the non-negative term $\frac{1}{{dt}}\sum\limits_{\tau  = t}^{d + t - 1} {\mathbb{E}\left\{ {{{\left( {Q_n^{\left( c \right)}\left( \tau  \right)} \right)}^2}} \right\}}$ on the left hand side of the above inequality and letting $t\rightarrow \infty$, it follows that
\begin{equation}
0 =  - \mathop {\lim \sup }\limits_{t \to \infty } \frac{1}{{dt}}\sum\limits_{\tau  = 0}^{d - 1} {\mathbb{E}\left\{ {{{\left( {Q_n^{\left( c \right)}\left( \tau  \right)} \right)}^2}} \right\}}  \le B\left( d \right) - \varepsilon \mathop {\lim \sup }\limits_{t \to \infty } \frac{1}{t}\sum\limits_{\tau  = 0}^{t - 1} {\sum\limits_{n,c} {\mathbb{E}\left\{ {Q_n^{\left( c \right)}\left( \tau  \right)} \right\}} },
\end{equation}
and then strong stability is achieved:
\begin{equation}
\label{eq_mean_time_average_backlog_bounded_proof}
\mathop {\lim \sup }\limits_{t \to \infty } \frac{1}{t}\sum\limits_{\tau  = 0}^{t - 1} {\sum\limits_{n,c} {\mathbb{E}\left\{ {Q_n^{\left( c \right)}\left( \tau  \right)} \right\}} }  \le \frac{B\left( d \right)}{\varepsilon }.
\end{equation}

\section{Proof of Lemma \ref{lemma: hat_vs_tildetilde}}
\label{appendix: hat_vs_tildetilde}
In this proof, there are three steps to finish the comparison between $\hat{Policy}$ and $\tilde{\tilde{Policy}}$: firstly, transform the analysis focus from the key metric $\sum\limits_n {\mathbb{E}\left\{ {\left. {\left. {{{\hat Z}_n}\left( {{\bf{\hat Q}}\left( {{t_0}} \right)} \right)} \right|_{{t_0}}^{{t_0} + t - 1}} \right|{\bf{\hat Q}}\left( {{t_0}} \right)} \right\}}$ to a new metric that is easy to manipulate in the later policy comparison; secondly, based on the new metric, compare $\hat{Policy}$ and $\tilde{\tilde{Policy}}$ over single epoch; thirdly, extend the comparison to multiple epochs.

\subsection{Lower bounding the original key backpressure metric by an expression consisting of a new metric under $\hat{Policy}$}
\label{subsec_hat_vs_tildetilde_1}
Based on Fig. \ref{fig_epoch relations among several policies} in Appendix \ref{appendix: policy_list}, we start by rewriting the expression of the key backpressure metric under $\hat {Policy}$ over the interval $\left[t_0,t_0+t-1\right]$ into the following form:
\begin{align}
&\ \ \ \ \sum\limits_n {\mathbb{E}\left\{ {\left. {\left. {{{\hat Z}_n}\left( {{\bf{\hat Q}}\left( {{t_0}} \right)} \right)} \right|_{{t_0}}^{{t_0} + t - 1}} \right|{\bf{\hat Q}}\left( {{t_0}} \right)} \right\}} \nonumber\\
&= \sum\limits_n {\mathbb{E}\left\{ {\frac{1}{t}\sum\limits_{\tau  = {t_0}}^{{u_{n,2}} - 1} {\sum\limits_c {\sum\limits_{k \in {{\cal K}_n}} {\hat b_{nk}^{\left( c \right)}\left( \tau  \right)\left[ {\hat Q_n^{\left( c \right)}\left( {{t_0}} \right) - \hat Q_k^{\left( c \right)}\left( {{t_0}} \right)} \right]} } } } \right.}  + \frac{1}{t}\sum\limits_{i = 2}^{{M_n}\left( {{t_0},t} \right) - 1} {{{\hat Z}_n}\left( {i,{\bf{\hat Q}}\left( {{t_0}} \right)} \right)}  \nonumber
\end{align}
\begin{align}
&+ \left. {\left. { \frac{1}{t}\sum\limits_{\tau  = {u_{n,{M_n}\left( {{t_0},t} \right)}}}^{{t_0} + t - 1} {\sum\limits_c {\sum\limits_{k \in {{\cal K}_n}} {\hat b_{nk}^{\left( c \right)}\left( \tau  \right)\left[ {\hat Q_n^{\left( c \right)}\left( {{t_0}} \right) - \hat Q_k^{\left( c \right)}\left( {{t_0}} \right)} \right]} } } } \right|{\bf{\hat Q}}\left( {{t_0}} \right)} \right\},
\label{eq_key_metric_hat1}
\end{align}
where
\begin{equation}
{{ \hat Z}_n}\left( {i,{\bf{\hat Q}}\left( x \right)} \right) = \sum\limits_{\tau  = {u_{n,i}}}^{{u_{n,i + 1}} - 1} {\sum\limits_c {\sum\limits_{k \in {{\cal K}_n}} { \hat b_{nk}^{\left( c \right)}\left( \tau  \right)\left[ {\hat Q_n^{\left( c \right)}\left( x \right) - \hat Q_k^{\left( c \right)}\left( x \right)} \right]} } },{\rm{\ for\ }}1\le i \le M_n\left(t_0,t\right).
\end{equation}
For the first summation term over $\left[t_0, u_{n,2}-1\right]$ in (\ref{eq_key_metric_hat1}), since $t_0$ must be located within epoch $1$, we can get
\begin{equation}
\frac{1}{t}\sum\limits_{\tau  = {t_0}}^{{u_2} - 1} {\sum\limits_c {\sum\limits_{k \in {{\cal K}_n}} {\hat b_{nk}^{\left( c \right)}\left( \tau  \right)\left[ {\hat Q_n^{\left( c \right)}\left( {{t_0}} \right) - \hat Q_k^{\left( c \right)}\left( {{t_0}} \right)} \right]} } }  = \frac{1}{t}\sum\limits_{\tau  = {u_1}}^{{u_2} - 1} {\sum\limits_c {\sum\limits_{k \in {{\cal K}_n}} {\hat b_{nk}^{\left( c \right)}\left( \tau  \right)\left[ {\hat Q_n^{\left( c \right)}\left( {{t_0}} \right) - \hat Q_k^{\left( c \right)}\left( {{t_0}} \right)} \right]} } }.
\label{eq_key_metric_hat_term1}
\end{equation}
For the third summation term over $\left[u_{n,M_n\left(t_0,t\right)}, t_0+t-1\right]$ in (\ref{eq_key_metric_hat1}), our goal is to lower bound it. Firstly, note that as $t$ grows, $\forall t \geq \left\lceil {{8 \mathord{\left/{\vphantom {8 \varepsilon }} \right.\kern-\nulldelimiterspace} \varepsilon }} \right\rceil \buildrel \Delta \over =\hat D_1$, we have ${1 \mathord{\left/{\vphantom {1 t}} \right.\kern-\nulldelimiterspace} t} \le {\varepsilon  \mathord{\left/{\vphantom {\varepsilon  8}} \right.\kern-\nulldelimiterspace} 8}$. Then it follows that, $\forall t\ge \hat{D}_1$,
\begin{align}
\frac{1}{t}\sum\limits_{\tau  = {t_0} + t}^{{u_{{M_n}\left( {{t_0},t} \right) + 1}} - 1} {\sum\limits_c {\sum\limits_{k \in {{\cal K}_n}} {\hat b_{nk}^{\left( c \right)}\left( \tau  \right)\left[ {\hat Q_n^{\left( c \right)}\left( {{t_0}} \right) - \hat Q_k^{\left( c \right)}\left( {{t_0}} \right)} \right]} } }   & \le \frac{1}{t}\sum\limits_{\tau  = {u_{n,{M_n}\left( {{t_0},t} \right)}}}^{{u_{n,{M_n}\left( {{t_0},t} \right) + 1}} - 1} {\sum\limits_c {\sum\limits_{k \in {{\cal K}_n}} {\hat b_{nk}^{\left( c \right)}\left( \tau  \right)\left[ {\hat Q_n^{\left( c \right)}\left( {{t_0}} \right) - \hat Q_k^{\left( c \right)}\left( {{t_0}} \right)} \right]} } }  \nonumber\\
&\le \frac{1}{t}{\sum\limits_{c} {\hat Q_n^{\left( c \right)}\left( {{t_0}} \right)} } \le \frac{\varepsilon }{8}\sum\limits_{c} {\hat Q_n^{\left( c \right)}\left( {{t_0}} \right)}.
\label{eq_key_metric_hat_term1.5}
\end{align}
Based on (\ref{eq_key_metric_hat_term1.5}), we can lower bound the third summation term in (\ref{eq_key_metric_hat1}) for each node $n$ as follows:
\begin{align}
&\ \ \ \ {\frac{1}{t}\sum\limits_{\tau  = {u_{n,{M_n}\left( {{t_0},t} \right) }}}^{{t_0} + t - 1} {\sum\limits_c {\sum\limits_{k \in {{\cal K}_n}} {\hat b_{nk}^{\left( c \right)}\left( \tau  \right)\left[ {\hat Q_n^{\left( c \right)}\left( {{t_0}} \right) - \hat Q_k^{\left( c \right)}\left( {{t_0}} \right)} \right]} } } }\nonumber\\
&=\frac{1}{t}\sum\limits_{\tau  = {u_{n,{M_n}\left( {{t_0},t} \right)}}}^{{u_{n,{M_n}\left( {{t_0},t} \right) + 1}} - 1} {\sum\limits_c {\sum\limits_{k \in {{\cal K}_n}} {\hat b_{nk}^{\left( c \right)}\left( \tau  \right)\left[ {\hat Q_n^{\left( c \right)}\left( {{t_0}} \right) - \hat Q_k^{\left( c \right)}\left( {{t_0}} \right)} \right]} } }  - \frac{1}{t}\sum\limits_{\tau  = {t_0} + t}^{{u_{n,{M_n}\left( {{t_0},t} \right) + 1}} - 1} {\sum\limits_c {\sum\limits_{k \in {{\cal K}_n}} {\hat b_{nk}^{\left( c \right)}\left( \tau  \right)\left[ {\hat Q_n^{\left( c \right)}\left( {{t_0}} \right) - \hat Q_k^{\left( c \right)}\left( {{t_0}} \right)} \right]} } }\nonumber\\
&\geq\frac{1}{t}{{\hat Z}_n}\left( {M_n\left(t_0,t\right),{\bf{\hat Q}}\left( {{t_0}} \right)} \right) -\frac{\varepsilon }{8}\sum\limits_{c} {\hat Q_n^{\left( c \right)}\left( {{t_0}} \right)}.
\label{eq_key_metric_hat_term3}
\end{align}
Then plug (\ref{eq_key_metric_hat_term1}) and (\ref{eq_key_metric_hat_term3}) into (\ref{eq_key_metric_hat1}) to get
\begin{align}
\sum\limits_n {\mathbb{E}\left\{ {\left. {\left. {{{\hat Z}_n}\left( {{\bf{\hat Q}}\left( {{t_0}} \right)} \right)} \right|_{{t_0}}^{{t_0} + t - 1}} \right|{\bf{\hat Q}}\left( {{t_0}} \right)} \right\}}\ge \sum\limits_n {\mathbb{E}\left\{ {\left. {\frac{1}{t}\sum\limits_{i = 1}^{{M_n}\left( {{t_0},t} \right)} {{{\hat Z}_n}\left( {i,{\bf{\hat Q}}\left( {{t_0}} \right)} \right)} } \right|{\bf{\hat Q}}\left( {{t_0}} \right)} \right\} - } \frac{\varepsilon }{8}\sum\limits_{n,c} {\hat Q_n^{\left( c \right)}\left( {{t_0}} \right)},
\label{eq_key_metric_hat2}
\end{align}

To facilitate the later proof of comparing $\hat {Policy}$ with $\tilde{\tilde{Policy}}$, we do the following transformation on the expectation term in (\ref{eq_key_metric_hat2}) by switching the backlog state coefficients:
\begin{align}
&\ \ \ \ \sum\limits_n {\mathbb{E}\left\{ {\left. {\frac{1}{t}\sum\limits_{i = 1}^{{M_n}\left( {{t_0},t} \right)} {{{\hat Z}_n}\left( {i,{\bf{\hat Q}}\left( {{t_0}} \right)} \right)} } \right|{\bf{\hat Q}}\left( {{t_0}} \right)} \right\}}\nonumber\\
&=\sum\limits_n {\mathbb{E}\left\{ {\left. {\frac{1}{t}\sum\limits_{i = 1}^{{M_n}\left( {{t_0},t} \right)} {{{\hat Z}_n}\left( {i,{\bf{\hat Q}}\left( {{u_i}} \right)} \right)} } \right|{\bf{\hat Q}}\left( {{t_0}} \right)} \right\}}  - \sum\limits_n {\mathbb{E}\left\{ {\left. {\frac{1}{t}\sum\limits_{i = 1}^{{M_n}\left( {{t_0},t} \right)} {\left[ {{{\hat Z}_n}\left( {i,{\bf{\hat Q}}\left( {{u_i}} \right)} \right) - {{\hat Z}_n}\left( {i,{\bf{\hat Q}}\left( {{t_0}} \right)} \right)} \right]} } \right|{\bf{\hat Q}}\left( {{t_0}} \right)} \right\}}.
\label{eq_key_metric_hat3}
\end{align}
The difference term caused by the switch on the right hand side of (\ref{eq_key_metric_hat3}) can be written as
\begin{align}
&\sum\limits_n {\mathbb{E}\left\{ {\left. {\frac{1}{t}\sum\limits_{i = 1}^{{M_n}\left( {{t_0},t} \right)} {\left[ {{{\hat Z}_n}\left( {i,{\bf{\hat Q}}\left( {{u_i}} \right)} \right) - {{\hat Z}_n}\left( {i,{\bf{\hat Q}}\left( {{t_0}} \right)} \right)} \right]} } \right|{\bf{\hat Q}}\left( {{t_0}} \right)} \right\}}\ \ \ \ \ \ \ \ \ \ \ \ \ \ \ \ \ \ \ \ \ \ \ \ \ \ \ \ \ \ \ \ \ \ \ \ \ \ \ \ \ \ \ \ \ \ \ \ \ \nonumber
\end{align}
\begin{align}
&= \sum\limits_n {\mathbb{E}\left\{ {\left. {\frac{1}{t}\sum\limits_{i = 1}^{{M_n}\left( {{t_0},t} \right)} {\sum\limits_{\tau  = {u_{n,i}}}^{{u_{n,i + 1}} - 1} {\sum\limits_c {\sum\limits_{k \in {{\cal K}_n}} {\hat b_{nk}^{\left( c \right)}\left( \tau  \right)\left[ {\hat Q_n^{\left( c \right)}\left( {{u_{n,i}}} \right) - \hat Q_n^{\left( c \right)}\left( {{t_0}} \right) + \hat Q_k^{\left( c \right)}\left( {{t_0}} \right) - \hat Q_k^{\left( c \right)}\left( {{u_{n,i}}} \right)} \right]} } } } } \right|{\bf{\hat Q}}\left( {{t_0}} \right)} \right\}}.
\label{eq_key_metric_difference}
\end{align}
In (\ref{eq_key_metric_difference}), $u_{n,2},\cdots,u_{n,M_n\left(t_0,t\right)}$ are within the interval $\left[t_0,t_0+t-1\right]$, while $u_{n,1}$ may be located before $t_0$, but $t_0-u_{n,1}+1\le T_n\left(i\right)$. Then we have the following relationship:
\begin{equation}
\hat Q_n^{\left( c \right)}\left( {{u_{n,i}}} \right) - \hat Q_n^{\left( c \right)}\left( {{t_0}} \right) \le \sum\limits_{\tau  = {t_0}}^{{u_{n,i}} - 1} {\left[ {\sum\limits_{k \in {{\cal K}_n}} {\hat b_{kn}^{\left( c \right)}\left( \tau  \right)}  + a_n^{\left( c \right)}\left( \tau  \right)} \right]}  \le t\left( {N + {A_{\max }}} \right),{\rm{\ for\ }}2 \le i \le {M_n}\left( {{t_0},t} \right);
\label{eq_backlog_difference1}
\end{equation}
\begin{equation}
\hat Q_k^{\left( c \right)}\left( {{t_0}} \right) - \hat Q_k^{\left( c \right)}\left( {{u_{n,i}}} \right) \le \sum\limits_{\tau  = {t_0}}^{{u_{n,i}} - 1} {\sum\limits_{m \in {{\cal K}_k}} {\hat b_{km}^{\left( c \right)}\left( \tau  \right)} }  \le t,{\rm{\ for\ }}2 \le i \le {M_n}\left( {{t_0},t} \right),{\rm{\ }}k \in {{\cal K}_n};
\label{eq_backlog_difference2}
\end{equation}
\begin{equation}
\hat Q_n^{\left( c \right)}\left( {{u_{n,1}}} \right) - \hat Q_n^{\left( c \right)}\left( {{t_0}} \right) \le \sum\limits_{\tau  = {u_{n,1}}}^{{t_0} - 1} {\sum\limits_{k \in {{\cal K}_n}} {\hat b_{nk}^{\left( c \right)}\left( \tau  \right)} }  \le t_0-u_{n,1};
\label{eq_backlog_difference3}
\end{equation}
\begin{equation}
\hat Q_k^{\left( c \right)}\left( {{t_0}} \right) - \hat Q_k^{\left( c \right)}\left( {{u_{n,1}}} \right) \le \sum\limits_{\tau  = {u_{n,1}}}^{{t_0} - 1} {\left[ {\sum\limits_{m \in {{\cal K}_k}} {\hat b_{mk}^{\left( c \right)}\left( \tau  \right)}  + a_k^{\left( c \right)}\left( \tau  \right)} \right]}  \le \left( {{t_0} - {u_{n,1}}} \right)\left( {N + {A_{\max }}} \right),{\rm{\ for\ }}k\in {\cal K}_n.
\label{eq_backlog_difference4}
\end{equation}
Plug (\ref{eq_backlog_difference1})-(\ref{eq_backlog_difference4}) into (\ref{eq_key_metric_difference}), and note that $\sum\limits_{\tau  = {u_{n,i}}}^{{u_{n,i + 1}} - 1} {\sum\limits_c {\sum\limits_{k \in {{\cal K}_n}} {\hat b_{nk}^{\left( c \right)}\left( \tau  \right)} } }  \le 1$. Then (\ref{eq_key_metric_difference}) can be upper bounded as follows:
\begin{align}
&\ \ \ \ \sum\limits_n {\mathbb{E}\left\{ {\left. {\frac{1}{t}\sum\limits_{i = 1}^{{M_n}\left( {{t_0},t} \right)} {\left[ {{{\hat Z}_n}\left( {i,{\bf{\hat Q}}\left( {{u_{n,i}}} \right)} \right) - {{\hat Z}_n}\left( {i,{\bf{\hat Q}}\left( {{t_0}} \right)} \right)} \right]} } \right|{\bf{\hat Q}}\left( {{t_0}} \right)} \right\}}\nonumber\\
&\le \sum\limits_n {\left[ {\frac{1}{t}\left( {N + {A_{\max }} + 1} \right)\mathbb{E}\left\{ {{t_0} - {u_{n,1}}} \right\} + \left( {N + {A_{\max }} + 1} \right)\mathbb{E}\left\{ {{M_n}\left( {{t_0},t} \right)-1} \right\}} \right]},
\label{eq_key_metric_difference2}
\end{align}
where the given backlog state observation ${{\bf{\hat Q}}\left( {{t_0}} \right)}$ can be dropped from the expectation expressions because, under $\hat{Policy}$, the values of $t_0-u_1$ and $M_n\left(t_0,t\right)$ are independent of ${{\bf{\hat Q}}\left( {{t_0}} \right)}$ . Define ${T_{\max }} = \mathop {\max }\limits_{n \in \cal N} \left\{ {\mathbb{E}\left\{ {{T_n}} \right\}} \right\}$. Since ${\mathbb{E}\left\{ {{T_n}} \right\}}$ only depends on the topology of the network and is finite (see (\ref{eq_finite_epoch_length}) in Appendix \ref{appendix: necessity_corollary_1}), we can increase the value of $t$ such that $t\geq {T_{\max }}$, which results in that $\mathbb{E}\left\{ {{t_0} - {u_{n,1}}} \right\} \le \mathbb{E}\left\{ {{T_n}} \right\} \le {T_{\max }}\le t$. Combining this result with the fact that $M_n\left(t_0,t\right)\le t$, (\ref{eq_key_metric_difference2}) can be further upper bounded as follows:
\begin{align}
&\ \ \ \ \sum\limits_n {\mathbb{E}\left\{ {\left. {\frac{1}{t}\sum\limits_{i = 1}^{{M_n}\left( {{t_0},t} \right)} {\left[ {{{\hat Z}_n}\left( {i,{\bf{\hat Q}}\left( {{u_{n,i}}} \right)} \right) - {{\hat Z}_n}\left( {i,{\bf{\hat Q}}\left( {{t_0}} \right)} \right)} \right]} } \right|{\bf{\hat Q}}\left( {{t_0}} \right)} \right\}}\le Nt\left( {N + {A_{\max }} + 1} \right)\buildrel \Delta \over = C_1\left(t\right).
\label{eq_key_metric_difference3}
\end{align}

Plug (\ref{eq_key_metric_difference3}) back into (\ref{eq_key_metric_hat3}), and we get
\begin{equation}
\sum\limits_n {\mathbb{E}\left\{ {\left. {\frac{1}{t}\sum\limits_{i = 1}^{{M_n}\left( {{t_0},t} \right)} {{{\hat Z}_n}\left( {i,{\bf{\hat Q}}\left( {{t_0}} \right)} \right)} } \right|{\bf{\hat Q}}\left( {{t_0}} \right)} \right\}}\geq \sum\limits_n {\mathbb{E}\left\{ {\left. {\frac{1}{t}\sum\limits_{i = 1}^{{M_n}\left( {{t_0},t} \right)} {{{\hat Z}_n}\left( {i,{\bf{\hat Q}}\left( {{u_{n,i}}} \right)} \right)} } \right|{\bf{\hat Q}}\left( {{t_0}} \right)} \right\}} -C_1\left(t\right).
\label{eq_key_metric_change_coefficients}
\end{equation}
If denoting $\hat D_2 = {\rm{max}}\left\{\hat D_1, T_{\rm{max}}\right\}$, and plugging the result of (\ref{eq_key_metric_change_coefficients}) into (\ref{eq_key_metric_hat2}) with $t \ge \hat D_2$, we finally get
\begin{equation}
\sum\limits_n {\mathbb{E}\left\{ {\left. {\left. {{{\hat Z}_n}\left( {{\bf{\hat Q}}\left( {{t_0}} \right)} \right)} \right|_{{t_0}}^{{t_0} + t - 1}} \right|{\bf{\hat Q}}\left( {{t_0}} \right)} \right\}} \geq \sum\limits_n {\mathbb{E}\left\{ {\left. {\frac{1}{t}\sum\limits_{i = 1}^{{M_n}\left( {{t_0},t} \right)} {{{\hat Z}_n}\left( {i,{\bf{\hat Q}}\left( {{u_{n,i}}} \right)} \right)} } \right|{\bf{\hat Q}}\left( {{t_0}} \right)} \right\}} -C_1\left(t\right)-\frac{\varepsilon }{8}\sum\limits_{n,c} {Q_n^{\left( c \right)}\left( {{t_0}} \right)}.
\label{eq_key_metric_hat4}
\end{equation}
Up to (\ref{eq_key_metric_hat4}), we lower bound the original key metric $\sum\limits_n {\mathbb{E}\left\{ {\left. {\left. {{{\hat Z}_n}\left( {{\bf{\hat Q}}\left( {{t_0}} \right)} \right)} \right|_{{t_0}}^{{t_0} + t - 1}} \right|{\bf{\hat Q}}\left( {{t_0}} \right)} \right\}}$ under $\hat {Policy}$ by the expression on the right hand side of (\ref{eq_key_metric_hat4}). The purpose of doing this is to transform the analysis focus from the original key metric into the new metric $\sum\limits_n {\mathbb{E}\left\{ {\left. {\frac{1}{t}\sum_{i = 1}^{{M_n}\left( {{t_0},t} \right)} {{{Z}_n}\left( {i,{\bf{\hat Q}}\left( {{u_{n,i}}} \right)} \right)} } \right|{\bf{\hat Q}}\left( {{t_0}} \right)} \right\}} $, which is easier to analyze by the knowledge of renewal processes (see Ref. \cite{Ross_book_stochastic_process_2001}, Chapter 3).

\subsection{Comparison between $\hat{Policy}$ and $\tilde{\tilde{Policy}}$ over a single epoch}
\label{subsec_hat_vs_tildetilde_2}
Before comparing $\hat{Policy}$ with $\tilde{\tilde{Policy}}$ over multiple epochs, the proof in this subsection first compares the two polices over a single epoch.

With the fact that $M_n\left(t_0,t\right)\le t$, define the following indicator function of integer $i = 1,2,3,\cdots,t$:
\begin{equation}
{1_n}\left( i \right) = \left\{ {\begin{array}{*{20}{c}}
{1,\;1 \le i \le {{ M}_n}\left( {{t_0},t} \right) \le t}\\
{0,{\rm{\ \ \ \ }}\;{{ M}_n}\left( {{t_0},t} \right) < i \le t,}
\end{array}} \right.
\label{eq_indicator_function_hat}
\end{equation}
Based on the description of $\hat {Policy}$, firstly, it makes decisions only based on the backlog state observation ${{\bf{\hat Q}}\left( {{u_{n,i}}} \right)}$ for each epoch $i$, and therefore, given the backlog state ${{\bf{\hat Q}}\left( {{u_{n,i}}} \right)}$, the value of ${{Z_n}\left( {i,{\bf{\hat Q}}\left( {{u_{n,i}}} \right)} \right)}$ under $\hat{Policy}$ is independent of ${{\bf{\hat Q}}\left( {{t_0}} \right)}$. Consequently, we have $\forall n \in \cal N$
\begin{equation}
\mathbb{E}\left\{ {\left. {{{\hat Z}_n}\left( {i,{\bf{\hat Q}}\left( {{u_{n,i}}} \right)} \right)} \right|{\bf{\hat Q}}\left( {{u_{n,i}}} \right)},1_n\left(i\right)=1 \right\} = \mathbb{E}\left\{ {\left. {{{\hat Z}_n}\left( {i,{\bf{\hat Q}}\left( {{u_{n,i}}} \right)} \right)} \right|{\bf{\hat Q}}\left( {{u_{n,i}}} \right),{\bf{\hat Q}}\left( {{t_0}} \right)},1_n\left(i\right)=1 \right\}.
\end{equation}
Secondly, according to Lemma \ref{lemma: characterize_metric_DIVBAR-RMIA} in Section \ref{sec: performance_anlaysis_main}, the metric $\mathbb{E}\left\{ {\left. {{Z_n}\left( {i,{\bf{\hat Q}}\left( {{u_{n,i}}} \right)} \right)} \right|{\bf{\hat Q}}\left( {{u_{n,i}}} \right),{1_n}\left( i \right) = 1} \right\}$ is maximized under $\hat{Policy}$ among all policies within the restricted policy set $\cal P$, to which $\tilde{\tilde{Policy}}$ belongs. Thus, with any possible value of ${{\bf{\hat Q}}\left( {{t_0}} \right)}$, we have
\begin{align}
&\ \ \ \ \mathbb{E}\left\{ {\left. {{{\hat Z}_n}\left( {i,{\bf{\hat Q}}\left( {{u_{n,i}}} \right)} \right)} \right|{\bf{\hat Q}}\left( {{u_{n,i}}} \right),{\bf{\hat Q}}\left( {{t_0}} \right)} ,1_n\left(i\right)=1\right\}\nonumber\\
&\ge \mathbb{E}\left\{ {\left. {{{\tilde {\tilde Z}}_n}\left( {i,{\bf{\hat Q}}\left( {{u_{n,i}}} \right)} \right)} \right|{\bf{\hat Q}}\left( {{u_{n,i}}} \right),{\bf{\hat Q}}\left( {{t_0}} \right)} ,1_n\left(i\right)=1\right\}\nonumber\\
&= \mathbb{E}\left\{ {\left. {{{\tilde {\tilde Z}}_n}\left( {i,{\bf{\hat Q}}\left( {{t_0}} \right)} \right)} \right|{\bf{\hat Q}}\left( {{u_{n,i}}} \right),{\bf{\hat Q}}\left( {{t_0}} \right)},1_n\left(i\right)=1 \right\} \nonumber\\
&\ \ \ - \mathbb{E}\left\{ {\left. {{{\tilde {\tilde Z}}_n}\left( {i,{\bf{\hat Q}}\left( {{t_0}} \right)} \right) - {{\tilde {\tilde Z}}_n}\left( {i,{\bf{\hat Q}}\left( {{u_{n,i}}} \right)} \right)} \right|{\bf{\hat Q}}\left( {{u_{n,i}}} \right),{\bf{\hat Q}}\left( {{t_0}} \right)} ,1_n\left(i\right)=1\right\}.
\label{eq_single_epoch_metric}
\end{align}

The difference term in (\ref{eq_single_epoch_metric}) can be written as
\begin{align}
&\ \ \ \ \mathbb{E}\left\{ {\left. {{{\tilde {\tilde Z}}_n}\left( {i,{\bf{\hat Q}}\left( {{t_0}} \right)} \right) - {{\tilde {\tilde Z}}_n}\left( {i,{\bf{\hat Q}}\left( {{u_{n,i}}} \right)} \right)} \right|{\bf{\hat Q}}\left( {{u_{n,i}}} \right),{\bf{\hat Q}}\left( {{t_0}} \right)} ,1_n\left(i\right)=1 \right\}\nonumber\\
&= \mathbb{E}\left\{ {\left. {\sum\limits_{\tau  = {u_{n,i}}}^{{u_{n,i + 1}} - 1} {\sum\limits_c {\sum\limits_{k \in {{\cal K}_n}} {\tilde {\tilde b}_{nk}^{\left( c \right)}\left( \tau  \right)} } } \left[ {\hat Q_n^{\left( c \right)}\left( {{t_0}} \right) - \hat Q_n^{\left( c \right)}\left( {{u_{n,i}}} \right) + \hat Q_k^{\left( c \right)}\left( {{u_{n,i}}} \right) - \hat Q_k^{\left( c \right)}\left( {{t_0}} \right)} \right]} \right|{\bf{\hat Q}}\left( {{u_{n,i}}} \right),{\bf{\hat Q}}\left( {{t_0}} \right)} ,1_n\left(i\right)=1\right\}.
\label{eq_metric_difference3}
\end{align}
Similar to (\ref{eq_key_metric_difference})-(\ref{eq_backlog_difference4}) but with the roles of $n$ and $k$ switched, (\ref{eq_metric_difference3}) is upper bounded as follows: $\forall t\ge \hat D_2$,
\begin{equation}
\mathbb{E}\left\{ {\left. {{{\tilde {\tilde Z}}_n}\left( {i,{\bf{\hat Q}}\left( {{t_0}} \right)} \right) - {{\tilde {\tilde Z}}_n}\left( {i,{\bf{\hat Q}}\left( {{u_{n,i}}} \right)} \right)} \right|{\bf{\hat Q}}\left( {{u_{n,i}}} \right),{\bf{\hat Q}}\left( {{t_0}} \right)} , 1\le i \le M_n\left(t_0,t\right) \right\}\le t\left( N + {A_{\max }}+1 \right);
\label{eq_key_metric_difference4}
\end{equation}
Denoting $C_2\left(t\right)=t\left( {N + {A_{\max }} + 1} \right)$, and plugging (\ref{eq_key_metric_difference4}) into (\ref{eq_single_epoch_metric}), it follows that, $\forall t \ge \hat D_2$,
\begin{equation}
\mathbb{E}\left\{ {\left. {{{\hat Z}_n}\left( {i,{\bf{\hat Q}}\left( {{u_{n,i}}} \right)} \right)} \right|{\bf{\hat Q}}\left( {{u_{n,i}}} \right),{\bf{\hat Q}}\left( {{t_0}} \right)}, 1_n\left(i\right)=1 \right\}\geq \mathbb{E}\left\{ {\left. {{{\tilde {\tilde Z}}_n}\left( {i,{\bf{\hat Q}}\left( {{t_0}} \right)} \right)} \right|{\bf{\hat Q}}\left( {{u_{n,i}}} \right),{\bf{\hat Q}}\left( {{t_0}} \right)}, 1_n\left(i\right)=1 \right\} - C_2\left(t\right).
\label{eq_metric_comparison_single_epoch}
\end{equation}
Then take expectations on both side of (\ref{eq_metric_comparison_single_epoch}) over all possible values of ${{\bf{\hat Q}}\left( {{u_{n,i}}} \right)}$, we get, $\forall t \ge \hat D_2$,
\begin{equation}
\mathbb{E}\left\{ {\left. {{{\hat Z}_n}\left( {i,{\bf{\hat Q}}\left( {{u_{n,i}}} \right)} \right)} \right|{\bf{\hat Q}}\left( {{t_0}} \right)} , 1_n\left(i\right)=1\right\} \ge \mathbb{E}\left\{ {\left. {{{\tilde {\tilde Z}}_n}\left( {i,{\bf{\hat Q}}\left( {{t_0}} \right)} \right)} \right|{\bf{\hat Q}}\left( {{t_0}} \right)} , 1_n\left(i\right)=1\right\} - {C_2}\left( t \right),
\label{eq_metric_comparison_single_epoch2}
\end{equation}
which completes the comparison on the key metrics over a single epoch under $\hat{Policy}$ and $\tilde{\tilde{Policy}}$.

\subsection{Comparison between $\hat{Policy}$ and $\tilde{\tilde{Policy}}$ over $M_n\left(t_0,t\right)$ epochs}
\label{subsec_hat_vs_tildetilde_3}
In the proof of this subsection, the comparison between $\hat{Policy}$ and $\tilde{\tilde{Policy}}$ is extended from a single epoch to $M_n\left(t_0,t\right)$ epochs. For the metric $\sum\limits_n {\mathbb{E}\left\{ {\left. {\frac{1}{t}\sum\nolimits_{i = 1}^{{M_n}\left( {{t_0},t} \right)} {{{\hat Z}_n}\left( {i,{\bf{\hat Q}}\left( {{u_{n,i}}} \right)} \right)} } \right|{\bf{\hat Q}}\left( {{t_0}} \right),{1_n}\left( i \right) = 1} \right\}}$, since $M_n\left(t_0,t\right)$ is random, we cannot directly pull $M_n\left(t_0,t\right)$ out of the expectation. Instead, we need to firstly rewrite the key metric under $\hat{Policy}$ as follows:
\begin{align}
\sum\limits_n {\mathbb{E}\left\{ {\left. {\frac{1}{t}\sum\limits_{i = 1}^{{M_n}\left( {{t_0},t} \right)} {{{\hat Z}_n}\left( {i,{\bf{\hat Q}}\left( {{u_{n,i}}} \right)} \right)} } \right|{\bf{\hat Q}}\left( {{t_0}} \right)} \right\}}  &= \frac{1}{t}\sum\limits_n {\sum\limits_{i = 1}^t  {\mathbb{E}\left\{ {\left. {{{\hat Z}_n}\left( {i,{\bf{\hat Q}}\left( {{u_{n,i}}} \right)} \right){1_n}\left( i \right)} \right|{\bf{\hat Q}}\left( {{t_0}} \right)} \right\}} }.
\label{eq_key_metric_renewal1}
\end{align}
Considering the term ${\mathbb{E}\left\{ {\left. {{{\hat Z}_n}\left( {i,{\bf{\hat Q}}\left( {{u_{n,i}}} \right)} \right){1_n}\left( i \right)} \right|{\bf{\hat Q}}\left( {{t_0}} \right)} \right\}}$, if $M_n\left(t_0,t\right)<i\le t$, then $1_n\left(i\right)=0$, and we have
\begin{equation}
\mathbb{E}\left\{ {\left. {{{\hat Z}_n}\left( {i,{\mathbf{\hat Q}}\left( {{u_{n,i}}} \right)} \right){1_n}\left( i \right)} \right|{\mathbf{\hat Q}}\left( {{t_0}} \right),{1_n}\left( i \right) = 0} \right\} = \mathbb{E}\left\{ {\left. {{{\tilde {\tilde Z}}_n}\left( {i,{\mathbf{\hat Q}}\left( {{t_0}} \right)} \right){1_n}\left( i \right)} \right|{\mathbf{\hat Q}}\left( {{t_0}} \right),{1_n}\left( i \right) = 0} \right\} = 0;
\label{eq_key_metric_renewal2}
\end{equation}
if $1\le i \le M_n\left(t_0,t\right)$, then $1_n\left(i\right)=1$, according to (\ref{eq_metric_comparison_single_epoch2}), we have, $\forall t \ge \hat D_2$,
\begin{equation}
\mathbb{E}\left\{ {\left. {{{\hat Z}_n}\left( {i,{\mathbf{\hat Q}}\left( {{u_{n,i}}} \right)} \right){1_n}\left( i \right)} \right|{\mathbf{\hat Q}}\left( {{t_0}} \right),{1_n}\left( i \right) = 1} \right\} \ge \mathbb{E}\left\{ {\left. {{{\tilde {\tilde Z}}_n}\left( {i,{\mathbf{\hat Q}}\left( {{t_0}} \right)} \right)} {1_n}\left( i \right)\right|{\mathbf{\hat Q}}\left( {{t_0}} \right),{1_n}\left( i \right) = 1} \right\} - {C_2}\left( t \right).
\label{eq_key_metric_renewal3}
\end{equation}
Based on (\ref{eq_key_metric_renewal2}) and (\ref{eq_key_metric_renewal3}), we can conclude that, $\forall t \ge \hat D_2$,
\begin{align}
 \mathbb{E}\left\{ {\left. {{{\hat Z}_n}\left( {i,{\mathbf{\hat Q}}\left( {{u_{n,i}}} \right)} \right){1_n}\left( i \right)} \right|{\mathbf{\hat Q}}\left( {{t_0}} \right)} \right\} &\geq \left[ {\mathbb{E}\left\{ {\left. {{{\tilde {\tilde Z}}_n}\left( {i,{\mathbf{\hat Q}}\left( {{t_0}} \right)} \right){1_n}\left( i \right)} \right|{\mathbf{\hat Q}}\left( {{t_0}} \right),{1_n}\left( i \right) = 0} \right\} - {C_2}\left( t \right)} \right]\Pr \left\{ {{1_n}\left( i \right) = 0} \right\}\nonumber\\
& \ \ \ + \left[ {\mathbb{E}\left\{ {\left. {{{\tilde {\tilde Z}}_n}\left( {i,{\mathbf{\hat Q}}\left( {{t_0}} \right)} \right){1_n}\left( i \right)} \right|{\mathbf{\hat Q}}\left( {{t_0}} \right),{1_n}\left( i \right) = 1} \right\} - {C_2}\left( t \right)} \right]\Pr \left\{ {{1_n}\left( i \right) = 1} \right\}\nonumber\\
&=\mathbb{E}\left\{ {\left. {{{\tilde {\tilde Z}}_n}\left( {i,{\bf{\hat Q}}\left( {{t_0}} \right)} \right){1_n}\left( i \right)} \right|{\bf{\hat Q}}\left( {{t_0}} \right)} \right\} - {C_2}\left( t \right).
\label{eq_key_metric_renewal4}
\end{align}
Plug (\ref{eq_key_metric_renewal4}) into right hand side of (\ref{eq_key_metric_renewal1}), and it follows that, $\forall t \ge \hat D_2$,
\begin{align}
\sum\limits_n {\mathbb{E}\left\{ {\left. {\frac{1}{t}\sum\limits_{i = 1}^{{M_n}\left( {{t_0},t} \right)} {{{\hat Z}_n}\left( {i,{\bf{\hat Q}}\left( {{u_{n,i}}} \right)} \right)} } \right|{\bf{\hat Q}}\left( {{t_0}} \right)} \right\}} {\rm{  }} &\ge \frac{1}{t}\sum\limits_n {\sum\limits_{i = 1}^t {\mathbb{E}\left\{ {\left. {{{\tilde {\tilde Z}}_n}\left( {i,{\bf{\hat Q}}\left( {{t_0}} \right)} \right){1_n}\left( i \right)} \right|{\bf{\hat Q}}\left( {{t_0}} \right)} \right\}} }  - N{C_2}\left( t \right).
\label{eq_key_metric_renewal5}
\end{align}
With the renewal operation, $\left\{{{{\tilde {\tilde Z}}_n}\left( {i,{\bf{\hat Q}}\left( {{t_0}} \right)} \right)}:1\le i\le M_n\left(t_0,t\right)\right\}$ are i.i.d., and the value of $1_n\left(i\right)$ only depends on $T'_{n,t_0}\left(1\right), T_n\left(2\right),\cdots, T_n\left(i-1\right)$, where $T'_{n,t_0}\left(1\right)=u_{n,2}-t_0$. Combining these two aspects, ${{{\tilde {\tilde Z}}_n}\left( {i,{\bf{\hat Q}}\left( {{t_0}} \right)} \right)}$ and ${{1_n}\left( i \right)}$ are independent. Then for the right hand side of (\ref{eq_key_metric_renewal5}), we have
\begin{align}
\mathbb{E}\left\{ {\left. {{{\tilde {\tilde Z}}_n}\left( {i,{\bf{\hat Q}}\left( {{t_0}} \right)} \right){1_n}\left( i \right)} \right|{\bf{\hat Q}}\left( {{t_0}} \right)} \right\} &= \mathbb{E}\left\{ {\left. {{{\tilde {\tilde Z}}_n}\left( {i,{\bf{\hat Q}}\left( {{t_0}} \right)} \right)} \right|{\bf{\hat Q}}\left( {{t_0}} \right)} \right\}\mathbb{E}\left\{ {{1_n}\left( i \right)} \right\}\nonumber\\
&=\tilde {\tilde z}_n\left( {{\bf{\hat Q}}\left( {{t_0}} \right)} \right)\Pr \left( {{M_n}\left( {{t_0},t} \right) \ge i} \right).
\label{eq_key_metric_renewal5.5}
\end{align}
Then plug (\ref{eq_key_metric_renewal5.5}) into (\ref{eq_key_metric_renewal5}), and it follows that
\begin{align}
\sum\limits_n {\mathbb{E}\left\{ {\left. {\frac{1}{t}\sum\limits_{i = 1}^{{M_n}\left( {{u_{n,i}},t} \right)} {{{\hat Z}_n}\left( {i,{\bf{\hat Q}}\left( {{u_{n,i}}} \right)} \right)} } \right|{\bf{\hat Q}}\left( {{t_0}} \right)} \right\}} &\ge \frac{1}{t}\sum\limits_n {\tilde {\tilde z}_n\left( {{\bf{\hat Q}}\left( {{t_0}} \right)} \right)\sum\limits_{i = 1}^t {\Pr \left\{ {{M_n}\left( {{t_0},t} \right) \ge i} \right\}} }- N{C_2}\left( t \right)\nonumber\\
&= \frac{1}{t}\sum\limits_n {\tilde {\tilde z}_n\left( {{\bf{\hat Q}}\left( {{t_0}} \right)} \right)\mathbb{E}\left\{ {{M_n}\left( {{t_0},t} \right)} \right\}}- N{C_2}\left( t \right).
\label{eq_key_metric_renewal7}
\end{align}

Finally, going back to (\ref{eq_key_metric_hat4}) in Subsection \ref{subsec_hat_vs_tildetilde_1} and plugging (\ref{eq_key_metric_renewal7}) in, we have, $\forall t\ge \hat D_2$,
\begin{equation}
\sum\limits_n {\mathbb{E}\left\{ {\left. {\left. {{{\hat Z}_n}\left( {{\bf{\hat Q}}\left( {{t_0}} \right)} \right)} \right|_{{t_0}}^{{t_0} + t - 1}} \right|{\bf{\hat Q}}\left( {{t_0}} \right)} \right\}} \geq \frac{1}{t}\sum\limits_n {\tilde {\tilde z}_n\left( {{\bf{\hat Q}}\left( {{t_0}} \right)} \right)\mathbb{E}\left\{ {{M_n}\left( {{t_0},t} \right)} \right\}}  - \left[ {N{C_2}\left( t \right) + {C_1}\left( t \right) + \frac{\varepsilon }{8}\sum\limits_{n,c} {\hat Q_n^{\left( c \right)}\left( {{t_0}} \right)} } \right],
\label{eq_key_metric_comparison1_1}
\end{equation}
which completes the comparison between the metric values under $\hat{Policy}$ and $\tilde{\tilde{Policy}}$.

\section{Proof of Lemma \ref{lemma: hathat_vs_hatprime}}
\label{appendix: hathat_vs_hatprime}
Under $\hat{\hat{Policy}}$, define $\hat{\hat X}_{nk}^{{\rm{MIA}},\left(c\right)}\left(i\right)$ as the random variable that takes value $1$ if node $k \in {\cal K}_n$ decodes a packet of commodity $c$ being transmitted from node $n$ by the end of epoch $i$, and takes value $0$ otherwise. The superscript MIA on $\hat{\hat X}_{nk}^{{\rm{MIA}},\left(c\right)}\left(i\right)$ is to emphasize that the value of the variable is related to the pre-accumulated partial information respective to the starting timeslot of epoch $i$ for node $n$, and the superscript $\left(c\right)$ indicates that the distribution of $\hat{\hat X}_{nk}^{{\rm{MIA}},\left(c\right)}\left(i\right)$ may vary with commodities. Define ${\hat {\hat \mu}}_n^{\left(c\right)}\left(i\right)$ as the decision variable under $\hat{\hat{Policy}}$ that takes value 1 if node $n$ decides to transmit a commodity $c$ packet in epoch $i$, and takes value $0$ otherwise. With the two definitions, we have the following relation:
\begin{equation}
\hat {\hat b}_{nk}^{\left( c \right)}\left( {{u_{n,i + 1}} - 1} \right) = \hat {\hat b}_{nk}^{\left( c \right)}\left( {{u_{n,i + 1}} - 1} \right)\hat {\hat X}_{nk}^{{\rm{MIA}},\left( c \right)}\left( i \right){\hat {\hat \mu} _n^{\left(c\right)}}\left( i \right).
\end{equation}

Under $\hat{\hat{Policy}}$, each node $n$ forwards each decoded packet to the successful receiving node with the largest positive differential backlog at the end of each epoch. Therefore, with similar arguments as in the proof of Lemma \ref{lemma: characterize_metric_DIVBAR-RMIA} in Section \ref{sec: performance_anlaysis_main}, we get a similar result for $\hat{\hat{Policy}}$ shown as follows:
\begin{align}
&\ \ \ \ \mathbb{E}\left\{ {\left. {{{\hat {\hat Z}}_n}\left( {i,{\bf{\hat {\hat Q}}}\left( {{u_{n,i}}} \right)} \right)} \right|{\bf{\hat {\hat Q}}}\left( {{u_{n,i}}} \right),{\bf{\hat {\hat Q}}}\left( {{t_0}} \right),{1_n}\left( i \right) = 1} \right\}\nonumber\\
&= \sum\limits_c{\mathbb{E}\left\{ {\left. { {\mathop {\max }\limits_{k \in {{\cal K}_n}} \left\{ {\hat {\hat X}_{nk}^{{\rm{MIA}},\left( c \right)}\left( i \right)\hat {\hat W}_{nk}^{\left( c \right)}\left( {{u_{n,i}}} \right)} \right\}}  } \right|{\bf{\hat {\hat Q}}}\left( {{u_{n,i}}} \right),{\bf{\hat {\hat Q}}}\left( {{t_0}} \right),{1_n}\left( i \right) = 1,{{\hat {\hat \mu} }_n^{\left(c\right)}}\left( i \right) = 1} \right\}}\mathbb{E}\left\{ {\left. {\hat {\hat \mu} _n^{\left( c \right)}\left( i \right)} \right|{\bf{\hat {\hat Q}}}\left( {{u_{n,i}}} \right)} \right\},
\label{eq_key_metric_hathat5}
\end{align}
where $\hat {\hat W}_{nk}^{\left( c \right)}\left( \tau  \right) = \max \left\{ {\hat {\hat Q}_n^{\left( c \right)}\left( \tau  \right) - \hat {\hat Q}_k^{\left( c \right)}\left( \tau  \right),0} \right\}$. In (\ref{eq_key_metric_hathat5}), we drop ${{\bf{\hat {\hat Q}}}\left( {{t_0}} \right)}$ and the condition ${{1_n}\left( i \right) = 1}$ from the given condition of the expectation of ${\hat {\hat \mu} _n^{\left( c \right)}\left( i \right)}$, because ${\hat {\hat \mu} _n^{\left( c \right)}\left( i \right)}$ are purely determined by the observation of ${{\bf{\hat {\hat Q}}}\left( {{u_{n,i}}} \right)}$ under $\hat{\hat{Policy}}$.

As defined in the summary of DIVBAR-MIA in Section \ref{sec: algorithm_discription}, $\hat{\hat c}_n\left(i\right)$ represents the commodity chosen by node $n$ to transmit in epoch $i$ under DIVBAR-MIA, given that node $n$ decides to transmit a commodity packet. Moreover, define $\hat{\hat{\mu}}_n\left(i\right)$ as the random variable that takes value 1 if node $n$ decides to transmit a commodity packet in epoch $i$, and takes value $0$ otherwise. Under $\hat{\hat{Policy}}$, since we have $\hat {\hat \mu} _n^{\left( \hat{\hat c}_n\left(i\right) \right)}\left( i \right) = {{\hat {\hat \mu} }_n}\left( i \right)$, it further follows from (\ref{eq_key_metric_hathat5}) that
\begin{align}
&\ \ \ \ \mathbb{E}\left\{ {\left. {{{\hat {\hat Z}}_n}\left( {i,{\bf{\hat {\hat Q}}}\left( {{u_{n,i}}} \right)} \right)} \right|{\bf{\hat {\hat Q}}}\left( {{u_{n,i}}} \right),{\bf{\hat {\hat Q}}}\left( {{t_0}} \right),{1_n}\left( i \right) = 1} \right\}\nonumber\\
&= \mathbb{E}\left\{ {\left. {\mathop {\max }\limits_{k \in {{\cal K}_n}} \left\{ {\hat {\hat X}_{nk}^{{\rm{MIA}},\left( {{{\hat {\hat c}}_n}\left( i \right)} \right)}\left( i \right)\hat {\hat W}_{nk}^{\left( {\hat {\hat c}}_n\left(i\right) \right)}\left( {{u_{n,i}}} \right)} \right\}} \right|{\bf{\hat {\hat Q}}}\left( {{u_{n,i}}} \right),{\bf{\hat {\hat Q}}}\left( {{t_0}} \right),{1_n}\left( i \right) = 1,{{\hat {\hat \mu} }_n}\left( i \right) = 1} \right\}\mathbb{E} \left\{ {\left. {{{\hat {\hat \mu} }_n}\left( i \right)} \right|{\bf{\hat {\hat Q}}}\left( {{u_{n,i}}} \right)} \right\}.
\label{eq_key_metric_hathat6}
\end{align}

Under either $\hat{\hat{Policy}}$ or $\hat{Policy'}$, each node $n$ uses backpressure strategy to choose the commodity to transmit in epoch $i$ based on the backlog state observation ${\bf\hat{\hat{Q}}}\left(u_{n,i}\right)$ without considering the pre-accumulated partial information at the receivers, and therefore, we can get two conclusions: first, the decisions made by node $n$ under the two policies of whether to transmit a commodity packet in epoch $i$ or not are the same, i.e.,
\begin{equation}
\hat \mu {'_n}\left( i \right) = {{\hat {\hat \mu} }_n}\left( i \right);
\end{equation}
second, the chosen commodities to transmit respectively under two policies are also the same, i.e., given that node $n$ decides to transmit a commodity packet, if letting $\hat{c'}_n\left(i\right)$ represent the commodity chosen by node $n$ under $\hat{Policy'}$ to transmit in epoch $i$, then
\begin{equation}
\hat{c'}_n\left(i\right)=\hat{\hat c}_n\left(i\right).
\end{equation}
Therefore, with the similar procedure as in the proof of Lemma \ref{lemma: characterize_metric_DIVBAR-RMIA}, the backpressure metric over a single epoch under $\hat{Policy'}$ can be computed as follows:
\begin{align}
&\ \ \ \ \mathbb{E}\left\{ {\left. {\hat Z{'_n}\left( {i,{\bf{\hat {\hat Q}}}\left( {{u_{n,i}}} \right)} \right)} \right|{\bf{\hat {\hat Q}}}\left( {{u_{n,i}}} \right),{\bf{\hat {\hat Q}}}\left( {{t_0}} \right),{1_n}\left( i \right) = 1} \right\}\nonumber\\
&= \mathbb{E}\left\{ {\left. {\mathop {\max }\limits_{k \in {{\cal K}_n}} \left\{ {X_{nk}^{{\cal P},{\rm{RMIA}}}\left( i \right)\hat {\hat W}_{nk}^{\left( {\hat c'\left( i \right)} \right)}\left( {{u_{n,i}}} \right)} \right\}} \right|{\bf{\hat {\hat Q}}}\left( {{u_{n,i}}} \right),{\bf{\hat {\hat Q}}}\left( {{t_0}} \right),{1_n}\left( i \right) = 1,\hat \mu {'_n}\left( i \right) = 1} \right\}\mathbb{E}\left\{ {\left. {\hat {\mu '}{_n}\left( i \right)} \right|{\bf{\hat {\hat Q}}}\left( {{u_{n,i}}} \right)} \right\}\nonumber\\
&= \mathbb{E}\left\{ {\left. {\mathop {\max }\limits_{k \in {{\cal K}_n}} \left\{ {X_{nk}^{{\cal P},{\rm{RMIA}}}\left( i \right)\hat {\hat W}_{nk}^{\left( {\hat {\hat c}\left( i \right)} \right)}\left( {{u_{n,i}}} \right)} \right\}} \right|{\bf{\hat {\hat Q}}}\left( {{u_{n,i}}} \right),{\bf{\hat {\hat Q}}}\left( {{t_0}} \right),{1_n}\left( i \right) = 1,{{\hat {\hat \mu} }_n}\left( i \right) = 1} \right\}\mathbb{E}\left\{ {\left. {{{\hat {\hat \mu} }_n}\left( i \right)} \right|{\bf{\hat {\hat Q}}}\left( {{u_{n,i}}} \right)} \right\},
\label{eq_key_metric_hatprime1}
\end{align}
where $X_{nk}^{{\cal P},{\rm{RMIA}}}\left(i\right)$ is the indicator variable of the first successful reception over link $\left(n,k\right)$ in epoch $i$ under a policy with RMIA in the restricted policy set $\cal P$ and has been defined in the proof of Lemma \ref{lemma: characterize_metric_DIVBAR-RMIA} in Section \ref{sec: performance_anlaysis_main}. Here we use $X_{nk}^{{\cal P},{\rm{RMIA}}}\left( i \right)$ because $\hat{Policy'}$ belongs to $\cal P$.

Comparing (\ref{eq_key_metric_hatprime1}) to (\ref{eq_key_metric_hathat6}), the key part is to compare $X_{nk}^{{\cal P},{\rm{RMIA}}}\left(i\right)$ and $\hat {\hat X}_{nk}^{{\rm{MIA}},\hat{\hat c}_n\left(i\right)}\left(i\right)$. Note that the channel realizations in epoch $i$ are the same for $\hat{\hat{Policy}}$ and $\hat{Policy'}$, but the decoding of the packet in the receiving node $k \in {\cal K}_n$ under $\hat{\hat{Policy}}$ may take advantage of the pre-accumulated partial information respective to timeslot $u_{n,i}$, while the receiving node can not take this advantage under $\hat{Policy'}$. Therefore, if $X_{nk}^{{\cal P},{\rm{RMIA}}}\left(i\right)$ takes value 1, $\hat {\hat X}_{nk}^{{\rm{MIA}},\hat{\hat c}_n\left(i\right)}\left(i\right)$ must take value 1, while the reverse statement may not be true, i.e.,
\begin{equation}
\hat {\hat X}_{nk}^{{\rm{MIA}},\hat{\hat c}_n\left(i\right)}\left(i\right) \ge X_{nk}^{{\cal P},{\rm{RMIA}}}\left( i \right).
\end{equation}
Going back to the comparison between (\ref{eq_key_metric_hatprime1}) with (\ref{eq_key_metric_hathat6}), it follows that
\begin{equation}
\mathbb{E}\left\{ {\left. {{{\hat {\hat Z}}_n}\left( {i,{\bf{\hat {\hat Q}}}\left( {{u_{n,i}}} \right)} \right)} \right|{\bf{\hat {\hat Q}}}\left( {{u_{n,i}}} \right),{\bf{\hat {\hat Q}}}\left( {{t_0}} \right),{1_n}\left( i \right) = 1} \right\}\ge \mathbb{E}\left\{ {\left. {\hat Z{'_n}\left( {i,{\bf{\hat {\hat Q}}}\left( {{u_{n,i}}} \right)} \right)} \right|{\bf{\hat {\hat Q}}}\left( {{u_{n,i}}} \right),{\bf{\hat {\hat Q}}}\left( {{t_0}} \right),{1_n}\left( i \right) = 1} \right\}.
\label{eq_key_metric_comparison_hathat_vs_hatprime_1}
\end{equation}

Finally, taking expectations over all possible values of ${\bf{\hat{\hat{Q}}}}\left(u_{n,i}\right)$ on both sides of (\ref{eq_key_metric_comparison_hathat_vs_hatprime_1}), it follows that
\begin{equation}
\mathbb{E}\left\{ {\left. {{{\hat {\hat Z}}_n}\left( {i,{\bf{\hat {\hat Q}}}\left( {{u_{n,i}}} \right)} \right)} \right|{\bf{\hat {\hat Q}}}\left( {{t_0}} \right),{1_n}\left( i \right) = 1} \right\}\ge \mathbb{E}\left\{ {\left. {\hat Z{'_n}\left( {i,{\bf{\hat {\hat Q}}}\left( {{u_{n,i}}} \right)} \right)} \right|{\bf{\hat {\hat Q}}}\left( {{t_0}} \right),{1_n}\left( i \right) = 1} \right\}.
\label{eq_key_metric_comparison_hathat_vs_hatprime_2}
\end{equation}

\bibliographystyle{IEEEtran}
\bibliography{Reference}

\end{document}